\pdfoutput=1
\documentclass[11pt,oneside,british,reqno]{amsart}
\usepackage{lmodern}
\usepackage{avant}
\usepackage{courier}

\usepackage[T1]{fontenc}
\usepackage[latin9]{inputenc}
\usepackage[a4paper,left=3cm, right=3cm,top=3cm, bottom=3cm]{geometry}
\usepackage{color}
\usepackage{stmaryrd}
\usepackage{babel}
\usepackage{float,subcaption,subfloat}
\usepackage{units}
\usepackage{amsbsy}
\usepackage{amstext}
\usepackage{graphicx}
\usepackage{amsthm}
\usepackage{amssymb}
\usepackage{setspace}
\usepackage[numbers]{natbib}
\usepackage[unicode=true,pdfusetitle,
 bookmarks=true,bookmarksnumbered=false,bookmarksopen=false,
 breaklinks=false,pdfborder={0 0 0},pdfborderstyle={},backref=false,colorlinks=true]
 {hyperref}
\hypersetup{
 pdfborderstyle={},linkcolor=blue,urlcolor=blue,citecolor=blue}
\usepackage{float}
\usepackage{bbold}
\usepackage[titletoc,title]{appendix}
\allowdisplaybreaks

\makeatletter
\numberwithin{equation}{section}
\numberwithin{figure}{section}
\numberwithin{table}{section}
\theoremstyle{plain}
 \newtheorem{thm}{\protect\theoremname}
 \newtheorem{cor}[thm]{\protect\corollaryname}
 \newtheorem{prop}[thm]{\protect\propositionname}
 \newtheorem{lem}[thm]{\protect\lemmaname}
\theoremstyle{definition}
 \newtheorem{rem}[thm]{\protect\remarkname}
 \newtheorem{example}[thm]{\protect\examplename}
 \newtheorem{defn}[thm]{\protect\definitionname}

\newcommand{\idd}{\mathbb{1}}

\numberwithin{equation}{section}

\numberwithin{equation}{section}
\numberwithin{figure}{section}
\numberwithin{table}{section}

\renewenvironment{proof}[1][\proofname]{\par
  \pushQED{\qed}%
  \normalfont \topsep6\p@\@plus6\p@\relax
  \list{}{%
    \settowidth{\leftmargin}{\itshape\proofname:\hskip\labelsep}%
    \setlength{\labelwidth}{0pt}%
    \setlength{\itemindent}{-\leftmargin}%
  }%
  \item[\hskip\labelsep\itshape#1\@addpunct{:}]\ignorespaces
}{%
  \popQED\endlist\@endpefalse
}

\makeatother

  \providecommand{\corollaryname}{Corollary}
  \providecommand{\examplename}{Example}
  \providecommand{\lemmaname}{Lemma}
  \providecommand{\propositionname}{Proposition}
  \providecommand{\remarkname}{Remark}
  \providecommand{\theoremname}{Theorem}
  \providecommand{\definitionname}{Definition}	

\begin{document}

\title[Signatures of partition functions and their complexity reduction]{Signatures of partition functions and their complexity reduction
through the KP II equation}

\maketitle

\begin{center}
\author{Mario Angelelli}\\
Department of Mathematics and Physics 
\par
``Ennio De Giorgi'', University of Salento and sezione INFN, 
\par
Lecce 73100, Italy 
\par\end{center}

\begin{abstract}
A statistical amoeba arises from a real-valued partition function
when the positivity condition for pre-exponential terms is relaxed,
and families of signatures are taken into account. This notion lets
us explore special types of constraints when we focus on those signatures
that preserve particular properties. Specifically, we look at sums
of determinantal type, and main attention is paid to a distinguished
class of soliton solutions of the Kadomtsev-Petviashvili (KP) II equation.
A characterization of the signatures preserving the determinantal
form, as well as the signatures compatible with the KP II equation,
is provided: both of them are reduced to choices of signs for columns
and rows of a coefficient matrix, and they satisfy the whole KP hierarchy.
Interpretations in term of information-theoretic properties, geometric
characteristics, and the relation with tropical limits are discussed. 
\end{abstract}

\vspace{2pc}


\section{\label{sec: Introduction} Introduction}

The concept of partition function encodes in a single object the statistical
data compatible with the physical constraints of a system, e.g., conservation
laws. Hence, it explicitly connects its probabilistic and physical
characteristics. This is a fundamental principle in the investigation
of composite systems \cite{LL1980,Feynman1982} and is now applied
in many different branches of sciences \cite{Kleinert2009}. The
basic form of the partition function involves a sum of exponential
terms 
\begin{equation}
\mathcal{Z}:=\sum_{\alpha\in\mathbb{I}}g_{\alpha}\cdot e^{-\frac{\varepsilon_{\alpha}}{k_{B}T}}\label{eq: standard partition function}
\end{equation}
where $k_{B}$ is Boltzmann's constant, $T$ is the temperature, $\varepsilon_{\alpha}$
are the energy levels labelled by an indexing set $\mathbb{I}$ and
$g_{\alpha}$ are the associated degenerations.

In practice, the implementation of this approach is limited by the
intrinsic complexity of the model, since there are few cases where
the partition function can be evaluated exactly, e.g., in a closed
form \cite{Baxter1982}. Indeed, many phenomena in complex systems
cannot be reduced to their individual components, thus (\ref{eq: standard partition function}) involves
collections of objects and, consequently, the complexity of the calculation grows exponentially with the size of the system.

The occurrence of exact formulas for the partition function implies
relations between the characterizing quantities of the system, e.g.,
correlation functions that are generated by the partition function
through derivation. This gives rise to a family of differential equations
whose compatibility follows from the existence of the partition function.
On the other hand, one can start from a family of differential equations
and look for their compatibility. This naturally leads to the investigation
on connections between the partition function formalism and the concept
of integrability, especially integrable hierarchies. In such a context,
the role of the partition function is similar to the notion of $\tau$-function,
which provides one with a unifying framework for hierarchies of nonlinear
partial differential equations and their remarkable behaviours, e.g.,
symmetries, infinite family of commuting flows, soliton solutions
\cite{Hirota1980,Nimmo1983}.

Explicit links between the partition function formalism and integrable
systems have already been identified and used to provide fundamental
techniques in modern theoretical physics \cite{Baxter1982}. Remarkably,
certain solutions of integrable PDEs can be interpreted as potentials,
such as the Witten-Dijkgraaf-Verlinde-Verlinde (WDVV) equations in
topological quantum field theories. A solution for the WDVV system
defines a free energy for this theory, namely, a function which generates
correlators by means of derivation \cite{Dubrovin1996}.

Recent studies have been devoted to the combinatorics underlying these
kinds of structures. One of most fruitful examples is the Kadomtsev-Petviashvili
II (KP II) equation 
\begin{equation}
\partial_{x}\left(-4\cdot\partial_{t}u+\partial_{xxx}u+6u\cdot\partial_{x}u\right)=-3\partial_{yy}u\label{eq: KP II equation}
\end{equation}
where $\partial_{x}u$ denotes the partial derivative of $u\equiv u(x,y,t)$
with respect to $x$. The KP II equation is one of most important
$(2+1)$-integrable PDEs and is considered a universal model of two-dimensional
integrable evolutionary equations. The combinatorial structures arise,
for example, taking into account particular classes of solutions of
the KP II equation, such as soliton solutions of determinantal (Wronskian)
type \cite{Nimmo1983}. These solutions are parametrized by points
in the real Grassmannian space and their regularity and algebraic
features have been extensively studied in the last years \cite{BC2006,CK2008,DM-H2011,KodamaWilliams2013,Abenda2017}.

The aim of this work is to explore these requirements starting from
a statistical perspective: we choose a class of combinatorial configurations
related to sums of exponentials that generalise (\ref{eq: standard partition function}).
Then, we investigate the information content encoded in constraints
through the reduction of the class of allowed configurations. This
purpose will be realised extending the method of statistical amoebas
introduced in \cite{AK2016b}, i.e., considering a family of discrete
``deformations'' of the partition function and focusing on those
that are compatible with some given constraints. A statistical amoeba
is obtained relaxing the assumption of positive degenerations $g_{\alpha}$
in (\ref{eq: standard partition function}). Negative degenerations
of energy levels can be related to an imaginary part of energies $\varepsilon_{\alpha}$
and, hence, used to describe metastable states \cite{Langer1969,Newman1980}.
This approach has proved useful in statistical physics, where the
locus of zeros of complex-valued partition functions is employed in
the analysis of phase transitions \cite{LeeYang1952a,LeeYang1952b,Langer1969,Newman1980,Biskup2004,Wei2014}.
In our context, the partition function is real-valued, but indefinite
signs for degenerations $g_{\alpha}$ open up the way to the study
of stability ($\mathcal{Z}>0$), instability $(\mathcal{Z}<0$) and
phase transitions ($\mathcal{Z}=0$).

In general, the requirement of compatibility of the choices of signs
for $g_{\alpha}$ with a given constraint affects both the number
and the form of allowed configurations. We will focus on determinantal
relations, where the complexity of calculations of determinants is
polynomial via Gaussian reduction (while other immanents, in general,
have exponential complexity \cite{Valiant1979}), and integrability,
which reduces the complexity through the occurrence of conserved quantities.
We stress this point referring to the construction in \cite{AK2016b},
where all the combinations of signs for the $N$ exponential terms
in the partition function are allowed, as a \emph{free} statistical
amoeba. On the other hand, restrictions on the allowed combinations
of signs give rise to a \emph{constrained} statistical amoeba.

A strong relation between determinantal relations and integrability
has been established by Sato (see, e.g., \cite{Sato1989} and further
developments in \cite{Miwa2000}), where an expansion in terms of
Schur functions is found to be a $\tau$-function for the KP hierarchy
if and only if its coefficients satisfy the \textit{Grassmann-Pl\"{u}cker
relations} \cite{GKZ1994,Hogben2006}. This has led to intensive
research on equivalent reformulations of the KP hierarchy such as
these involving the infinite Grassmannian \cite{Sato1989,Miwa2000}.
Whilst sharing some concepts with these issues, the present investigation
is devoted to families of functions generated by a special class of
\textit{solutions} of the KP equation. It should be remarked that
we concentrate only on the first equation (\ref{eq: KP II equation})
in the hierarchy for what concerns constraints reducing the complexity,
and that we work in finite dimensionality. Moreover, the construction
of statistical amoebas does not rely only on the global (determinantal)
form of the initial function, but also on the \emph{individual} terms
in its expansion, because the generation of other functions goes through
changes of signs for a given expression of the type (\ref{eq: standard partition function}).
The consistency check with constraints and the construction of allowed
configurations depend on relations between non-vanishing terms.

When we start from solitons in Wronskian form, the constraints come
either from the determinantal relations or from the fulfilment of
the KP II equation. In fact, we will see that both these assumptions
reduce the allowed combinations of signs to the same family, i.e.,
the set of transformations that can be obtained by flipping the signs of some rows and/or columns of a coefficient matrix $\mathbf{A}$:
this is showed in Theorem \ref{thm: compatible choices of signs, solution KP, including vanishing}.
The complexity reduction takes place because the information about
the chosen signature for the terms in the exponential sum, which are
labelled by \emph{subsets} of $\{1,\dots,n\}$, is stored in the \emph{elements}
of $\{0,1,\dots,n\}$: in general, this presentation is non-unique
due to a certain equivalence relation (Proposition \ref{prop: number row/column configurations }).
The connection with other equations in the hierarchy follows from
the special form of the transformations in this family (Theorem 
\ref{thm: connections KP eq, KP hierarchy, and determinantal form}).

The main scope of this work is the formal derivation of the above-mentioned
properties. However, it is important to highlight potential applications
of this approach: for instance, choosing a coefficient matrix and
constructing the corresponding soliton solution, one can encode a
sequence of ``bits'' in the signs of the pre-exponential terms,
and recover a special family of sequences checking that the KP II
equation is satisfied. We will quantify the information content following
from the check of the KP II equation through the Kullback-Leibler
divergence between two probability distributions on the strings of
signatures. These results also have geometric implications, and we
will briefly discuss the links with oblique projections and their
statistical relevance. Likewise, we will point out the connection
with the tropical limit in statistical physics as introduced in \cite{AK2015}
and developed in \cite{Angelelli2017}. These issues could be of
interest for a better understanding of the relation between statistical
physics, complex systems and learning methods \cite{Behrens1994,Johansson2006,Ma2014}.
We will outline some of these connections and postpone their detailed
study to a separate article.

The paper is organized as follows: in Section \ref{sec: Background}
we fix the notation and summarize the basic notions that provide a
starting point for our investigation. In Section \ref{sec: Basic statistical mappings}
we introduce the link between soliton solutions of the KP hierarchy
and statistical amoebas, generalizing the latter to higher dimensions.
The occurrence of a particular structure to be preserved in the construction
of statistical amoebas is the focus of the next sections: in Section
\ref{sec: Constrained statistical amoebas_ solitons} we prove that
the compatibility of a choice of signs for the coefficients of exponential
terms with the KP II equation implies that it is induced by rows and
columns sign flips for a coefficient matrix. In Section \ref{sec: Number of distinct configurations}
we consider the number of distinct configurations that can be obtained
in this way. The relation with the strata of the corresponding free
amoeba \cite{AK2016b} is considered in Section \ref{sec: Levels of constrained amoebas}. Two applications of the present framework are addressed in Section \ref{sec: Applications}: the information content in the KP II constraint is studied via the Kullback-Leibler divergence (Subsection \ref{subsec: Application to complexity}), while an intersection property is  discussed in geometric terms (Subsection \ref{subsec: Intersection property and its geometric interpretation}). Finally we draw conclusions and address some issues of potential interest for future investigations in Section \ref{sec: Conclusion and future perspectives}.

\section{\label{sec: Background} Preliminaries}

We briefly summarise the basic state of the art regarding Wronskian
soliton solutions of the KP II equation and the statistical amoeba
formalism which will be used in the rest of the paper. Before that,
it is worth introducing some notations to enhance clearness.

\subsection{\label{subsec: Notation} Notation }

In the following, we will denote by $\mathcal{P}[n]$ the power set
of $[n]:=\{1,\dots,n\}$, by $\mathcal{P}_{k}[n]$ the collection
of subsets of $[n]$ with $k$ elements and by $\mathcal{I}\Delta\mathcal{J}:=\left(\mathcal{I}\setminus\mathcal{J}\right)\cup\left(\mathcal{J}\setminus\mathcal{I}\right)$
the symmetric difference of $\mathcal{I},\mathcal{J}\in\mathcal{P}[n]$.
We will also use the notation $\{\pm1\}:=\{+1,-1\}$, 
\begin{equation}
\mathcal{I}_{\beta}^{\alpha}:=\mathcal{I}\backslash\{\alpha\}\cup\{\beta\}\quad\alpha\in\mathcal{I},\beta\notin\mathcal{I}\label{eq: k-subset single exchange}
\end{equation}
and, similarly, $\mathcal{I}_{\beta}:=\mathcal{I}\cup\{\beta\}$,
$\mathcal{I}^{\alpha}:=\mathcal{I}\backslash\{\alpha\}$, $\mathcal{I}^{\alpha_{1}\alpha_{2}}:=\mathcal{I}\backslash\{\alpha_{1},\alpha_{2}\}$,
\textit{etc}. The expression $\mathcal{I}_{\beta}^{\alpha}$ implicitly assumes that $\alpha\in\mathcal{I}$, and $\beta\notin\mathcal{I}$, unless $\alpha=\beta$, in which case we have $\mathcal{I}_{\alpha}^{\alpha}:=\mathcal{I}$.

The symbol $\Delta_{\mathbf{A}}(\mathcal{I})$ (respectively $\Delta_{\mathbf{K}}(\mathcal{I})$)
is the maximal minor of $\mathbf{A}\in\mathbb{R}^{k\times n}$ (respectively,
$\mathbf{K}$) whose columns (respectively, rows) are indexed by $\mathcal{I}\in\mathcal{P}_{k}[n]$.
We will occasionally use $\Delta(\mathbf{A};\mathcal{I})$ instead
of $\Delta_{\mathbf{A}}(\mathcal{I})$ for the sake of clearness.
When a permutation $\pi$ of $[n]$ is involved, an additional sign
comes from the parity of the number of inversions induced by $\pi$
on $\mathcal{H}$. Particular attention will be paid to the set of
pivot columns $\mathcal{V}:=\{\nu_{1},\dots,\nu_{k}\}$, which is
the least element of $\mathcal{P}_{k}[n]$ (with respect to the lexicographical
order) associated with a non-vanishing minor $\Delta_{\mathbf{A}}(\mathcal{V})$.

The $(n\times k)$-Vandermonde
matrix relative to real parameters $\kappa_{1},\dots,\kappa_{n}$
is $\mathbf{K}:=\left(\kappa_{\alpha}^{i-1}\right)_{\alpha\in[n]}^{i\in[k]}$. The determinant of a general (not necessarily maximal) minor of $\mathbf{K}$ is 
\begin{equation}
\mathrm{VdM}(\boldsymbol{\ensuremath{\kappa}};\mathcal{H}):=\det\left(\mathbf{K}|_{\mathcal{H}\times[k]}\right)=\prod_{\alpha<\beta}^{\alpha,\beta\in\mathcal{H}}(\kappa_{\beta}-\kappa_{\alpha}),\quad\mathcal{H}\subseteq[n].\label{eq: Vandermonde determinant}
\end{equation}
So $\mathbf{K}$ has maximal rank in cases of pairwise distinct soliton
parameters $\kappa_{1},\dots,\kappa_{n}$, as we will assume.
We will look at exponential sums whose terms also depend on variables
$\boldsymbol{x}\in\mathbb{R}^{d}$, $d\geq3$, where $x_{1}=x$, $x_{2}=y$,
and $x_{3}=t$.

The real Grassmannian $\mathcal{GR}_{k,n}(\mathbb{R})$ is the space
of $k$-dimensional linear subspaces of a $n$-dimensional vector
space over $\mathbb{R}$. It can be presented as the quotient of the
space of $k$-dimensional frames in $\mathbb{R}^{n}$ by the left
action of $GL_{k}(\mathbb{R})$. If $\mathbb{R}_{\star}^{k\times n}$
denotes the set of real ($k\times n$)-matrices of maximal rank $\min\{k,n\}=k$,
one gets 
\begin{equation}
\mathcal{GR}_{k,n}(\mathbb{R})\cong\nicefrac{\mathbb{R}_{\star}^{k\times n}}{GL_{k}(\mathbb{R})}\label{eq: Grassmannian as quotient}
\end{equation}
where $GL_{k}(\mathbb{R})$ acts as left multiplication. Greek indices
$\alpha,\beta\in[n]$ will often represent columns, while Latin indices
$i,j\in[k]$ will be used to label rows. Additional notation will
be introduced in specific paragraphs.

\subsection{\label{subsec: Solitons, Hirota method and Wronskian solutions}
KP II equation and Wronskian solutions}

Determinantal solitons define a distinguished class of solutions that
can be derived from Hirota's direct method \cite{Hirota1980,Nimmo1983} introducing derivatives ${\displaystyle D_{x}}$ acting on pairs of functions: 
\begin{equation}
{\displaystyle D_{x}(f\cdot g):=(\partial_{x}f)\cdot g-f\cdot(\partial_{x}g)=(\partial_{x_{1}}-\partial_{x_{2}})f(x_{1})g(x_{2})|_{x_{1}=x_{2}}.}\label{eq: Hirota derivatives}
\end{equation}
So one can rewrite the KP II equation (\ref{eq: KP II equation})
in the following homogeneous bilinear form 
\begin{equation}
\mathrm{D_{KP}}(\tau,\tau):=(D_{x}^{4}-4\cdot D_{x}D_{t}+3\cdot D_{y}^{2})\tau\cdot\tau=0.\label{eq: bilinear KP}
\end{equation}
By the same token, one can rewrite the other equations of the KP hierarchy
in bilinear form \cite{Hirota1980,Miwa2000}. The tau-function $\tau(\boldsymbol{x})$
of the KP equation is related to $u$ via 
\begin{equation}
u(\boldsymbol{x}):=2\cdot\frac{\partial^{2}}{\partial x_{1}^{2}}\ln\tau(\boldsymbol{x}).\label{eq: solution from tau}
\end{equation}
It has been shown (see, e.g., \cite{Grammaticos1994}) that the Hirota
derivative (\ref{eq: Hirota derivatives}) can be characterized as
a derivative operator with gauge invariance under the simultaneous
action $f\mapsto e^{kx}f$ and $g\mapsto e^{kx}g$, i.e., 
\begin{equation}
\mathrm{D_{KP}}(e^{-\vartheta(\boldsymbol{x})}\cdot f(\boldsymbol{x}),e^{-\vartheta(\boldsymbol{x})}\cdot g(\boldsymbol{x}))=e^{-2\vartheta(\boldsymbol{x})}\cdot\mathrm{D_{KP}}(f(\boldsymbol{x}),g(\boldsymbol{x}))\label{eq: gauge invariance Hirota derivatives}
\end{equation}
for any linear function $\vartheta(\boldsymbol{x})$ of $\boldsymbol{x}$.
This is manifest in the antisymmetric form of $D_{x}$ in (\ref{eq: Hirota derivatives})
and is also reflected in the expression of a special class of solutions
\cite{Nimmo1983} 
\begin{equation}
\mathrm{Wr}\left(f_{1},\dots,f_{n}\right)=\det\left(\partial_{x}^{(i-1)}f_{m}\right)_{i,m\in[k]}\label{eq: Wronskian}
\end{equation}
where $f_{\alpha}$, $\alpha\in[n]$, are independent solutions for
the following system of partial differential equations 
\begin{equation}
\frac{\partial}{\partial x_{r}}f=\frac{\partial^{r}}{\partial x_{1}^{r}}f,\quad r\in[d].\label{eq: linear equations for KP}
\end{equation}
One can take a certain number, say $M$, of solutions of (\ref{eq: linear equations for KP}) in the
form 
\begin{equation}
E_{\alpha}(\boldsymbol{x})=\exp\varphi_{\alpha}(\boldsymbol{x}):=\exp\left(\sum_{r=1}^{d}\kappa_{\alpha}^{r}x_{r}\right)\label{eq: exponential dispersion relation}
\end{equation}
with real parameters $\kappa_{\alpha}$. A particular choice of solutions
of (\ref{eq: linear equations for KP}) comes from linear combinations
of these exponentials, with coefficients given by the entries of a
matrix $\mathbf{A}$: $f_{i}(\boldsymbol{x})={\displaystyle \sum_{\alpha=1}^{n}}A_{i\alpha}E_{\alpha}(\boldsymbol{x})$,
that is $\overrightarrow{f}=\mathbf{A}\cdot\overrightarrow{E}$. One
has 
\begin{equation}
\hspace*{-1.5cm}\left(\begin{array}{cccc}
E_{1} & \partial_{x}E_{1} & \dots & \partial_{x}^{(k-1)}E_{1}\\
\vdots & \vdots & \ddots & \vdots\\
E_{n} & \partial_{x}E_{n} & \dots & \partial_{x}^{(k-1)}E_{n}
\end{array}\right)=\left(\begin{array}{cccc}
E_{1} & \kappa_{1}E_{1} & \dots & \kappa_{1}^{k-1}E_{1}\\
\vdots & \vdots & \ddots & \vdots\\
E_{n} & \kappa_{n}E_{n} & \dots & \kappa_{n}^{k-1}E_{n}
\end{array}\right)=\mathbf{\ensuremath{\Theta}}\cdot\mathbf{K}\label{eq: exponentials to Vandermonde}
\end{equation}
where $\mathbf{\ensuremath{\Theta}}:=\mathrm{diag}\left(E_{1},E_{2},\dots,E_{n}\right)$.
Finally, the resulting soliton solution is equivalently expressed
using the Cauchy-Binet expansion \cite{Hogben2006} as 
\begin{eqnarray}
\tau(\boldsymbol{x}) & := & \det(\mathbf{A}\cdot\mathbf{\ensuremath{\Theta}}(\boldsymbol{x})\cdot\mathbf{K})\label{eq: Cauchy-Binet formula, solitons, a}\\
 & = & \sum_{\mathcal{I}\in\mathcal{P}_{k}[n]}\Delta_{\mathbf{A}}(\mathcal{I})\cdot\Delta_{\mathbf{K}}(\mathcal{I})\cdot e^{\sum_{\alpha\in\mathcal{I}}\varphi_{\alpha}(\boldsymbol{x})}.\label{eq: Cauchy-Binet formula, solitons}
\end{eqnarray}
It should be remarked that the left action of $GL_{k}(\mathbb{R})$
on $\mathbb{R}_{\star}^{k\times n}$, i.e., the multiplication of $\mathbf{A}$
by a full-rank $k\times k$ real matrix, induces the multiplication
of $\tau$ by a constant and, hence, leaves solutions (\ref{eq: solution from tau})
invariant. Such an action is equivalent to row operations on $\mathbf{A}$,
thus soliton solutions are parametrized by points of the real Grassmannian
$\mathcal{GR}_{k,n}(\mathbb{R})$ \cite{Hogben2006} rather than
points on $\mathbb{R}_{\star}^{k\times n}$ (see, e.g., \cite{BC2006,DM-H2011,KodamaWilliams2013}).

The solution (\ref{eq: solution from tau}) is regular at $\tau(\boldsymbol{x})>0$.
If $\tau(\boldsymbol{x})<0$, then $\ln\tau(\boldsymbol{x})$ is multivalued,
but its imaginary part does not depend on $\boldsymbol{x}$ and disappears
after derivation in (\ref{eq: solution from tau}). Thus, possible
singularities of soliton solutions are related to the points where
$\tau$ vanishes. If the order $\kappa_{1}<\dots<\kappa_{n}$ for
soliton parameters is fixed, then the locus of zeros of $\tau$ is
not empty if there exist maximal minors with opposite sign, i.e., if
$\mathbf{A}$ parametrises a point outside the totally non-negative
part of the Grassmannian \cite{KodamaWilliams2013}.

\subsection{\label{subsec: Statistical amoebas } Statistical amoebas }

The analysis of roots of the partition function is a fundamental
technique in the study of stability, metastability and phase transitions in composite
systems \cite{LeeYang1952a,LeeYang1952b,Newman1980,Wei2014}. Finite
sums of exponentials of the type (\ref{eq: standard partition function})
are positive for real values of energies $E_{n}$ and temperature
$T$ and positive degeneracies $g_{n}$. So the zeros of the partition
function define a \emph{singular locus} in the complex
domain and, in many cases, they approach the real line when the number
of terms involved in (\ref{eq: standard partition function}) becomes
large (thermodynamic limit).

If one restricts to the real domain, the partition function can vanish
if not all the degeneracies $g_{n}$ have the same sign. This corresponds
to real partition function of indefinite signatures and relates to the concept of negative probabilities \cite{Wigner1932,Dirac1942,Feynman1987,Burgin2010}. This proposal has been developed in \cite{AK2016b} for partition functions of the type 
\begin{equation}
\mathcal{Z}:=\sum_{\alpha=1}^{N}g_{\alpha}\cdot e^{f_{\alpha}(\boldsymbol{x})}\label{eq: standard partition function, exponential sums}
\end{equation}
where the functions $f_{\alpha}(\boldsymbol{x})$ represent ``micro-free
energies'' that depend on certain parameters $\boldsymbol{x}\in\mathbb{R}^{d}$
(e.g., temperature, external magnetic fields, \textit{etc.}). When one
fixes $s\leq\frac{N}{2}$, the \emph{$s$-stratum} of the statistical
amoeba consists of the zero loci of the functions produced by any
possible combination of $s$ signs among the $N$ terms in (\ref{eq: standard partition function, exponential sums}),
i.e., $\boldsymbol{g}\in\{\pm1\}^{N}$ with $\#\left(\boldsymbol{g}^{-1}\left(\{-1\}\right)\right)=s$.
In great generality, that is under the only assumption of polynomial
functions $f_{\alpha}(\boldsymbol{x})$, $\alpha\in[N]$, a pattern can
be found in the study of the singular locus: the $s$-stratum is confined
in a region $\mathbb{R}^{d}\setminus\mathcal{D}_{s-}$ of the space
of parameters $\mathbb{R}^{d}$, and the \emph{instability domains} $\mathcal{D}_{s-}$ obey the following inclusion property 
\begin{equation}
\mathcal{D}_{s-}\subseteq\mathcal{D}_{\hat{s}-},\quad1\leq s<\hat{s}<\frac{N}{2}.\label{eq: inclusion property ZCD}
\end{equation}
The restriction $s<\frac{N}{2}$ also avoids the redundancy given
by the simultaneous reversal of all the signs. This is equivalently expressed
via the involution $\boldsymbol{g}^{-1}\left(\{-1\}\right)\mapsto[n]\setminus\boldsymbol{g}^{-1}\left(\{-1\}\right)$, which preserves the singular locus and exchanges the role of
equilibrium ($\mathcal{Z}>0$) and non-equilibrium ($\mathcal{Z}<0$)
regions, since $\mathcal{Z}(-\boldsymbol{g})=-\mathcal{Z}(\boldsymbol{g})$. Thus, the strata associated with $s>\frac{N}{2}$ are said to generate the statistical \emph{antiamoeba}.

When all the ${{N} \choose {s}}$ combinations of $s$
negative coefficients $g_{\alpha}$ are taken into account, the set $\mathcal{D}_{s-}$ coincides with the locus of points where the maximal number of negative partition functions (\ref{eq: standard partition function, exponential sums}) is obtained. This maximum
is the same for all the systems with polynomial $f_{\alpha}$ and
equals ${{N-1} \choose {s-1}}$. Furthermore, the polynomial assumption
is also suitable for the study of the tropical limit \cite{AK2015,AK2016b},
both in the linear and in the nonlinear cases (the latter is referred
as a multi-scaling tropical limit \cite{AK2016}).

\section{\label{sec: Basic statistical mappings} From soliton solutions to
statistical amoebas}

The explicit form of many soliton solutions of partial
differential equations can be derived from a sum of exponentials \cite{Hirota1980}.
In order to highlight the relation between statistical amoebas and
$\tau$-functions, it is worth noting that the partition function
(\ref{eq: standard partition function}) can be expressed as 
\begin{equation}
\mathcal{Z}=\det\left(\boldsymbol{g}^{T}\cdot\mathbf{\ensuremath{\Theta}}(\boldsymbol{\ensuremath{\varepsilon}})\cdot\boldsymbol{K}_{0}\right)\label{eq: determinantal partition function}
\end{equation}
where $\boldsymbol{g}:=\left(g_{1},g_{2},\dots,g_{N}\right)^{T}$, $\mathbf{\ensuremath{\Theta}}(\boldsymbol{\ensuremath{\varepsilon}}):=\mathrm{diag}\left(e^{-\frac{\varepsilon_{1}}{k_{B}T}},e^{-\frac{\varepsilon_{2}}{k_{B}T}},\dots,e^{-\frac{\varepsilon_{N}}{k_{B}T}}\right)$
and $\boldsymbol{K}_{0}:=(\underset{N}{\underbrace{1,\dots,1}})^{T}$.
In particular, $\boldsymbol{g}$ is a totally positive vector (all its
entries are positive) and $\boldsymbol{K}_{0}$ can be interpreted as
a $n\times1$ Vandermonde matrix. This formula for the partition function
coincides with a Wronskian $\tau$-function (\ref{eq: Cauchy-Binet formula, solitons, a}).

More generally, we can express such a type of $\tau$-functions as a sum of exponentials through the Cauchy-Binet expansion of the determinant of a product: if one introduces
\begin{eqnarray}
g_{\mathcal{I}} & := & \Delta_{\mathbf{A}}(\mathcal{I})\cdot\Delta_{\mathbf{K}}(\mathcal{I}),\label{eq: determinantal degeneration}\\
\Lambda_{\mathcal{I}}(\boldsymbol{x}) & := & \Delta_{\mathbf{A}}(\mathcal{I})\cdot\Delta_{\mathbf{K}}(\mathcal{I})\cdot\exp\left(\sum_{\alpha\in\mathcal{I}}\varphi_{\alpha}(\boldsymbol{x})\right),\quad\mathcal{I}\in\mathcal{P}_{k}[n]\label{eq: exponential terms}
\end{eqnarray}
then (\ref{eq: Cauchy-Binet formula, solitons}) is of the form (\ref{eq: standard partition function, exponential sums}),
\begin{eqnarray}
\tau(\boldsymbol{x}) & = & \det\left(\mathbf{A}\cdot\mathbf{\ensuremath{\Theta}}(\boldsymbol{x})\cdot\mathbf{K}\right)\nonumber \\
& = & \sum_{\mathcal{I}\in\mathcal{P}_{k}[n]}\Lambda_{\mathcal{I}}(\boldsymbol{x}).\label{eq: Cauchy-Binet as partition function}
\end{eqnarray}
In this way, we move from degenerations multiplicities $\boldsymbol{g}$
to more general products of minors $g_{\mathcal{I}}$. The dimension
$k$, which coincides with the rank of $\mathbf{A}$ and $\mathbf{K}$
when $\tau$ does not identically vanish, indicates the number of
line solitons at $x_{2}\gg0$, while $n-k$ is related to line solitons
at $x_{2}\ll0$. In the statistical perspective, $\tau$ is the partition
function for a statistical model whose configurations correspond to
subsets $\mathcal{I}\in\mathcal{P}_{k}[n]$ and have energies ${\displaystyle \sum_{\alpha\in\mathcal{I}}\varphi_{\alpha}}$.
If $\kappa_{\alpha}\in\mathbb{Z}$ and $\mathbf{A}\in\mathbb{Z}_{\star}^{k\times n}$
is a matrix with integer entries and maximal rank, then (\ref{eq: determinantal degeneration})
is an integer, $\Delta_{\mathbf{A}}(\mathcal{I})$ and $\Delta_{\mathbf{K}}(\mathcal{I})$
can be seen as degenerations for independent events and $\Delta_{\mathbf{A}}(\mathcal{I})\cdot\Delta_{\mathbf{K}}(\mathcal{I})$
is the joint degeneration. The $\alpha$th column of $\mathbf{A}$
generalizes $g_{\alpha}$ in (\ref{eq: standard partition function}),
so it can be regarded as a degeneration vector relative to the $\alpha$th
energy level \textit{$\varphi_{\alpha}(\boldsymbol{x})$, $\alpha\in[n]$. }

Regularity hypotheses for the case $k=1$ can be extended to $\tau$-functions
at $k\geq3$ too. For example, the entries of $\boldsymbol{g}$ in (\ref{eq: standard partition function})
and (\ref{eq: determinantal partition function}) are assumed to be
positive since they are a measure for degenerations associated with
energy levels. At $k>1$, this property generalizes to a real matrix
$\mathbf{A}\in\mathbb{R}_{\star}^{k\times n}$ where all the maximal
minors are non-negative. If one fixes the ordering $\kappa_{1}<\dots<\kappa_{n}$
for soliton parameters, then this request guarantees (indeed, it is
equivalent to, see e.g. \cite{Kodama2006}) the positivity of the
Cauchy-Binet expansion (\ref{eq: Cauchy-Binet formula, solitons}),
hence the regularity of the solution of the original KP II equation.

In the next section, this kind of request will be relaxed: the matrix
$\mathbf{A}$ is only assumed to obey the full-rank condition. However,
some peculiarities of the total non-negative case will be discussed
in Section \ref{sec: Levels of constrained amoebas}. 
\begin{defn}
\label{def: signatures and matroid} A \emph{choice of signs}, or
\emph{signature}, is a map 
\begin{equation}
\Sigma:\,\mathfrak{G}\longrightarrow\{\pm1\}\label{eq: choice of signs}
\end{equation}
where 
\begin{equation}
\mathfrak{G}:=\left\{ \mathcal{I}\in\mathcal{P}_{k}[n]:\,\Delta_{\mathbf{A}}(\mathcal{I})\neq0\right\} .\label{eq: set of vanishing minors}
\end{equation}
$\Sigma$ acts on (\ref{eq: Cauchy-Binet as partition function})
via $g(\mathcal{I})\mapsto\Sigma(\mathcal{I})\cdot g(\mathcal{I})$,
which returns a new function 
\begin{equation}
\sum_{\mathcal{I}\in\mathcal{P}_{k}[n]}\Sigma(\mathcal{I})\cdot g_{\mathcal{I}}\cdot e^{\sum_{\alpha\in\mathcal{I}}\varphi_{\alpha}(\boldsymbol{x})}.\label{eq: action of signature on tau}
\end{equation}
\end{defn}
The set $\mathfrak{G}$ is a \textit{matroid} on $[n]$: this means
that the exchange relation holds, i.e., 
\begin{equation}
\hspace*{-1.5cm}\mathrm{for\,\,all\,\,}\mathcal{A},\mathcal{B}\in\mathfrak{G},\,\alpha\in\mathcal{A}\setminus\mathcal{B},\mathrm{\,there\,\,exists\,\,}\beta\in\mathcal{B}\setminus\mathcal{A}\mathrm{\,\,such\,\,that\,\,}\mathcal{A}_{\beta}^{\alpha}\in\mathfrak{G}.\label{eq: exchange relation}
\end{equation}
In particular, we refer to a \textit{generic case} as one where $\mathfrak{G}=\mathcal{P}_{k}[n]$. 

\section{\label{sec: Constrained statistical amoebas_ solitons} Complexity
reduction through the KP II constraint}

The determinantal partition function (\ref{eq: determinantal partition function})
at $k=1$ reduces to the statistical amoebas studied in \cite{AK2016b}.
Cases at $k>1$ have more complications, due to the occurrence of
functional relations among the terms in the $\tau$-function (\ref{eq: Cauchy-Binet as partition function}).
Indeed, the minors of a ($k\times n$)-matrix have to satisfy the
well-known Grassmann-Pl\"{u}cker relations (see, e.g., \cite{GKZ1994},
§3.1). In particular, the three-terms Pl\"{u}cker relations 
\begin{equation}
\Delta_{\mathbf{A}}(\mathcal{H}_{\alpha\gamma})\cdot\Delta_{\mathbf{A}}(\mathcal{H}_{\beta\delta})=\Delta_{\mathbf{A}}(\mathcal{H}_{\alpha\beta})\cdot\Delta_{\mathbf{A}}(\mathcal{H}_{\gamma\delta})+\Delta_{\mathbf{A}}(\mathcal{H}_{\alpha\delta})\cdot\Delta_{\mathbf{A}}(\mathcal{H}_{\beta,\gamma})\label{eq: three-terms Plucker relations}
\end{equation}
hold for all $1\leq\alpha<\beta<\gamma<\delta\leq n$ and $\mathcal{H}\subset[n]$
with $\#\mathcal{H}=k-2$ and $\{\alpha,\beta,\gamma,\delta\}\cap\mathcal{H}=\emptyset$. 
\begin{example}
\label{exa: no 1-stratum free} For a generic matrix $\mathbf{A}\in\mathbb{R}_{\star}^{k\times n}$,
$k>1$, not all the signatures (\ref{eq: choice of signs}) for $\boldsymbol{g}\in\mathbb{R}^{{{n} \choose {k}}}$
in (\ref{eq: Cauchy-Binet as partition function}) correspond to another
choice of matrix $\mathbf{\ensuremath{\tilde{A}}}\in\mathbb{R}^{k\times n}$
in (\ref{eq: determinantal degeneration}). In fact, consider $\mathbf{A}\in\mathbb{R}^{k\times n}$
such that there exist two non-vanishing minors $\Delta_{\mathbf{A}}(\mathcal{H}_{\alpha\gamma})$
and $\Delta_{\mathbf{A}}(\mathcal{H}_{\beta\delta})$, e.g., taking
$\mathbf{A}$ parametrizing a point in the totally positive part of
the Grassmannian $\mathcal{GR}_{k,n}(\mathbb{R})$. Suppose that the
choice 
\begin{equation}
g_{\mathcal{I}}\mapsto\left\{ \begin{array}{cc}
-g_{\mathcal{I}}, & \mathcal{I}=\mathcal{H}_{\alpha\gamma}\\
g_{\mathcal{I}}, & \mbox{otherwise}
\end{array}\right.,\label{eq: non-integrable choice sign}
\end{equation}
which lies in the free statistical $1$-amoeba relative to (\ref{eq: Cauchy-Binet as partition function}),
corresponds to a certain matrix $\mathbf{\ensuremath{\tilde{A}}}\in\mathbb{R}^{k\times n}$.
The relation (\ref{eq: three-terms Plucker relations}) gives 
\begin{eqnarray}
\Delta_{\mathbf{A}}(\mathcal{H}_{\alpha\gamma})\cdot\Delta_{\mathbf{A}}(\mathcal{H}_{\beta\delta}) & = & \Delta_{\mathbf{A}}(\mathcal{H}_{\alpha\beta})\cdot\Delta_{\mathbf{A}}(\mathcal{H}_{\gamma\delta})+\Delta_{\mathbf{A}}(\mathcal{H}_{\alpha\delta})\cdot\Delta_{\mathbf{A}}(\mathcal{H}_{\beta\gamma})\nonumber \\
 & = & \Delta_{\mathbf{\ensuremath{\tilde{A}}}}(\mathcal{H}_{\alpha\beta})\cdot\Delta_{\mathbf{\ensuremath{\tilde{A}}}}(\mathcal{H}_{\gamma\delta})+\Delta_{\mathbf{\ensuremath{\tilde{A}}}}(\mathcal{H}_{\alpha\delta})\cdot\Delta_{\mathbf{\ensuremath{\tilde{A}}}}(\mathcal{H}_{\beta\gamma})\nonumber \\
 & = & \Delta_{\mathbf{\ensuremath{\tilde{A}}}}(\mathcal{H}_{\alpha\gamma})\cdot\Delta_{\mathbf{\ensuremath{\tilde{A}}}}(\mathcal{H}_{\beta\delta}) = -\Delta_{\mathbf{A}}(\mathcal{H}_{\alpha\gamma})\cdot\Delta_{\mathbf{A}}(\mathcal{H}_{\beta\delta})\label{eq: three-point Grassmann-Plucker}
\end{eqnarray}
that means $\Delta_{\mathbf{A}}(\mathcal{H}_{\alpha\gamma})\cdot\Delta_{\mathbf{A}}(\mathcal{H}_{\beta\delta})=0$,
a contradiction. It follows that, for points in the totally positive
Grassmannian, no signature of the type (\ref{eq: non-integrable choice sign})
preserves the form (\ref{eq: Wronskian}). 
\end{example}
In general, it is not a trivial task to check if a given map of the
type (\ref{eq: choice of signs}) follows from maximal minors of a
certain matrix. This issue is related to combinatorial structures
behind Grassmann-Pl\"{u}cker relations, namely \textit{chirotopes} $\chi:\,[n]^{k}\longrightarrow\{-1,0,+1\}$
(see, e.g., \cite{Ziegler2004,Bjorner1999} for more details). In
particular, chirotopes coming from $\mathbf{A}\in\mathbb{R}^{k\times n}$,
in the sense that $\chi(\mathcal{I})=\mathrm{sign}(\Delta_{\mathbf{A}}(\mathcal{I}))$
for all $\mathcal{I}\in[n]^{k}$, are said to be realizable. Both
the check of the realizability of chirotopes and their enumeration
have non-trivial complexity \cite{Fukuda2012}. However, in the present framework, the data on the initial function, in particular the ``lengths'' of the exponential and pre-exponential terms, are given:
starting from (\ref{eq: determinantal degeneration}), these absolute
values are known to be compatible with at least one determinantal
(Wronskian) form, and can be used to explore other signatures.

We exclude situations where (\ref{eq: Cauchy-Binet as partition function})
identically vanishes, so $\mathbf{A}$ has maximal rank and there
is no null row $\boldsymbol{0}_{n}^{T}$. If there exists a null column,
then the $\tau$-function does not depend on the corresponding soliton.
This leads to the reduction of a $n$- to a $(n-1)$-soliton solution.
So we can assume that there is no null column without loss of generality.

We now consider the conditions on choices of signs that map a $\tau$-function
of the KP II equation (\ref{eq: bilinear KP}) to another $\tau$-function, which will be called \textit{solitonic signatures}. Let (\ref{eq: choice of signs}) be any choice of signs for coefficients
$g_{\mathcal{I}}$ in (\ref{eq: Cauchy-Binet formula, solitons}),
which is equivalent to a choice of a partition of $\mathfrak{G}$
in two disjoint subsets $\mathfrak{G}=\boldsymbol{\ensuremath{\mathrm{PS}}}\cup\boldsymbol{\ensuremath{\mathrm{NS}}}$,
$\boldsymbol{\ensuremath{\mathrm{PS}}}\cap\boldsymbol{\ensuremath{\mathrm{NS}}}=\emptyset$,
where 
\begin{equation}
\boldsymbol{\ensuremath{\mathrm{PS}}}:=\Sigma^{-1}\left(\{+1\}\right),\quad\boldsymbol{\ensuremath{\mathrm{NS}}}:=\Sigma^{-1}\left(\{-1\}\right).\label{eq: partition from signature}
\end{equation}
So one can write the resulting exponential sum as 
\begin{equation}
\sum_{\mathcal{I}\in\mathcal{P}_{k}[n]}\Sigma(\mathcal{I})\cdot\Delta_{\mathbf{A}}(\mathcal{I})\cdot\Delta_{\mathbf{K}}(\mathcal{I})\cdot e^{\sum_{\alpha\in\mathcal{I}}\varphi_{\alpha}(\boldsymbol{x})}=\tau(\boldsymbol{x})-2\cdot\tau_{\Sigma}(\boldsymbol{x})\label{eq: signed tau-function}
\end{equation}
where 
\begin{equation}
\tau_{\Sigma}(\boldsymbol{x}):=\sum_{\mathcal{I}\in\boldsymbol{\ensuremath{\mathrm{NS}}}}\Delta_{\mathbf{A}}(\mathcal{I})\cdot\Delta_{\mathbf{K}}(\mathcal{I})\cdot e^{\sum_{\alpha\in\mathcal{I}}\varphi_{\alpha}(\boldsymbol{x})}.\label{eq: signed tau-function, negative part}
\end{equation}
Moreover, we will say 
\begin{equation}
\mathcal{I}\cong\mathcal{J}\Leftrightarrow\Sigma(\mathcal{I})=\Sigma(\mathcal{J}),\quad\mathcal{I},\mathcal{J}\in\mathcal{P}_{k}[n],\label{eq: congruence k-subsets}
\end{equation}
which defines a relation on $\mathfrak{G}$.

We do not assume \textit{a priori} that a partition (\ref{eq: partition from signature})
returns a determinant (\ref{eq: Wronskian}): this means that the
determinantal properties (i.e., Grassmann-Pl\"{u}cker relations) hold for
the family $\left\{ \Delta_{\mathbf{A}}(\mathcal{I}):\right.$ $\left.\mathcal{I}\in\mathfrak{G}\right\} $,
but not necessarily for $\left\{ \Sigma(\mathcal{I})\cdot\Delta_{\mathbf{A}}(\mathcal{I}):\,\mathcal{I}\in\mathfrak{G}\right\} $.

The bilinearity of the Hirota derivative and the KP operator (\ref{eq: bilinear KP})
gives 
\begin{eqnarray}
 &  & \mathrm{D_{KP}}(\tau-2\cdot\tau_{\Sigma},\tau-2\cdot\tau_{\Sigma})\nonumber \\
 & = & \mathrm{D_{KP}}(\tau,\tau)+4\cdot\mathrm{D_{KP}}(\tau_{\Sigma},\tau_{\Sigma})-2\cdot\mathrm{D_{KP}}(\tau,\tau_{\Sigma})-2\cdot\mathrm{D_{KP}}(\tau_{\Sigma},\tau)\nonumber \\
 & = & 4\cdot\mathrm{D_{KP}(\tau_{\Sigma},\tau_{\Sigma})}-2\cdot\mathrm{D_{KP}}(\tau,\tau_{\Sigma})-2\cdot\mathrm{D_{KP}}(\tau_{\Sigma},\tau)\label{eq: Dkp orthogonality, pre}
\end{eqnarray}
since $\mathrm{D_{KP}}(\tau,\tau)=0$. Therefore, one has 
\begin{equation}
0=\mathrm{D_{KP}}(\tau,\tau_{\Sigma})-\mathrm{D_{KP}(\tau_{\Sigma},\tau_{\Sigma})}=\mathrm{D_{KP}(\tau-\tau_{\Sigma},\tau_{\Sigma})}.\label{eq: Dkp orthogonality}
\end{equation}
We can say that $\tau-\tau_{\Sigma}$ is ``orthogonal'' to $\tau_{\Sigma}$
with respect to the $\mathrm{D_{KP}}$ bilinear operator. The bilinearity
also implies that (\ref{eq: Dkp orthogonality}) is equivalent to
\begin{eqnarray}
0 & = & \mathrm{D_{KP}(\tau-\tau_{\Sigma},\tau_{\Sigma})}\nonumber \\
& = & \sum_{\mathcal{A}\in\boldsymbol{\ensuremath{\mathrm{PS}}}}\sum_{\mathcal{B}\in\boldsymbol{\ensuremath{\mathrm{NS}}}}g(\mathcal{A})g(\mathcal{B})\cdot\mathrm{D_{KP}}\left[\exp\left(\sum_{\alpha\in\mathcal{A}}\sum_{u=1}^{d}\kappa_{\alpha}^{u}x_{u}\right),\exp\left(\sum_{\beta\in\mathcal{B}}\sum_{u=1}^{d}\kappa_{\beta}^{u}x_{u}\right)\right]\nonumber \\
& = & \sum_{\mathcal{A}\in\boldsymbol{\ensuremath{\mathrm{PS}}}}\sum_{\mathcal{B}\in\boldsymbol{\ensuremath{\mathrm{NS}}}}g(\mathcal{A})g(\mathcal{B})\exp\left(2\sum_{\beta\in\mathcal{B}}\sum_{u=1}^{d}\kappa_{\beta}^{u}x_{u}\right)\mathrm{D_{KP}}\left[1,\exp\left(\sum_{u=1}^{d}\sum_{\alpha\in\mathcal{A}}\kappa_{\alpha}^{u}x_{u}-\sum_{\beta\in\mathcal{B}}\kappa_{\beta}^{u}x_{u}\right)\right].\nonumber \\
\label{eq: transversal +/-}
\end{eqnarray}
From 
\begin{eqnarray}
 &  & \mathrm{D_{KP}}\left[1,\exp\left(\sum_{u=1}^{d}\sum_{\alpha\in\mathcal{A}}\kappa_{\alpha}^{u}x_{u}-\sum_{\beta\in\mathcal{B}}\kappa_{\beta}^{u}x_{u}\right)\right]\nonumber \\
 & = & (\partial_{x_{1}}^{4}-4\cdot\partial_{x_{1}}\partial_{x_{3}}+3\cdot\partial_{x_{2}}^{2})\exp\left(\sum_{u=1}^{d}\sum_{\alpha\in\mathcal{A}}\kappa_{\alpha}^{u}x_{u}-\sum_{\beta\in\mathcal{B}}\kappa_{\beta}^{u}x_{u}\right)\nonumber \\
 & = & C(\mathcal{A},\mathcal{B};\boldsymbol{\ensuremath{\kappa}})\cdot\exp\left(\sum_{u=1}^{d}\sum_{\alpha\in\mathcal{A}}\kappa_{\alpha}^{u}x_{u}-\sum_{\beta\in\mathcal{B}}\kappa_{\beta}^{u}x_{u}\right)\label{eq: from Hirota to usual derivative}
\end{eqnarray}
where 
\begin{eqnarray}
C(\mathcal{A},\mathcal{B};\boldsymbol{\ensuremath{\kappa}}) & := & \left[\sum_{\alpha\in\mathcal{A}}\kappa_{\alpha}-\sum_{\beta\in\mathcal{B}}\kappa_{\beta}\right]^{4}+3\cdot\left[\sum_{\alpha\in\mathcal{A}}\kappa_{\alpha}^{2}-\sum_{\beta\in\mathcal{B}}\kappa_{\beta}^{2}\right]^{2}\nonumber \\
 &  & -4\cdot\left[\sum_{\alpha\in\mathcal{A}}\kappa_{\alpha}-\sum_{\beta\in\mathcal{B}}\kappa_{\beta}\right]\cdot\left[\sum_{\alpha\in\mathcal{A}}\kappa_{\alpha}^{3}-\sum_{\beta\in\mathcal{B}}\kappa_{\beta}^{3}\right].\label{eq: coefficient}
\end{eqnarray}
The equation (\ref{eq: transversal +/-}) is equivalent to 
\begin{equation}
\sum_{\mathcal{A}\in\boldsymbol{\ensuremath{\mathrm{PS}}}}\sum_{\mathcal{B}\in\boldsymbol{\ensuremath{\mathrm{NS}}}}g(\mathcal{A})g(\mathcal{B})\cdot C(\mathcal{A},\mathcal{B};\boldsymbol{\ensuremath{\kappa}})\cdot\exp\left(\sum_{u=1}^{d}\sum_{\alpha\in\mathcal{A}}\kappa_{\alpha}^{u}x_{u}+\sum_{\beta\in\mathcal{B}}\kappa_{\beta}^{u}x_{u}\right)=0.\label{eq: transversal +/-, b}
\end{equation}
Two exponentials in (\ref{eq: transversal +/-, b}) associated with
the pairs $(\mathcal{A},\mathcal{B})$ and $(\mathcal{C},\mathcal{D})$
identically coincide if and only if 
\begin{equation}
\sum_{\alpha\in\mathcal{A}}\kappa_{\alpha}^{u}+\sum_{\beta\in\mathcal{B}}\kappa_{\beta}^{u}=\sum_{\gamma\in\mathcal{C}}\kappa_{\gamma}^{u}+\sum_{\delta\in\mathcal{D}}\kappa_{\delta}^{u},\quad u\in[d].\label{eq: condition coincidence exponentials}
\end{equation}
The only occurrences of (\ref{eq: condition coincidence exponentials})
that are satisfied for a \emph{generic} choice of the parameters $\boldsymbol{\ensuremath{\kappa}}$
involve the cases when $\mathcal{A}\cup\mathcal{B}=\mathcal{C}\cup\mathcal{D}$
and $\mathcal{A}\cap\mathcal{B}=\mathcal{C}\cap\mathcal{D}$. So,
if one assumes that there is no algebraic dependence that relates
these sums when $\mathcal{A}\cup\mathcal{B}\neq\mathcal{C}\cup\mathcal{D}$
or $\mathcal{A}\cap\mathcal{B}\neq\mathcal{C}\cap\mathcal{D}$, then
each exponential term in (\ref{eq: condition coincidence exponentials})
is determined by the union $\mathcal{A}\cup\mathcal{B}$ and the intersection
$\mathcal{A}\cap\mathcal{B}$, since 
\begin{equation}
\sum_{u=1}^{d}\sum_{\alpha\in\mathcal{A}}\kappa_{\alpha}^{u}x_{u}+\sum_{\beta\in\mathcal{B}}\kappa_{\beta}^{u}x_{u}=\sum_{u=1}^{d}\sum_{\alpha\in\mathcal{A}\cap\mathcal{B}}\kappa_{\alpha}^{u}\cdot x_{u}+\sum_{\beta\in\mathcal{A}\cup\mathcal{B}}\kappa_{\beta}^{u}\cdot x_{u}.\label{eq: exponential from intersection and union}
\end{equation}
In particular, one has the following 
\begin{lem}
\label{lem: all three terms or none} Assume that the minors $\Delta_{\mathbf{A}}(\mathcal{I})$
with $\mathcal{I}\in\mathcal{P}_{k}[n]$ are not vanishing, i.e., $\mathfrak{G}=\mathcal{P}_{k}[n]$.
Then for each $\mathcal{H}\subseteq\mathcal{L}\subseteq[n]$ with
$\#\mathcal{H}+4=\#\mathcal{L}=k+2$, the number of pairs $(\mathcal{I}_{+},\mathcal{I}_{-})$
associated with a term in (\ref{eq: transversal +/-, b}) with union
$\mathcal{I}_{+}\cup\mathcal{I}_{-}=\mathcal{L}$ and intersection
$\mathcal{I}_{+}\cap\mathcal{I}_{-}=\mathcal{H}$ is equal to (\ref{eq: exponential from intersection and union})
is $0$ or $3$. With the notation introduced in (\ref{eq: k-subset single exchange})-(\ref{eq: congruence k-subsets}),
this can be stated as 
\begin{equation}
\mathcal{I}_{\gamma\delta}^{\alpha\beta}\cong\mathcal{I}\Leftrightarrow\mathcal{I}_{\gamma}^{\alpha}\cong\mathcal{I}_{\delta}^{\beta}\Leftrightarrow\mathcal{I}_{\delta}^{\alpha}\cong\mathcal{I}_{\gamma}^{\beta}.\label{eq: 0 or 3, lemma restated}
\end{equation}
\end{lem}
\begin{proof}
Suppose that there exist subsets $\mathcal{I}_{+}\in\boldsymbol{\ensuremath{\mathrm{PS}}}$,
$\mathcal{I}_{-}\in\boldsymbol{\ensuremath{\mathrm{NS}}}$ such that $\mathcal{H}:=\mathcal{I}_{+}\cap\mathcal{I}_{-}=:\{\gamma_{1},\dots,\gamma_{k-2}\}$,
$\gamma_{1}<\dots<\gamma_{k-2}$, and $\mathcal{L}:=\mathcal{I}_{+}\cup\mathcal{I}_{-}$,
so the number of terms in (\ref{eq: transversal +/-, b}) associated
with $\mathcal{H}$ and $\mathcal{L}$ is not $0$. In particular,
$\#(\mathcal{I}_{+}\Delta\mathcal{I}_{-})=4$. Then, there exist $\alpha_{+},\beta_{+}\in\mathcal{I}_{+}$
and $\alpha_{-},\beta_{-}\in\mathcal{I}_{-}$ such that $\mathcal{I}_{\pm}=\mathcal{H}\cup\{\alpha_{\pm},\beta_{\pm}\}$.
There are two additional pairs of subsets different from $\{\mathcal{I}_{+},\mathcal{I}_{-}\}$
with the same union and intersection. Thus, in order to have a vanishing
coefficient for the associated term in (\ref{eq: transversal +/-, b}),
at least one of these two pairs has to belong to $\boldsymbol{\ensuremath{\mathrm{PS}}}\times\boldsymbol{\ensuremath{\mathrm{NS}}}$.
Let us suppose that only one of these two possibilities is in $\boldsymbol{\ensuremath{\mathrm{PS}}}\times\boldsymbol{\ensuremath{\mathrm{NS}}}$,
call it $(\mathcal{L}_{+},\mathcal{L}_{-})$ with $\mathcal{L}_{\pm}=\mathcal{H}\cup\{\alpha_{\pm},\beta_{\mp}\}$.
In such a case, the term $C(\mathcal{I}_{+},\mathcal{I}_{-};\boldsymbol{\ensuremath{\kappa}})$
in (\ref{eq: coefficient}) is 
\begin{equation}
C(\mathcal{I}_{+},\mathcal{I}_{-};\boldsymbol{\ensuremath{\kappa}})=12\cdot(\kappa_{\alpha_{+}}-\kappa_{\alpha_{-}})(\kappa_{\beta_{+}}-\kappa_{\alpha_{-}})(\kappa_{\alpha_{+}}-\kappa_{\beta_{-}})(\kappa_{\beta_{+}}-\kappa_{\beta_{-}})\label{eq: coefficient, distance 2 case}
\end{equation}
that is not vanishing since the soliton parameters are pairwise distinct
by assumption. Using the parity $\mathrm{P}((\beta_{-},\alpha_{-},\beta_{+},\alpha_{+}))$
of the permutation $(\beta_{-},\alpha_{-},\beta_{+},\alpha_{+})$,
we introduce 
\begin{eqnarray}
\sigma(\alpha_{+},\beta_{+}|\alpha_{-},\beta_{-}) & := & \mathrm{sign}\left[(\alpha_{+}-\alpha_{-})(\beta_{+}-\alpha_{-})(\alpha_{+}-\beta_{-})(\beta_{+}-\beta_{-})\right]\nonumber \\
 & = & \mathrm{P}((\beta_{-},\alpha_{-},\beta_{+},\alpha_{+}))\cdot\mathrm{sign}\left[(\alpha_{+}-\beta_{+})(\alpha_{-}-\beta_{-})\right]\label{eq: sign 2-2, a}
\end{eqnarray}
which gives 
\begin{eqnarray}
\hspace{-2cm} &  & \Delta_{\mathbf{K}}(\mathcal{I}_{+})\cdot\Delta_{\mathbf{K}}(\mathcal{I}_{-})\cdot C(\mathcal{I}_{+},\mathcal{I}_{-};\boldsymbol{\ensuremath{\kappa}})\nonumber \\
\hspace{-2cm} & = & 12\cdot\prod_{\alpha_{i}<\alpha_{j}}^{\mathcal{I}_{+}}(\kappa_{\alpha_{j}}-\kappa_{\alpha_{i}})\cdot\prod_{\beta_{i}<\beta_{j}}^{\mathcal{I}_{-}}(\kappa_{\beta_{j}}-\kappa_{\beta_{i}})\cdot(\kappa_{\alpha_{+}}-\kappa_{\alpha_{-}})(\kappa_{\beta_{+}}-\kappa_{\alpha_{-}})(\kappa_{\alpha_{+}}-\kappa_{\beta_{-}})(\kappa_{\beta_{+}}-\kappa_{\beta_{-}})\nonumber \\
\hspace{-2cm} & = & 12\cdot\prod_{\gamma_{i}<\gamma_{j}}^{\mathcal{I}_{+}\cap\mathcal{I}_{-}}(\kappa_{\gamma_{j}}-\kappa_{\gamma_{i}})\cdot\prod_{\delta_{i}<\delta_{j}}^{\mathcal{I}_{+}\cup\mathcal{I}_{-}}(\kappa_{\delta_{j}}-\kappa_{\delta_{i}})\cdot\sigma(\alpha_{+},\beta_{+}|\alpha_{-},\beta_{-})\nonumber \\
\hspace{-2cm} & = & 12\cdot\mathrm{VdM}(\boldsymbol{\ensuremath{\kappa}};\mathcal{I}_{+}\cap\mathcal{I}_{-})\cdot\mathrm{VdM}(\boldsymbol{\ensuremath{\kappa}};\mathcal{I}_{+}\cup\mathcal{I}_{-})\cdot\sigma(\alpha_{+},\beta_{+}|\alpha_{-},\beta_{-}).\label{eq: equivalent expression coefficient, a}
\end{eqnarray}
The signs $\sigma(\alpha_{+},\beta_{-}|\alpha_{-},\beta_{+})$ and
$\sigma(\alpha_{+},\alpha_{-}|\beta_{+},\beta_{-})$ can be found
in the same way, taking into account that 
\begin{equation}
\mathrm{P}((\beta_{-},\alpha_{-},\beta_{+},\alpha_{+}))=-\mathrm{P}((\beta_{-},\beta_{+},\alpha_{-},\alpha_{+}))=-\mathrm{P}((\beta_{+},\alpha_{-},\beta_{-},\alpha_{+})).\label{eq: antisymmetry Vandermonde}
\end{equation}
Furthermore, one can consider 
\begin{equation}
S(\alpha,\mathcal{H}):=\max\{i\in[k]:\,\gamma_{i-1}<\alpha\},\quad\alpha\in[n]\backslash\mathcal{H}\label{eq: four splitting indices}
\end{equation}
where $\gamma_{0}:=0$. In this way, it is easy to check that 
\begin{equation}
\Delta_{\mathbf{A}}(\mathcal{H}_{\alpha\beta})=\Delta_{\mathbf{A}}(\alpha,\beta;\mathcal{H})\cdot(-1)^{1+S(\alpha,\mathcal{H})+S(\beta,\mathcal{H})}\cdot\mathrm{sign}\left[\alpha-\beta\right].\label{sign 2-2, b}
\end{equation}
where $\Delta_{\mathbf{A}}(\alpha,\beta;\mathcal{H})$ is the product
of $\Delta_{\mathbf{A}}(\mathcal{H}_{\alpha\beta})$ and the parity
of the permutation $(\alpha,\beta,\gamma_{1},\dots,\gamma_{k-2})$.
Thus, one gets 
\begin{eqnarray}
\hspace{-2cm} &  & g(\mathcal{I}_{+})\cdot g(\mathcal{I}_{-})\cdot C_{\mathcal{I}_{+}\mathcal{I}_{-}}(\boldsymbol{\ensuremath{\kappa}})\nonumber \\
\hspace{-2cm} & = & 12\cdot\mathrm{VdM}\left(\boldsymbol{\ensuremath{\kappa}};\mathcal{I}_{+}\cap\mathcal{I}_{-}\right)\cdot\mathrm{VdM}\left(\boldsymbol{\ensuremath{\kappa}};\mathcal{I}_{+}\cup\mathcal{I}_{-}\right)\cdot\mathrm{P}((\beta_{-},\alpha_{-},\beta_{+},\alpha_{+}))\nonumber \\
\hspace{-2cm} & \cdot & \left(\Delta_{\mathbf{A}}\left(\mathcal{H}_{\alpha_{+}\beta_{+}}\right)\cdot\Delta_{\mathbf{A}}\left(\mathcal{H}_{\alpha_{-}\beta_{-}}\right)\cdot\mathrm{sign}\left[(\alpha_{+}-\beta_{+})(\alpha_{-}-\beta_{-})\right]\right)\nonumber \\
\hspace{-2cm} & = & \left\{ 12\cdot\mathrm{VdM}\left(\boldsymbol{\ensuremath{\kappa}};\mathcal{I}_{+}\cap\mathcal{I}_{-}\right)\cdot\mathrm{VdM}\left(\boldsymbol{\ensuremath{\kappa}};\mathcal{I}_{+}\cup\mathcal{I}_{-}\right)\cdot(-1)^{S(\alpha_{+},\mathcal{H})+S(\beta_{+},\mathcal{H})+S(\alpha_{-},\mathcal{H})+S(\beta_{-},\mathcal{H})}\right\} \nonumber \\
\hspace{-2cm} & \cdot & \left\{ \mathrm{P}((\beta_{-},\alpha_{-},\beta_{+},\alpha_{+}))\cdot\Delta_{\mathbf{A}}(\alpha_{+},\beta_{+}|\mathcal{H})\cdot\Delta_{\mathbf{A}}(\alpha_{-},\beta_{-}|\mathcal{H})\right\} .\label{eq: sign 2-2, c}
\end{eqnarray}
The first term in braces in (\ref{eq: sign 2-2, c}) is the same for
all the pairs $(\mathcal{G}_{+},\mathcal{G}_{-})$ with $\mathcal{G}_{+}\cap\mathcal{G}_{-}=\mathcal{H}$
and $\mathcal{G}_{+}\cup\mathcal{G}_{-}=\mathcal{L}$. For the second
term in braces, the antisymmetry (\ref{eq: antisymmetry Vandermonde})
implies that 
\begin{eqnarray}
\hspace{-2cm} &  & g(\mathcal{L}_{+})\cdot g(\mathcal{L}_{-})\cdot C_{\mathcal{L}_{+}\mathcal{L}_{-}}(\boldsymbol{\ensuremath{\kappa}})\nonumber \\
\hspace{-2cm} & = & \left[12\cdot\mathrm{VdM}(\boldsymbol{\ensuremath{\kappa}};\mathcal{I}_{+}\cap\mathcal{I}_{-})\cdot\mathrm{VdM}(\boldsymbol{\ensuremath{\kappa}};\mathcal{I}_{+}\cup\mathcal{I}_{-})\cdot(-1)^{S(\alpha_{+},\mathcal{H})+S(\beta_{+},\mathcal{H})+S(\alpha_{-},\mathcal{H})+S(\beta_{-},\mathcal{H})}\right]\nonumber \\
\hspace{-2cm} & \cdot & \left[-\mathrm{P}((\beta_{-},\alpha_{-},\beta_{+},\alpha_{+}))\cdot\Delta_{\mathbf{A}}(\alpha_{+},\beta_{-}|\mathcal{H})\cdot\Delta_{\mathbf{A}}(\alpha_{-},\beta_{+}|\mathcal{H})\right].\label{eq: sign 2-2, d}
\end{eqnarray}
The three-terms Pl\"{u}cker relations (\ref{eq: three-terms Plucker relations}),
which are valid for minors $\Delta_{\mathbf{A}}(\mathcal{I})$, can
be stated for any four pairwise distinct elements $\delta_{a}\in[n]\backslash\mathcal{H}$,
$a\in[4]$, as 
\begin{equation}
\Delta_{\mathbf{A}}(\delta_{1},\delta_{2}|\mathcal{H})\cdot\Delta_{\mathbf{A}}(\delta_{3},\delta_{4}|\mathcal{H})=\Delta_{\mathbf{A}}(\delta_{1},\delta_{3}|\mathcal{H})\cdot\Delta_{\mathbf{A}}(\delta_{2},\delta_{4}|\mathcal{H})-\Delta_{\mathbf{A}}(\delta_{1},\delta_{4}|\mathcal{H})\cdot\Delta_{\mathbf{A}}(\delta_{2},\delta_{3}|\mathcal{H}).\label{eq: three-terms Plucker relations, bis}
\end{equation}
In particular, if one looks at the sum of (\ref{eq: sign 2-2, c})
and (\ref{eq: sign 2-2, d}) and applies (\ref{eq: three-terms Plucker relations, bis})
with $(\delta_{1},\delta_{2},\delta_{3},\delta_{4})\equiv(\alpha_{+},\alpha_{-},\beta_{+},\beta_{-})$,
the result is 
\begin{eqnarray}
\hspace{-2cm} &  & g(\mathcal{I}_{+})\cdot g(\mathcal{I}_{-})\cdot C_{\mathcal{I}_{+}\mathcal{I}_{-}}(\boldsymbol{\ensuremath{\kappa}})+g(\mathcal{L}_{+})\cdot g(\mathcal{L}_{-})\cdot C_{\mathcal{L}_{+}\mathcal{L}_{-}}(\boldsymbol{\ensuremath{\kappa}})\nonumber \\
\hspace{-2cm} & = & 12\cdot\mathrm{VdM}(\boldsymbol{\ensuremath{\kappa}};\mathcal{I}_{+}\cap\mathcal{I}_{-})\cdot\mathrm{VdM}(\boldsymbol{\ensuremath{\kappa}};\mathcal{I}_{+}\cup\mathcal{I}_{-})\nonumber \\
\hspace{-2cm} & \cdot & (-1)^{S(\alpha_{+},\mathcal{H})+S(\beta_{+},\mathcal{H})+S(\alpha_{-},\mathcal{H})+S(\beta_{-},\mathcal{H})}\cdot\mathrm{P}((\beta_{-},\alpha_{-},\beta_{+},\alpha_{+}))\nonumber \\
\hspace{-2cm} & \cdot & \left(\Delta_{\mathbf{A}}(\alpha_{+},\beta_{+}|\mathcal{H})\cdot\Delta_{\mathbf{A}}(\alpha_{-},\beta_{-}|\mathcal{H})-\Delta_{\mathbf{A}}(\alpha_{+},\beta_{-}|\mathcal{H})\cdot\Delta_{\mathbf{A}}(\alpha_{-},\beta_{+}|\mathcal{H})\right)\nonumber \\
\hspace{-2cm} & = & 12\cdot\mathrm{VdM}(\boldsymbol{\ensuremath{\kappa}};\mathcal{I}_{+}\cap\mathcal{I}_{-})\cdot\mathrm{VdM}(\boldsymbol{\ensuremath{\kappa}};\mathcal{I}_{+}\cup\mathcal{I}_{-})\cdot(-1)^{S(\alpha_{+},\mathcal{H})+S(\beta_{+},\mathcal{H})+S(\alpha_{-},\mathcal{H})+S(\beta_{-},\mathcal{H})}\nonumber \\
\hspace{-2cm} & \cdot & \mathrm{P}((\beta_{-},\alpha_{-},\beta_{+},\alpha_{+}))\cdot\Delta_{\mathbf{A}}(\alpha_{+},\alpha_{-}|\mathcal{H})\cdot\Delta_{\mathbf{A}}(\beta_{+},\beta_{-}|\mathcal{H})\neq0\label{eq: sign 2-2, e}
\end{eqnarray}
since we have assumed that all the minors are not vanishing. Hence,
the coefficient in (\ref{eq: transversal +/-, b}) associated with
$\mathcal{H}$ and $\mathcal{L}$ is not vanishing and the associated
$\tau$-function is not a solution of the KP equation (\ref{eq: bilinear KP}).

On the contrary, if all the three terms are involved in (\ref{eq: transversal +/-, b}),
then their sum is 
\begin{eqnarray}
 &  & 12\cdot\mathrm{VdM}(\boldsymbol{\ensuremath{\kappa}};\mathcal{I}_{+}\cap\mathcal{I}_{-})\cdot\mathrm{VdM}(\boldsymbol{\ensuremath{\kappa}};\mathcal{I}_{+}\cup\mathcal{I}_{-})\nonumber \\
 & \cdot & (-1)^{S(\alpha_{+},\mathcal{H})+S(\beta_{+},\mathcal{H})+S(\alpha_{-},\mathcal{H})+S(\beta_{-},\mathcal{H})}\cdot\mathrm{P}((\beta_{-},\alpha_{-},\beta_{+},\alpha_{+}))\nonumber \\
 & \cdot & \left(\Delta_{\mathbf{A}}(\alpha_{+},\beta_{+}|\mathcal{H})\cdot\Delta_{\mathbf{A}}(\alpha_{-},\beta_{-}|\mathcal{H})-\Delta_{\mathbf{A}}(\alpha_{+},\beta_{-}|\mathcal{H})\cdot\Delta_{\mathbf{A}}(\alpha_{-},\beta_{+}|\mathcal{H})\right.\nonumber \\
 &  & \left.-\Delta_{\mathbf{A}}(\alpha_{+},\alpha_{-}|\mathcal{H})\cdot\Delta_{\mathbf{A}}(\beta_{+},\beta_{-}|\mathcal{H})\right)=0\label{eq: sign 2-2, f}
\end{eqnarray}
since the second term in square brackets vanishes due to the three-terms
Pl\"{u}cker relations (\ref{eq: three-terms Plucker relations, bis}). 
\end{proof}
\begin{prop}
\label{prop: all three terms or none, including vanishing minors}
Let $\mathcal{I}\in\mathcal{P}_{k}[n]$, $\alpha_{1},\alpha_{2}\in\mathcal{I}$
and $\delta_{1},\delta_{2}\in[n]\setminus\mathcal{I}$ such that $\mathcal{I},{\displaystyle \mathcal{I}_{\delta_{1}\delta_{2}}^{\alpha_{1}\alpha_{2}}}\in\mathfrak{G}$.
Then at least one of the two products $\Delta_{\mathbf{A}}(\mathcal{I}_{\delta_{T}}^{\alpha_{1}})\cdot\Delta_{\mathbf{A}}(\mathcal{I}_{\delta_{U}}^{\alpha_{2}})$,
$\{T,U\}=\{1,2\}$, is not vanishing. Further, if $\Delta_{\mathbf{A}}(\mathcal{I}_{\delta_{T}}^{\alpha_{1}})\cdot\Delta_{\mathbf{A}}(\mathcal{I}_{\delta_{U}}^{\alpha_{2}})\neq0$,
then $\Sigma(\mathcal{I}_{\delta_{T}}^{\alpha_{1}})\cdot\Sigma(\mathcal{I}_{\delta_{U}}^{\alpha_{2}})=\Sigma(\mathcal{I})\cdot\Sigma({\displaystyle \mathcal{I}_{\delta_{1}\delta_{2}}^{\alpha_{1}\alpha_{2}}})$. 
\end{prop}
\begin{proof}
First suppose that one of the product vanishes, e.g., $\Delta_{\mathbf{A}}(\mathcal{I}_{\delta_{2}}^{\alpha_{1}})\cdot\Delta_{\mathbf{A}}(\mathcal{I}_{\delta_{1}}^{\alpha_{2}})=0$
without loss of generality. Then the three-terms Pl\"{u}cker relations
(\ref{eq: three-terms Plucker relations}) imply that $|\Delta_{\mathbf{A}}(\mathcal{I}_{\delta_{2}}^{\alpha_{1}})\cdot\Delta_{\mathbf{A}}(\mathcal{I}_{\delta_{1}}^{\alpha_{2}})|=|\Delta_{\mathbf{A}}(\mathcal{I})\cdot\Delta_{\mathbf{A}}({\displaystyle \mathcal{I}_{\delta_{1}\delta_{2}}^{\alpha_{1}\alpha_{2}}})|\neq0$.
Further, $\Sigma(\mathcal{I}_{\delta_{T}}^{\alpha_{1}})\cdot\Sigma(\mathcal{I}_{\delta_{U}}^{\alpha_{2}})\neq\Sigma(\mathcal{I})\cdot\Sigma({\displaystyle \mathcal{I}_{\delta_{1}\delta_{2}}^{\alpha_{1}\alpha_{2}}})$
implies that there is exactly one non-vanishing term in (\ref{eq: transversal +/-, b})
associated with a pair in $\boldsymbol{\ensuremath{\mathrm{PS}}}\times\boldsymbol{\ensuremath{\mathrm{NS}}}$
with intersection $\mathcal{I}^{\alpha_{1}\alpha_{2}}$ and union
$\mathcal{I}_{\delta_{1}\delta_{2}}$, hence (\ref{eq: sign 2-2, f})
do not vanish. Thus $\Sigma(\mathcal{I}_{\delta_{T}}^{\alpha_{1}})\cdot\Sigma(\mathcal{I}_{\delta_{U}}^{\alpha_{2}})=\Sigma(\mathcal{I})\cdot\Sigma({\displaystyle \mathcal{I}_{\delta_{1}\delta_{2}}^{\alpha_{1}\alpha_{2}}})$.
If instead both $\Delta_{\mathbf{A}}(\mathcal{I}_{\delta_{1}}^{\alpha_{1}})\cdot\Delta_{\mathbf{A}}(\mathcal{I}_{\delta_{2}}^{\alpha_{2}})$
and $\Delta_{\mathbf{A}}(\mathcal{I}_{\delta_{2}}^{\alpha_{1}})\cdot\Delta_{\mathbf{A}}(\mathcal{I}_{\delta_{1}}^{\alpha_{2}})$
are not vanishing, then (\ref{eq: 0 or 3, lemma restated}) holds
because of Lemma \ref{lem: all three terms or none} (see in particular
(\ref{eq: sign 2-2, e}) and (\ref{eq: sign 2-2, f})).  
\end{proof}
\begin{lem}
\label{lem: chain of non-vanishing} Let $\mathcal{H},\mathcal{K}\in\mathfrak{G}$
with $r:=\#(\mathcal{H}\backslash\mathcal{K})$. Then there exists
a finite sequence $\mathcal{L}_{0}:=\mathcal{H},\,\mathcal{L}_{1},\,\dots,\mathcal{L}_{r}:=\mathcal{K}$
of elements of $\mathfrak{G}$ such that $\#(\mathcal{L}_{u-1}\Delta\mathcal{L}_{u})=2$,
$u\in[r]$. 
\end{lem}
\begin{proof}
Let $\mathcal{H}\backslash\mathcal{K}=:\{\gamma_{1},\dots,\gamma_{r}\}$
and $\mathcal{L}_{0}:=\mathcal{H}$, so the exchange property (\ref{eq: exchange relation})
implies that there exists $\Psi(\gamma_{1})\in\mathcal{K}\setminus\mathcal{H}$
such that $\mathcal{L}_{1}:=\mathcal{H}\setminus\{\gamma_{1}\}\cup\{\Psi(\gamma_{1})\}\in\mathfrak{G}$.
Note that $\#(\mathcal{L}_{1}\Delta\mathcal{K})=r-1$. One can iterate
the process, starting from $\mathcal{L}_{u}\in\mathfrak{G}$ and finding
$\Psi(\gamma_{u+1})\in\mathcal{K}\setminus\mathcal{L}_{u}$ in order
to define 
\begin{equation}
\mathcal{L}_{u+1}:=\mathcal{H}\backslash\{\gamma_{1},\dots,\gamma_{u+1}\}\cup\{\Psi(\gamma_{1}),\dots,\Psi(\gamma_{u+1})\},\quad u\in[r]\label{eq: chain of non-vanishing}
\end{equation}
such that $\mathcal{L}_{u+1}\in\mathfrak{G}$ too. At each step, the
element $\Psi(\gamma_{u+1})$ is different from $\Psi(\gamma_{t})$
for all $t\leq u$ because $\Psi(\gamma_{t})\in\mathcal{L}_{u}$,
while $\Psi(\gamma_{u+1})\in\mathcal{K}\setminus\mathcal{L}_{u}$.
So the map $\Psi:\,\mathcal{H}\backslash\mathcal{K}\longrightarrow\mathcal{K}\backslash\mathcal{H}$
is injective and, hence, a bijection, since $\#\mathcal{H}=k=\#\mathcal{K}$
implies $\#(\mathcal{H}\backslash\mathcal{K})=\#(\mathcal{K}\backslash\mathcal{H})$.
All the subsets $\mathcal{L}_{0},\mathcal{L}_{1},\dots,\mathcal{L}_{r}$
are elements of $\mathfrak{G}$ and $\#(\mathcal{L}_{u-1}\Delta\mathcal{L}_{u})=2$
for all $u\in[r]$.  
\end{proof}
Now we consider the following relations 
\begin{equation}
\alpha\approx_{\mathcal{A}}\beta\Leftrightarrow\mathcal{A}\cong\mathcal{A}_{\alpha}^{\beta},\quad\alpha\in[n]\setminus\mathcal{I},\,\beta\in\mathcal{I}.\label{eq: internal relation, implicit}
\end{equation}
If there is $\lambda\in\mathcal{A}$ such that $\mathcal{I}_{\lambda}^{\alpha},\mathcal{I}_{\lambda}^{\beta}\in\mathfrak{G}$
where $\mathcal{I}=\mathcal{A}_{\alpha}^{\lambda}$, then one can
express $\alpha\approx_{\mathcal{I}}\beta$ as $\mathcal{I}_{\lambda}^{\alpha}\cong\mathcal{I}_{\lambda}^{\beta}$.
This does not depend on the choice of $\lambda\in[n]\backslash\mathcal{I}$,
as can be easily shown: 
\begin{rem}
\label{rem: internal relation} Let $\mathcal{I}\in\mathcal{P}_{k}[n]$
and $\alpha,\beta\in\mathcal{I}$. Then $\mathcal{I}_{\gamma}^{\alpha}\cong\mathcal{I}_{\gamma}^{\beta}$
holds for a certain $\gamma\in[n]\backslash\mathcal{I}$ if and only
if $\mathcal{I}_{\delta}^{\alpha}\cong\mathcal{I}_{\delta}^{\beta}$
holds for all $\delta\in[n]\backslash\mathcal{I}$ such that $\mathcal{I}_{\delta}^{\alpha},\mathcal{I}_{\delta}^{\beta}\in\mathfrak{G}$. 
\end{rem}
\begin{proof}
One implication is trivial. So assume that there exist $\mathcal{I}\in\mathcal{P}_{k}[n]$,
$\alpha,\beta\in\mathcal{I}$ and $\gamma,\delta\in[n]\backslash\mathcal{I}$
such that $\mathcal{I}_{\gamma}^{\alpha}\cong\mathcal{I}_{\gamma}^{\beta}$
and $\mathcal{I}_{\delta}^{\alpha}\cong-\mathcal{I}_{\delta}^{\beta}$
(clearly $\gamma\neq\delta$). The case $0\in\{\mathcal{I}_{\delta}^{\alpha},\mathcal{I}_{\delta}^{\beta}\}$
is ruled out by the assumption $\mathcal{I}_{\delta}^{\alpha},\mathcal{I}_{\delta}^{\beta}\in\mathfrak{G}$.
If $\mathcal{I}_{\gamma}^{\alpha}\cong\mathcal{I}_{\delta}^{\beta}$,
then $\mathcal{I}_{\delta}^{\alpha}\cong\mathcal{I}_{\gamma}^{\beta}$
by Proposition \ref{prop: all three terms or none, including vanishing minors}.
This means that $\mathcal{I}_{\delta}^{\beta}\cong\mathcal{I}_{\gamma}^{\alpha}\cong\mathcal{I}_{\gamma}^{\beta}\cong\mathcal{I}_{\delta}^{\alpha}\cong-\mathcal{I}_{\delta}^{\beta}$,
i.e., a contradiction. Now suppose that $\mathcal{I}_{\gamma}^{\alpha}\cong-\mathcal{I}_{\delta}^{\beta}$,
which means $\mathcal{I}_{\delta}^{\alpha}\cong-\mathcal{I}_{\gamma}^{\beta}$
by Proposition \ref{prop: all three terms or none, including vanishing minors}.
Thus $\mathcal{I}_{\delta}^{\beta}\cong-\mathcal{I}_{\gamma}^{\alpha}\cong-\mathcal{I}_{\gamma}^{\beta}\cong\mathcal{I}_{\delta}^{\alpha}\cong-\mathcal{I}_{\delta}^{\beta}$,
i.e., a contradiction.  
\end{proof}

In fact, the relations (\ref{eq: internal relation, implicit}) are
compatible for different choices of $\mathcal{A}$ too. 
\begin{prop}
\label{prop: compatibility internal equivalences} There are no subsets
$\mathcal{A},\mathcal{B}\in\mathcal{P}_{k}[n]$, $\alpha\in[n]\setminus(\mathcal{A}\cup\mathcal{B})$,
and $\beta\in\mathcal{A}\cap\mathcal{B}$ such that $\mathcal{A},\mathcal{A}_{\alpha}^{\beta},\mathcal{B},\mathcal{B}_{\alpha}^{\beta}\in\mathfrak{G}$
and $\Sigma(\mathcal{A})\cdot\Sigma((\mathcal{A})_{\alpha}^{\beta})=-\Sigma(\mathcal{B})\cdot\Sigma((\mathcal{B})_{\alpha}^{\beta})$. 
\end{prop}
\begin{proof}
Introduce $\mathcal{A}_{1}:=\mathcal{A}$, $\mathcal{B}_{1}:=\mathcal{B}$,
$\mathcal{A}_{2}:=(\mathcal{A})_{\alpha}^{\beta}$, $\mathcal{B}_{2}:=(\mathcal{B})_{\alpha}^{\beta}$.
Note that $\mathcal{A}_{1}\setminus\mathcal{B}_{1}=\mathcal{A}_{2}\setminus\mathcal{B}_{2}$,
$\mathcal{B}_{1}\setminus\mathcal{A}_{1}=\mathcal{B}_{2}\setminus\mathcal{A}_{2}$
and $\alpha,\beta\notin\mathcal{A}_{1}\Delta\mathcal{B}_{1}$.

We shall prove the statement by induction on the distance between
$\mathcal{A}_{1}$ and $\mathcal{B}_{1}$, that is $r:=\frac{1}{2}\cdot\#\left(\mathcal{A}_{1}\Delta\mathcal{B}_{1}\right)$.
First consider the case $r=1$ and set $\mathcal{B}_{1}\setminus\mathcal{A}_{1}=\{\Psi(\gamma_{1})\}$.
In particular, the elements $\alpha,\beta,\gamma_{1},\Psi(\gamma_{1})$
are all distinct and the set $\mathcal{H}:=\mathcal{A}_{1}\cap\mathcal{A}_{2}\cap\mathcal{B}_{1}\cap\mathcal{B}_{2}\in\mathcal{P}_{k-2}[n]$
satisfies $\mathcal{A}_{1}=\mathcal{H}_{\beta\gamma_{1}}$, $\mathcal{A}_{2}=\mathcal{H}_{\alpha\gamma_{1}}$,
$\mathcal{B}_{1}=\mathcal{H}_{\beta\Psi(\gamma_{1})}$ and $\mathcal{B}_{2}=\mathcal{H}_{\alpha\Psi(\gamma_{1})}$.
By hypothesis, $\mathcal{A}_{1},\mathcal{A}_{2},\mathcal{B}_{1},\mathcal{B}_{2}\in\mathfrak{G}$,
so Proposition \ref{prop: all three terms or none, including vanishing minors}
gives $\Sigma(\mathcal{A}_{1})\cdot\Sigma(\mathcal{B}_{2})=\Sigma(\mathcal{A}_{2})\cdot\Sigma(\mathcal{B}_{1})$.
Multiplying both sides of this equation by $\Sigma(\mathcal{A}_{2})\cdot\Sigma(\mathcal{B}_{2})$
and applying $\Sigma(\mathcal{A}_{2})^{2}=\Sigma(\mathcal{B}_{2})^{2}=1$
we find 
\begin{equation}
\Sigma(\mathcal{A}_{1})\cdot\Sigma(\mathcal{A}_{2})=\Sigma(\mathcal{B}_{1})\cdot\Sigma(\mathcal{B}_{2}).\label{eq: closure signs, distance 1}
\end{equation}
Now assume that the statement holds whenever $\frac{1}{2}\cdot\#\left(\mathcal{A}\Delta\mathcal{B}\right)<r$
and take $\mathcal{A}_{1},\mathcal{A}_{2}$, $\mathcal{B}_{1},\mathcal{B}_{2}\in\mathfrak{G}$
with $\frac{1}{2}\cdot\#\left(\mathcal{A}_{1}\Delta\mathcal{B}_{1}\right)=r$.
From Lemma \ref{lem: chain of non-vanishing} there exist a labelling
$\mathcal{A}_{1}\setminus\mathcal{B}_{1}=\{\gamma_{1},\dots,\gamma_{r}\}$
and two bijections $\Psi_{T}:\,\mathcal{A}_{T}\setminus\mathcal{B}_{T}\longrightarrow\mathcal{B}_{T}\setminus\mathcal{A}_{T}$,
$T\in\{1,2\}$, such that 
\begin{equation}
\mathcal{L}_{T}^{(u)}:=\left(\mathcal{A}_{T}\right)_{\Psi_{T}(\gamma_{1}),\dots,\Psi_{T}(\gamma_{u})}^{\gamma_{1}\dots\gamma_{u}}\in\mathfrak{G},\quad u\in[r],\,T\in\{1,2\}.\label{eq: parallel chains of non-vanishing}
\end{equation}
Let us focus on the four sets $\mathcal{A}_{1}$, $\mathcal{A}_{2}$,
$\mathcal{L}_{1}^{(r-1)}$, $\mathcal{L}_{2}^{(r-1)}$: $\mathcal{A}_{1},\mathcal{A}_{2}\in\mathfrak{G}$
by hypothesis, and $\mathcal{L}_{1}^{(r-1)},\mathcal{L}_{2}^{(r-1)}\in\mathfrak{G}$
by construction. Further, one has $\beta\in\left(\mathcal{A}_{1}\cap\mathcal{L}_{1}^{(r-1)}\right)\setminus\left(\mathcal{A}_{2}\cap\mathcal{L}_{2}^{(r-1)}\right)$
and $\alpha\in\left(\mathcal{A}_{2}\cap\mathcal{L}_{2}^{(r-1)}\right)\setminus\left(\mathcal{A}_{1}\cap\mathcal{L}_{1}^{(r-1)}\right)$.
Since the distance between $\mathcal{L}_{T}^{(u)}$ and $\mathcal{A}_{T}$
increases by $1$ at each step, one gets $\frac{1}{2}\cdot\#\left(\mathcal{A}_{T}\Delta\mathcal{L}_{T}^{(r-1)}\right)=r-1$.
So the inductive hypothesis applies and 
\begin{equation}
\Sigma(\mathcal{A}_{1})\cdot\Sigma(\mathcal{L}_{1}^{(r-1)})=\Sigma(\mathcal{A}_{2})\cdot\Sigma(\mathcal{L}_{2}^{(r-1)}).\label{eq: closure signs, distance r, r-1}
\end{equation}
Likewise, $\mathcal{L}_{T}^{(r)}=\mathcal{B}_{T}$ , $\beta\in\left(\mathcal{L}_{1}^{(r-1)}\cap\mathcal{B}_{1}\right)\setminus\left(\mathcal{L}_{2}^{(r-1)}\cap\mathcal{B}_{2}\right)$
and $\alpha\in\left(\mathcal{L}_{2}^{(r-1)}\cap\mathcal{B}_{2}\right)\setminus\left(\mathcal{L}_{1}^{(r-1)}\cap\mathcal{B}_{1}\right)$,
so one can repeat the same argument as in (\ref{eq: closure signs, distance 1})
and get 
\begin{equation}
\Sigma(\mathcal{L}_{1}^{(r-1)})\cdot\Sigma(\mathcal{B}_{1})=\Sigma(\mathcal{L}_{2}^{(r-1)})\cdot\Sigma(\mathcal{B}_{2}).\label{eq: closure signs, distance r, 1}
\end{equation}
Multiplying (\ref{eq: closure signs, distance r, r-1}) and (\ref{eq: closure signs, distance r, 1})
side by side, and applying $\Sigma(\mathcal{L}_{1}^{(r-1)})^{2}=+1$,
one finds 
\begin{equation}
\Sigma(\mathcal{A}_{1})\cdot\Sigma(\mathcal{B}_{1})=\Sigma(\mathcal{A}_{2})\cdot\Sigma(\mathcal{B}_{2})\label{closure signs, distance r}
\end{equation}
that means $\alpha\approx_{\mathcal{A}}\beta$ if and only if $\alpha\approx_{\mathcal{B}}\beta$. 
\end{proof}

Hence, for each $\alpha,\beta\in[n]$, the relation (\ref{eq: internal relation, implicit})
depends only on $\alpha$ and $\beta$, and not on $\mathcal{I}$
such that $\mathcal{I},\mathcal{I}_{\alpha}^{\beta}\in\mathfrak{G}$.
If such a $\mathcal{I}$ exists, we can introduce 
\begin{equation}
\chi(\alpha,\beta)=\Sigma(\mathcal{I})\cdot\Sigma(\mathcal{I}_{\alpha}^{\beta})\label{eq: relative sign between columns}
\end{equation}
without ambiguity. Moreover, $\chi$ is symmetric in its arguments,
since $\Sigma(\mathcal{I})\cdot\Sigma(\mathcal{I}_{\alpha}^{\beta})=\Sigma(\mathcal{I}_{\alpha}^{\beta})\cdot\Sigma((\mathcal{I}_{\beta}^{\alpha})_{\alpha}^{\beta})$
and one can substitute $\mathcal{I}\mapsto\mathcal{I}_{\beta}^{\alpha}$
in (\ref{eq: relative sign between columns}).

We will see that the signs (\ref{eq: relative sign between columns})
characterize $\Sigma$, so we focus on these pairs. 
\begin{defn}
\label{def: paths X matrix} For each $\mathcal{I}:=\{\beta_{1},\dots,\beta_{k}\}\in\mathfrak{G}$
take the $k\times(n-k)$ matrix $\mathbf{X}(\mathcal{I})$ defined
as $\left(\mathbf{X}(\mathcal{I})\right)_{i\alpha}:=\Sigma(\mathcal{I})\cdot\Sigma(\mathcal{I}_{\alpha}^{\beta_{i}})$
if $\mathcal{I}_{\alpha}^{\beta_{i}}\in\mathfrak{G}$ and $0$ otherwise
($\alpha\in[n]\setminus\mathcal{I}$). We refer to a \textit{path}
$\Phi_{\mathcal{I}}$ on $\mathbf{X}(\mathcal{I})$ as a sequence
of non-vanishing entries obtained moving alternately along rows and
columns of $\mathbf{X}(\mathcal{I})$, i.e., of the form 
\begin{equation}
(i_{1},\alpha_{1})\rightarrow(i_{2},\alpha_{1})\rightarrow(i_{2},\alpha_{2})\rightarrow(i_{3},\alpha_{2})\rightarrow\dots\label{eq: path}
\end{equation}
with $i_{T}\neq i_{T+1}$, $\alpha_{T}\neq \alpha_{T+1}$ and $\mathcal{I}^{i_{T}}_{\alpha_{T}},\mathcal{I}^{i_{T+1}}_{\alpha_{T}}\in\mathfrak{G}$
for all $T$.
Two indices $\alpha,\beta\in[n]$ are said to be \textit{connected
by a path} $\Phi_{\mathcal{I}}$ if $\alpha$ is a component of the
first element of $\Phi_{\mathcal{I}}$ but not of the second, and
$\beta$ is a component of the last element of $\Phi_{\mathcal{I}}$
but not of the second-to-last. 
\end{defn}
For instance, if $\alpha\in\mathcal{I}$ and $\beta\in[n]\setminus\mathcal{I}$,
then a path from $\alpha$ to $\beta$ starts with $(\alpha,\gamma)$
and ends with $(\delta,\beta)$, for some $\gamma\in[n]\setminus\mathcal{I},\delta\in\mathcal{I}$;
if instead both $\alpha$ and $\beta$ are elements of $\mathcal{I}$,
a path starts with $(\alpha,\gamma)$ and ends with $(\beta,\delta)$,
$\gamma,\delta\in[n]\setminus\mathcal{I}$. These paths induce a relation
$\rightarrow_{\mathcal{I}}$ on $[n]$. 
\begin{defn}
\label{def: path connection relation} For all $\alpha,\beta\in[n]$,
we say $\alpha\rightarrow_{\mathcal{I}}\beta$ if $\alpha=\beta$
or there is a path in $\mathbf{X}(\mathcal{I})$ that connects $\alpha$
and $\beta$. 
\end{defn}
This is an equivalence relation: it is reflexive by definition; if
$\alpha\rightarrow_{\mathcal{I}}\beta$, then the reverse path gives
$\beta\rightarrow_{\mathcal{I}}\alpha$, so $\rightarrow_{\mathcal{I}}$
is symmetric; it is also transitive, as follows from the path $\alpha\rightarrow_{\mathcal{I}}\gamma$
obtained from the concatenation of $\alpha\rightarrow_{\mathcal{I}}\beta$
and $\beta\rightarrow_{\mathcal{I}}\gamma$ and the simplification
of consecutive reverse subpaths. Moreover, the following result holds: 
\begin{prop}
\label{prop: closure property signs along closed path} The product
of signs of edges $\chi$ along any closed path of $\mathbf{X}(\mathcal{I})$,
$\mathcal{I}\in\mathfrak{G}$, is $+1$. 
\end{prop}
\begin{proof}
Let $\mathcal{I}:=\{\gamma_{1},\dots,\gamma_{k}\}$ and denote $\gamma_{i_{T}}$
by $i_{T}$ in the rest of the proof for the sake of clearness. The
same number of moves along rows and along columns is required to close
a path, so it has an even length $2r$. Thus, the product of signs
along the closed path is 
\begin{equation}
\Phi_{\mathcal{I}}(i_{1},\dots,i_{r}\mid\alpha_{1},\dots,\alpha_{r}):=\prod_{T=1}^{r}\Sigma(\mathcal{I}_{\alpha_{T}}^{i_{T}})\cdot\Sigma(\mathcal{I}_{\alpha_{T}}^{i_{T+1}})\label{eq: sign of the path}
\end{equation}
where indices $T$ are taken modulo $r$, e.g., $i_{r+1}=i_{1}$. From
this one can also see that the statement is invariant under the replacement
of $\mathbf{X}(\mathcal{I})$ by $\mathbf{X}(\mathcal{J})\in\mathfrak{G}$,
as long as all the involved subsets $\mathcal{J}_{\alpha_{T}}^{i_{T}}$
and $\mathcal{J}_{\alpha_{T}}^{i_{T+1}}$ are in $\mathfrak{G}$:
indeed, both sides of (\ref{eq: sign of the path}) can be multiplied
by $(\Sigma(\mathcal{I}))^{2r}\cdot(\Sigma(\mathcal{J}))^{2r}=+1$,
so the factors are expressed as $\Sigma(\mathcal{I}_{\alpha_{T}}^{i_{T}})\cdot\Sigma(\mathcal{I})=\Sigma(\mathcal{J}_{\alpha_{T}}^{i_{T}})\cdot\Sigma(\mathcal{J})$
(respectively $\Sigma(\mathcal{I}_{\alpha_{T}}^{i_{T+1}})\cdot\Sigma(\mathcal{I})=\Sigma(\mathcal{J}_{\alpha_{T}}^{i_{T+1}})\cdot\Sigma(\mathcal{J})$)
by Proposition \ref{prop: compatibility internal equivalences}.

We prove the statement by induction on $r$. The base cases are $r=1$,
which is trivial by symmetry of $\chi(\alpha,\gamma_{i})$ with respect
to the interchange of its arguments, and $r=2$. In the latter situation,
there exist $i\neq m$ and $\alpha\neq\beta$ such that $\mathcal{I}_{\alpha}^{\gamma_{i}},\mathcal{I}_{\alpha}^{\gamma_{m}},\mathcal{I}_{\beta}^{\gamma_{m}},\mathcal{I}_{\beta}^{\gamma_{i}}\in\mathfrak{G}$.
By Proposition \ref{prop: all three terms or none, including vanishing minors}
the associated product of signs is $+1$. Now assume that the thesis
holds for $s\leq r-1$ and take any path of length $2r$. First suppose
that there exist $T\in[r]$, $S\in[r-2]$, such that $\mathcal{I}_{\alpha_{T+S}}^{i_{T}}\in\mathfrak{G}$.
Then one can write 
\begin{eqnarray}
& & \Phi_{\mathcal{I}}(i_{1},\dots,i_{r}\mid\alpha_{1},\dots,\alpha_{r})\nonumber \\ 
& = & \left[\left(\prod_{W=0}^{S-1}\Sigma(\mathcal{I}_{\alpha_{T+W}}^{i_{T+W}})\cdot\Sigma(\mathcal{I}_{\alpha_{T+W}}^{i_{T+W+1}})\right)\cdot\Sigma(\mathcal{I}_{\alpha_{T+S}}^{i_{T+S}})\cdot\Sigma(\mathcal{I}_{\alpha_{T+S}}^{i_{T}})\right]\nonumber \\
& \cdot & \left[\Sigma(\mathcal{I}_{\alpha_{T+S}}^{i_{T}})\cdot\left(\prod_{W=S}^{r-2}\Sigma(\mathcal{I}_{\alpha_{T+W}}^{i_{T+W+1}})\cdot\Sigma(\mathcal{I}_{\alpha_{T+W+1}}^{i_{T+W+1}})\right)\cdot\Sigma(\mathcal{I}_{\alpha_{T+r-1}}^{i_{T}})\right]\nonumber \\
& = & \Phi_{\mathcal{I}}(i_{T},i_{T+1},\dots,i_{T+S}\mid\alpha_{T},\alpha_{T+1},\dots,\alpha_{T+S})\nonumber \\
& \cdot & \Phi_{\mathcal{I}}(i_{T},i_{T+S+1},i_{T+S+2},\dots,i_{T+r-1}\mid\alpha_{T+S},\alpha_{T+S+1},\dots,\alpha_{T+r-1}).\nonumber \\
\label{eq: case reduction (t,t+1) along path, b}
\end{eqnarray}
The lengths of these two closed path are $2S+2$ and $2(r-S)$ respectively,
which lie in $\{4,\dots,2r-2\}$; so the inductive hypothesis applies
to both these paths and this gives the result.

On the other hand, if such $T,S$ do not exist, then $\mathcal{I}_{\alpha_{T}}^{i_{T+S}}\in\mathcal{P}_{k}[n]\setminus\mathfrak{G}$
for all $T\in[r]$ and $S\notin\{1,r\}$: by Proposition \ref{prop: all three terms or none, including vanishing minors},
the condition $\mathcal{I}_{\alpha_{T}}^{i_{T}},\mathcal{I}_{\alpha_{T+S}}^{i_{T+S}}\in\mathfrak{G}$,
along with $\mathcal{I}_{\alpha_{T+S}}^{i_{T}}\in\mathcal{P}_{k}[n]\setminus\mathfrak{G}$
at $S\notin\{r-1,r\}$ and $\mathcal{I}_{\alpha_{T}}^{i_{T+S}}\in\mathcal{P}_{k}[n]\setminus\mathfrak{G}$
at $S\notin\{1,r\}$, implies that $\mathcal{I}_{\alpha_{T}\alpha_{T+S}}^{i_{T}i_{T+S}}\in\mathfrak{G}$
for all $S\neq r$. Likewise, from $\mathcal{I}_{\alpha_{T}}^{i_{T+1}},\mathcal{I}_{\alpha_{T+S}}^{i_{T+S+1}}\in\mathfrak{G}$,
$\mathcal{I}_{\alpha_{T+S}}^{i_{T+1}}\in\mathcal{P}_{k}[n]\setminus\mathfrak{G}$
($S\notin\{1,r\}$) and $\mathcal{I}_{\alpha_{T}}^{i_{T+S+1}}\in\mathcal{P}_{k}[n]\setminus\mathfrak{G}$
($S\notin\{r-1,r\}$), one gets $\mathcal{I}_{\alpha_{T}\alpha_{T+S}}^{i_{T+1}i_{T+S+1}}\in\mathfrak{G}$
whenever $S\neq r$. Further, $\mathcal{I}_{\alpha_{2}}^{i_{2}},\mathcal{I}_{\alpha_{T}}^{i_{T+1}}\in\mathfrak{G}$
and $\mathcal{I}_{\alpha_{2}}^{i_{T+1}}\in\mathcal{P}_{k}[n]\setminus\mathfrak{G}$
lead to $\mathcal{I}_{\alpha_{2}\alpha_{T}}^{i_{2}i_{T+1}}\in\mathfrak{G}$
for all $T\notin\{1,2\}$.

So we fix $\mathcal{J}:=\mathcal{I}_{\alpha_{2}}^{i_{2}}\in\mathfrak{G}$:
by previous observations, one has $\mathcal{J}_{\alpha_{T}}^{i_{T}},\mathcal{J}_{\alpha_{T}}^{i_{T+1}}\in\mathfrak{G}$
for all $T\notin\{1,2\}$, $\mathcal{J}_{\alpha_{1}}^{i_{1}}\in\mathfrak{G}$
and $\mathcal{J}_{\alpha_{1}}^{i_{3}}=\mathcal{I}_{\alpha_{1}\alpha_{2}}^{i_{2}i_{3}}\in\mathfrak{G}$.
One can apply Proposition \ref{prop: all three terms or none, including vanishing minors}
to $\mathcal{J}_{\alpha_{1}}^{i_{1}}$ and $\left(\mathcal{J}_{\alpha_{1}}^{i_{1}}\right)_{i_{1}}^{i_{3}}=\mathcal{J}_{\alpha_{1}}^{i_{3}}$,
which are in $\mathfrak{G}$, and write 
\begin{eqnarray}
\Sigma\left(\mathcal{J}_{\alpha_{1}}^{i_{1}}\right)\cdot\Sigma\left(\mathcal{J}_{\alpha_{1}}^{i_{3}}\right) & = & \Sigma\left(\mathcal{I}_{\alpha_{1}\alpha_{2}}^{i_{1}i_{2}}\right)\cdot\Sigma\left(\mathcal{I}_{\alpha_{1}\alpha_{2}}^{i_{2}i_{3}}\right)\cdot\Sigma(\mathcal{I})^{2}\nonumber \\
 & = & \Sigma\left(\mathcal{I}_{\alpha_{1}}^{i_{1}}\right)\cdot\Sigma\left(\mathcal{I}_{\alpha_{2}}^{i_{2}}\right)\cdot\Sigma\left(\mathcal{I}_{\alpha_{1}}^{i_{2}}\right)\cdot\Sigma\left(\mathcal{I}_{\alpha_{2}}^{i_{3}}\right)\label{eq: transitivity T, T+1, T+2}
\end{eqnarray}
Thus consider the path $\Phi_{\mathcal{J}}(i_{1},i_{3},i_{4},\dots,i_{r}|\alpha_{1},\alpha_{3},\alpha_{4},\dots,\alpha_{r})$
of length $2(r-1)$ on $\mathbf{X}(\mathcal{J})$, which also satisfies
$\Sigma(\mathcal{J}_{\alpha_{T}}^{i_{T}})\cdot\Sigma(\mathcal{J}_{\alpha_{T}}^{i_{T+1}})=\Sigma(\mathcal{I}_{\alpha_{T}}^{i_{T}})\cdot\Sigma(\mathcal{I}_{\alpha_{T}}^{i_{T+1}})$
for all $T\notin\{1,2\}$ (all the subsets are in $\mathfrak{G}$
and Proposition \ref{prop: compatibility internal equivalences} holds)
and $\Sigma(\mathcal{J}_{\alpha_{1}}^{i_{1}})\cdot\Sigma(\mathcal{J}_{\alpha_{1}}^{i_{3}})=\Sigma(\mathcal{I}_{\alpha_{1}}^{i_{1}})\cdot\Sigma(\mathcal{I}_{\alpha_{1}}^{i_{2}})\cdot\Sigma(\mathcal{I}_{\alpha_{2}}^{i_{2}})\cdot\Sigma(\mathcal{I}_{\alpha_{2}}^{i_{3}})$
by (\ref{eq: transitivity T, T+1, T+2}). Since the involved subsets
are in $\mathfrak{G}$, the inductive hypothesis applies to the path
in $\mathbf{X}(\mathcal{J})$ and that gives 
\begin{equation}
\Phi_{\mathcal{I}}(i_{1},\dots,i_{r}\mid\alpha_{1},\dots,\alpha_{r})=\Phi_{\mathcal{J}}(i_{1},i_{3},\dots,i_{r}\mid\alpha_{1},\alpha_{3},\dots,\alpha_{r})=+1.\label{eq: closed path sign +}
\end{equation} 
\end{proof}

A special role is assumed by the pivot set $\mathcal{V}$, since the
hypothesis of no-null columns implies that each $\alpha\in[n]$ is
associated with an element $\nu_{i}\in\mathcal{V}$ (possibly $\alpha=\nu_{i}$)
such that $\mathcal{V}_{\alpha}^{\nu_{i}}\in\mathfrak{G}$. So we
can finally state the main result: 
\begin{thm}
\label{thm: compatible choices of signs, solution KP, including vanishing}
A choice of signs $\Sigma:\,\mathfrak{G}\longrightarrow\{\pm1\}$
returns a solution of the KP II equation (\ref{eq: bilinear KP})
if and only if $\Sigma$ is induced by a choice of signs for rows
and columns of $\mathbf{A}$ (up to the action of $GL_{k}(\mathbb{R})$). 
\end{thm}
\begin{proof}
One implication is trivial, since a choice of signs for rows and columns
of $\mathbf{A}$ induces a (singular) soliton solution by construction.
So consider any choice of signs $\Sigma$ that returns a solution
of the KP II equation and choose an arbitrary element in $\mathfrak{G}$,
e.g., $\mathcal{V}$. For each $\mathcal{I}\in\mathfrak{G}$ and $\alpha\in\mathcal{I}\setminus\mathcal{V}$
, given that $\Delta_{\mathbf{A}}(\mathcal{I})\cdot\Delta_{\mathbf{A}}(\mathcal{V})\neq0$,
the Pl\"{u}cker relations imply that there exists at least one non-vanishing
term of the type $\Delta_{\mathbf{A}}(\mathcal{I}_{\nu_{i}}^{\alpha})\cdot\Delta_{\mathbf{A}}(\mathcal{V}_{\alpha}^{\nu_{i}})$.
Thus one can always find $\nu(\alpha)\in\mathcal{V}\setminus\mathcal{I}$
such that both $\mathcal{I}_{\nu(\alpha)}^{\alpha}$ and $\mathcal{V}_{\alpha}^{\nu(\alpha)}$
are elements of $\mathfrak{G}$. So let $\mathcal{I}\setminus\mathcal{V}=\{\alpha_{1},\dots,\alpha_{r}\}$:
as in Lemma \ref{lem: chain of non-vanishing}, we start with $\mathcal{L}_{0}:=\mathcal{I}$
and, from $\mathcal{L}_{u-1}\in\mathfrak{G}$ and $\alpha_{u}\in\mathcal{L}_{u-1}\setminus\mathcal{V}$,
$u\in[r]$, we find 
\begin{equation}
\nu(\alpha_{u})\in\mathcal{V}\setminus\mathcal{L}_{u-1},\quad\mathcal{L}_{u}:=\mathcal{L}_{u-1}\backslash\{\alpha_{u}\}\cup\{\nu(\alpha_{u})\}\label{eq: representatives in pivot set}
\end{equation}
such that $\mathcal{L}_{u},\mathcal{V}_{\alpha_{u}}^{\nu(\alpha_{u})}\in\mathfrak{G}$.
Also in this case, from $\nu(\alpha_{s})\in\mathcal{L}_{u-1}$ for
all $s<u$ and $\nu(\alpha_{u})\in\mathcal{V}\setminus\mathcal{L}_{u-1}$,
it follows that all $\nu(\alpha_{u})$ are pairwise distinct. So we
get 
\begin{eqnarray}
\Sigma(\mathcal{I})\cdot\Sigma(\mathcal{V}) & = & \prod_{u=1}^{r}\Sigma(\mathcal{L}_{u-1})\cdot\Sigma(\mathcal{L}_{u})=\prod_{u=1}^{r}\Sigma(\mathcal{L}_{u-1})\cdot\Sigma\left((\mathcal{L}_{u-1})_{\nu(\alpha_{u})}^{\alpha_{u}}\right)\nonumber \\
 & = & \prod_{u=1}^{r}\Sigma(\mathcal{V}_{\alpha_{u}}^{\nu(\alpha_{u})})\cdot\Sigma\left((\mathcal{V}_{\alpha_{u}}^{\nu(\alpha_{u})})_{\nu(\alpha_{u})}^{\alpha_{u}}\right)\quad\mathrm{(from\,\,Proposition\,\,\ref{prop: compatibility internal equivalences})}\nonumber \\
 & = & \prod_{u=1}^{r}\Sigma(\mathcal{V}_{\alpha_{u}}^{\nu(\alpha_{u})})\cdot\Sigma(\mathcal{V})=\prod_{u=1}^{r}\chi(\alpha_{i},\nu(\alpha_{i})).\label{eq: factorization}
\end{eqnarray}

Now consider the equivalence $\rightarrow_{\mathcal{V}}$. Each class
$\mathcal{C}_{p}$ contains at least one element $\nu_{i_{p}}\in\mathcal{V}$,
since we have assumed that there are no null columns: fix a sign $\chi(\nu_{i_{p}})\in\{\pm1\}$
for each of them. For any $\mathcal{\alpha\in\mathcal{C}}_{p}$, take
a path $\Phi(i_{p}\rightarrow\alpha)$ on $\mathbf{X}(\mathcal{V})$
connecting $\nu_{i_{p}}$ and $\alpha$, and set 
\begin{equation}
\chi(\alpha):=\chi(i_{p})\cdot\left(\prod_{(m,\delta)\in\Phi(i_{p}\rightarrow\alpha)}\chi(\delta,\nu_{m})\right).\label{eq: construction of signs}
\end{equation}
This definition is well-posed since it does not depend on the choice
of the path by Proposition \ref{prop: closure property signs along closed path}.
If $\mathcal{V}_{\alpha}^{\nu_{i}}\in\mathfrak{G}$, then $\nu_{i}\in\mathcal{V}$
belongs to the same class of $\alpha$, because the path with only
one element $(\nu_{i},\alpha)$ connects them. The concatenation of
$\{(\nu_{i},\alpha)\}$, the reverse of $\Phi(i_{p}\rightarrow\alpha)$
and $\Phi(i_{p}\rightarrow\nu_{i})$ makes a closed path, whose product
of signs is equal to $+1$ by Proposition \ref{prop: closure property signs along closed path}.
So 
\begin{eqnarray}
\hspace*{-1cm}\chi(\alpha,\nu_{i}) & = & \prod_{(m,\delta)\in\Phi(i_{p}\rightarrow\alpha)}\chi(\delta,\nu_{m})\cdot\prod_{(l,\gamma)\in\Phi(i_{p}\rightarrow\nu_{i})}\chi(\gamma,\nu_{l})\nonumber \\
\hspace*{-1cm} & = & \left(\chi(i_{p})\cdot\prod_{(m,\delta)\in\Phi(i_{p}\rightarrow\alpha)}\chi(\delta,\nu_{m})\right)\cdot\left(\chi(i_{p})\cdot\prod_{(l,\gamma)\in\Phi(i_{p}\rightarrow\nu_{i})}\chi(\gamma,\nu_{l})\right)\nonumber \\
\hspace*{-1cm} & = & \chi(\alpha)\cdot\chi(\nu_{i}).\label{eq: factorization 2 to 1}
\end{eqnarray}
Finally, since $\mathcal{V}\setminus\mathcal{I}=\{\nu(\alpha_{1}),\dots,\nu(\alpha_{r})\}$,
we can express $\Sigma(\mathcal{I})$ as 
\begin{eqnarray}
\hspace*{-1cm}\Sigma(\mathcal{I}) & = & \Sigma(\mathcal{V})\cdot\prod_{u=1}^{r}\chi(\alpha_{u},\nu(\alpha_{u}))\quad\mathrm{(from\,\,(\ref{eq: factorization}))}\nonumber \\
\hspace*{-1cm} & = & \Sigma(\mathcal{V})\cdot\prod_{u=1}^{r}\chi(\alpha_{u})\cdot\chi(\nu(\alpha_{u}))\quad\mathrm{(from\,\,(\ref{eq: factorization 2 to 1}))}\nonumber \\
\hspace*{-1cm} & = & \Sigma(\mathcal{V})\cdot\left(\prod_{i=1}^{k}\chi(\nu_{i})\right)\cdot\left(\prod_{i=1}^{k}\chi(\nu_{i})\right)\cdot\left(\prod_{u=1}^{r}\chi(\nu(\alpha_{u}))\right)\cdot\left(\prod_{u=1}^{r}\chi(\alpha_{u})\right)\nonumber \\
\hspace*{-1cm} & = & R\cdot\left(\prod_{\nu_{i}\in\mathcal{I}\cap\mathcal{V}}\chi(\nu_{i})\right)\cdot\left(\prod_{u=1}^{r}\chi(\alpha_{u})\right)=R\cdot\prod_{\alpha\in\mathcal{I}}\chi(\alpha).\label{eq: reduction to row and column operations}
\end{eqnarray}
where 
\begin{equation}
R:=\Sigma(\mathcal{V})\cdot\left(\prod_{i=1}^{k}\chi(\nu_{i})\right).\label{eq: row normalization}
\end{equation}
Hence $\Sigma$ is induced by a choice of sign $\chi(\alpha)$ for
columns $\alpha\in[n]$ and $R$ for an arbitrary row of $\mathbf{A}$.
\end{proof}

The previous results can be summarised in the following theorem,
which relates the requirements coming from the determinantal form
(\ref{eq: determinantal partition function}), the KP II equation
and the whole KP hierarchy. 
\begin{thm}
\label{thm: connections KP eq, KP hierarchy, and determinantal form} For
a generic choice of $\boldsymbol{\ensuremath{\kappa}}$, a signature (\ref{eq: choice of signs})
returns a solution $\tau(\boldsymbol{x})-2\cdot\tau_{\Sigma}(\boldsymbol{x})$
of the KP II equation (\ref{eq: bilinear KP}) if and only if it returns
a solution $\tau(x_{1},x_{2},x_{3},\dots)-2\cdot\tau_{\Sigma}(x_{1},x_{2},x_{3},\dots)$
for the whole KP hierarchy, if and only if $\Sigma$ preserves the determinant form (\ref{eq: Cauchy-Binet formula, solitons, a}). 
\end{thm}
\begin{proof}
A choice of signs (\ref{eq: choice of signs}) returns a solution
of the KP II equation (\ref{eq: bilinear KP}) if and only if it is
induced by a choice of signs for rows and columns of $\mathbf{A}$
by Theorem \ref{thm: compatible choices of signs, solution KP, including vanishing}.
In this case, the associated function $\tau(x_{1},x_{2},x_{3},\dots)-2\cdot\tau_{\Sigma}(x_{1},x_{2},x_{3},\dots)$ is a solution of the whole KP hierarchy. 

If we focus on the determinantal structure of the $\tau$-function, all the choices of signs for rows and columns clearly preserve the form (\ref{eq: Cauchy-Binet formula, solitons, a}). On the other hand, if a signature $\Sigma$ preserve this structure, then the KP II equation is satisfied (together with all the other members of the hierarchy). So, from Theorem \ref{thm: compatible choices of signs, solution KP, including vanishing}, this solution can be expressed in terms of the initial function via a choice of signs for rows and columns of $\mathbf{A}$.  
\end{proof} 
A more detailed analysis on the signatures preserving the determinantal constraints, which also includes the preservation of a specific subset of soliton parameters, is given in Appendix \ref{sec: Determinantal constraints}. 

A special situation is when $\Sigma$ is defined over the whole set
$\mathcal{P}_{k}[n]$, i.e., $\mathfrak{G}=\mathcal{P}_{k}[n]$. In
such a case, one can easily verify that Remark \ref{rem: internal relation}
implies that $\approx_{\mathcal{I}}$ is an equivalence. If some vanishing
minors occur, then the transitivity of the relation (\ref{eq: internal relation, implicit})
is not guaranteed. So one could look for a transitive extension of
all the relations $\approx_{\mathcal{I}}$ at varying $\mathcal{I}$.
An \textit{extension} of $\Sigma$, that
is a map $\tilde{\Sigma}:\,\mathcal{P}_{k}[n]\longrightarrow\{\pm1\}$
that satisfies $\tilde{\Sigma}(\mathcal{I})=\Sigma(\mathcal{I})$
for all $\mathcal{I}\in\mathfrak{G}$, can be obtained using (\ref{eq: construction of signs})
to label minors in $\mathcal{P}_{k}[n]\setminus\mathfrak{G}$ with
a sign compatible with Proposition \ref{prop: all three terms or none, including vanishing minors}.
So the following holds: 
\begin{cor}
\label{cor: integrable complete extension} For a solitonic choice
of sign $\Sigma$, there exists an extension $\tilde{\Sigma}$ of
$\Sigma$ to the whole set $\mathcal{P}_{k}[n]$ such that $\tilde{\Sigma}$
is a solitonic choice of signs too. 
\end{cor}
It could be interesting to extend this approach to more general expressions
for the $\tau$-function and to other hierarchies, in order to check
the complexity reduction coming from the initial data and the specific
requirements.

\section{\label{sec: Number of distinct configurations} Number of distinct
configurations}

Having identified a family of signatures that preserve specific requirements,
it is worth exploring some of its combinatorial aspects in order to
clarify the effects of the initial data, e.g., the coefficient matrix
$\mathbf{A}$, on the set of allowed configurations.

The overall sign given by row signature is obtained with the choice
$R\in\{\pm1\}$ for an arbitrary row of $\mathbf{A}$ (see (\ref{eq: row normalization})).
Sign flips for columns can be expressed as an action of $\{\pm1\}^{[n]}$
on $\mathbb{R}^{k\times n}$ by right multiplication, i.e., $\boldsymbol{\ensuremath{\sigma}}_{\alpha}(\mathbf{A}):=\mathbf{A}\cdot\boldsymbol{\ensuremath{\sigma}}_{\alpha}$
and $\boldsymbol{\ensuremath{\sigma}}_{\mathcal{I}}(\mathbf{A}):=\mathbf{A}\cdot\boldsymbol{\ensuremath{\sigma}}_{\mathcal{I}}$
where 
\begin{eqnarray}
\boldsymbol{\ensuremath{\sigma}}_{\alpha} & := & \left(
\begin{array}{ccc}
\idd_{\alpha-1} & \mathbf{0} & \mathbf{0}\\
\mathbf{0} & -1 & \mathbf{0}\\
\mathbf{0} & \mathbf{0} & \idd_{n-\alpha}
\end{array}\right),\quad\alpha\in[n],\label{eq: sign switch operator, a}\\
\boldsymbol{\ensuremath{\sigma}}_{\mathcal{S}} & := & \prod_{\alpha\in\mathcal{S}}\boldsymbol{\ensuremath{\sigma}}_{\alpha},\quad\mathcal{S}\in\mathcal{P}[n].\label{eq: sign switch operator, A}
\end{eqnarray}
Accordingly, (\ref{eq: Cauchy-Binet as partition function}) becomes
\begin{equation}
\tau_{\mathcal{S}}(\boldsymbol{x}):=\det\left(\boldsymbol{\ensuremath{\sigma}}_{\mathcal{S}}(\mathbf{A})\cdot\boldsymbol{\ensuremath{\Theta}}(\boldsymbol{x})\cdot\mathbf{K}\right).\label{eq: tau function induced by column operations}
\end{equation}

If $k=1$, then the minors of $\boldsymbol{A}\in\mathbb{R}^{1\times n}$
are the entries of a row vector. Then, for every non-zero matrix $\boldsymbol{B}\in\mathbb{R}^{1\times n}$
such that $|A_{\alpha}|=|B_{\alpha}|$ for all $\alpha\in[N]$, one
can recover $\boldsymbol{A}$ from (\ref{eq: sign switch operator, A})
choosing $\mathcal{S}=\left\{ \alpha\in[n]:\,A_{1,\alpha}\cdot B_{1,\alpha}<0\right\} $.
If $\mathfrak{G}=\mathcal{P}_{k}[n]$, such a choice is unique. The
use of (\ref{eq: sign switch operator, A}) in the case $k>1$ gives
a suitable generalization. 
\begin{lem}
\label{lem: cluster classes} Let $\mathbf{A}\in\mathbb{R}^{k\times n}$
and $\mathcal{C}_{1},\dots,\mathcal{C}_{P}$ be the equivalence classes
associated with the relation $\rightarrow_{\mathcal{V}}$ in Definition
\ref{def: path connection relation}. For all $\mathcal{H}_{1},\mathcal{H}_{2}\subseteq[n]$,
the equalities 
\begin{equation}
\mathrm{sign}\left[\Delta\left(\boldsymbol{\ensuremath{\sigma}}_{\mathcal{H}_{1}}(\mathbf{A});\mathcal{I}\right)\right]=\mathrm{sign}\left[\Delta\left(\boldsymbol{\ensuremath{\sigma}}_{\mathcal{H}_{2}}(\mathbf{A});\mathcal{I}\right)\right],\quad\mathcal{I}\in\mathfrak{G}\label{eq: hp equal signs minors, +}
\end{equation}
imply that there exists $\mathfrak{p}\subseteq[P]$ such that 
\begin{equation}
\mathcal{H}_{1}\Delta\mathcal{H}_{2}=\bigcup_{q\in\mathfrak{p}}\mathcal{C}_{q}.\label{eq: clustering equivalence classes same signature, +}
\end{equation}
\end{lem}
\begin{proof}
First note that the condition (\ref{eq: hp equal signs minors, +})
holds for $\mathbf{A}$ if and only if it holds for $\boldsymbol{\ensuremath{\sigma}}_{\mathcal{H}_{2}}(\mathbf{A})$:
indeed, for all $\mathcal{S}\subseteq[n]$, $\mathbf{A}$ and $\boldsymbol{\ensuremath{\sigma}}_{\mathcal{S}}(\mathbf{A})$
have the same set of non-vanishing maximal minors and, in particular,
the action of $\boldsymbol{\ensuremath{\sigma}}_{\mathcal{H}_{2}}$ does
not affect the relation $\rightarrow_{\mathcal{V}}$ relative to the
pivot set $\mathcal{V}$. In the substitution $\mathbf{A}\mapsto\boldsymbol{\ensuremath{\sigma}}_{\mathcal{H}_{2}}(\mathbf{A})$,
the additional signs coming from $\boldsymbol{\ensuremath{\sigma}}_{\alpha}$,
$\alpha\in\mathcal{H}_{2}$, appear on both sides of (\ref{eq: hp equal signs minors, +}),
then 
\begin{eqnarray}
\boldsymbol{\ensuremath{\sigma}}_{\mathcal{H}_{1}}(\mathbf{A}) & \mapsto & \boldsymbol{\ensuremath{\sigma}}_{\mathcal{H}_{1}}\left(\boldsymbol{\ensuremath{\sigma}}_{\mathcal{H}_{2}}(\mathbf{A})\right)=(\boldsymbol{\ensuremath{\sigma}}_{\mathcal{H}_{1}\cap\mathcal{H}_{2}})^{2}\left(\boldsymbol{\ensuremath{\sigma}}_{\mathcal{H}_{1}\Delta\mathcal{H}_{2}}(\mathbf{A})\right)=\boldsymbol{\ensuremath{\sigma}}_{\mathcal{H}_{1}\Delta\mathcal{H}_{2}}(\mathbf{A}),\nonumber \\
\boldsymbol{\ensuremath{\sigma}}_{\mathcal{H}_{2}}(\mathbf{A}) & \mapsto & (\boldsymbol{\ensuremath{\sigma}}_{\mathcal{H}_{2}})^{2}(\mathbf{A})=\mathbf{A}\label{eq: switching to disentangle}
\end{eqnarray}
since $\boldsymbol{\ensuremath{\sigma}}_{\mathcal{H}_{1}\cap\mathcal{H}_{2}}^{2}=\boldsymbol{\ensuremath{\sigma}}_{\mathcal{H}_{2}}^{2}=\idd_{n}$.
So (\ref{eq: hp equal signs minors, +}) holds for the pair $(\mathcal{H}_{1},\mathcal{H}_{2})$
if and only if it holds for $(\mathcal{H}_{1}\Delta\mathcal{H}_{2},\emptyset)$.

Take any class $\mathcal{C}_{q}$, $q\in[P]$, and two elements of
$\alpha,\beta\in\mathcal{C}_{q}$. Then, there exists a path $\Phi_{\mathcal{V}}(\alpha\rightarrow\beta)$
in $\mathbf{X}(\mathcal{V})$, which consists of a chain of pairs
$(\nu_{i_{T}},\gamma_{T})$, $\nu_{i_{T}},\gamma_{T}\in\mathcal{C}_{q}$,
associated with subsets $\mathcal{V}_{\gamma_{T}}^{\nu_{i_{T}}}\in\mathfrak{G}$.
If the condition (\ref{eq: hp equal signs minors, +}) holds for such
subsets, one gets 
\[
\nu_{i_{T}}\in\mathcal{H}_{1}\Delta\mathcal{H}_{2}\Leftrightarrow\gamma_{T}\in\mathcal{H}_{1}\Delta\mathcal{H}_{2}
\]
for all $T$. The concatenation of all these equivalences for all
the pairs in the path gives the implication 
\begin{equation}
(\ref{eq: hp equal signs minors, +})\Rightarrow(\alpha\in\mathcal{H}_{1}\Delta\mathcal{H}_{2}\Leftrightarrow\beta\in\mathcal{H}_{1}\Delta\mathcal{H}_{2})\label{eq: clustering equivalence classes same signature, +, bis}
\end{equation}
for all $q\in[P]$ and $\alpha,\beta\in\mathcal{C}_{q}$, which is
equivalent to (\ref{eq: clustering equivalence classes same signature, +}).
\end{proof}

Thus the redundancy in the representation of signatures by subsets
of $[n]$ is due to the elements of 
\begin{equation}
\left\{ \bigcup_{q\in\mathfrak{p}}\mathcal{C}_{q},\,\mathfrak{p}\subseteq[P]\right\} .\label{eq: family of conjugated signatures}
\end{equation}
Note that $\bigcup_{q\in\mathfrak{p}}\mathcal{C}_{q}$ are pairwise
distinct for different choices of $\mathfrak{p}$ since $(\mathcal{C}_{1},\dots,\mathcal{C}_{P})$
is a partition. This still holds for the elements $\mathcal{H}_{2}$
in 
\begin{equation}
\left\{ \mathcal{H}_{1}\Delta\left(\bigcup_{q\in\mathfrak{p}}\mathcal{C}_{q}\right),\,\mathfrak{p}\subseteq[P]\right\} .\label{eq: family of conjugated signatures translated}
\end{equation}
satisfying (\ref{eq: clustering equivalence classes same signature, +}),
since the symmetric difference is invertible and, hence, the mapping
$\bigcup_{q\in\mathfrak{p}}\mathcal{C}_{q}\mapsto\mathcal{H}_{1}\Delta\left(\bigcup_{q\in\mathfrak{p}}\mathcal{C}_{q}\right)$
is bijective. 
\begin{prop}
\label{prop: number row/column configurations } The number of distinct
signatures obtained from sign choices (\ref{eq: sign switch operator, A})
for columns and (\ref{eq: row normalization}) for a row is $2^{n+1-P}$,
where $P$ is the number of classes of $\rightarrow_{\mathcal{V}}$
associated with $\mathbf{A}$. 
\end{prop}
\begin{proof}
For each subset $\mathcal{I}\in\mathfrak{G}$ consider a map constructed
as in (\ref{eq: representatives in pivot set}) that associates $\nu(\alpha)\in\mathcal{I}\setminus\mathcal{V}$
to a unique $\alpha\in\mathcal{I}\setminus\mathcal{V}$ so that $\mathcal{V}_{\alpha}^{\nu(\alpha)}\in\mathfrak{G}$.
This also implies $\alpha\rightarrow_{\mathcal{V}}\nu(\alpha)$. Thus,
the non-vanishing minors of $\mathbf{A}$ correspond to subsets which
intersect all the classes $\mathcal{C}_{1},\dots,\mathcal{C}_{P}$
and, in particular, 
\begin{equation}
\#(\mathcal{I}\cap\mathcal{C}_{q})=\#(\mathcal{V}\cap\mathcal{C}_{q})=:k_{q},\quad q\in[P]\label{eq: constant intersection with classes}
\end{equation}
does not depend on $\mathcal{I}\in\mathfrak{G}$. Then, the involutions
$\Delta_{q}$ defined by 
\begin{equation}
\mathcal{H}\mapsto\Delta_{q}(\mathcal{H}):=\mathcal{H}\Delta\mathcal{C}_{q},\quad\mathcal{H}\subseteq[n],\,q\in[P]\label{eq: involutions for classes}
\end{equation}
act as follows 
\begin{eqnarray}
\mathrm{sign}\left[\Delta\left(\boldsymbol{\ensuremath{\sigma}}_{\Delta_{q}\mathcal{H}}(\mathbf{A});\mathcal{I}\right)\right] & = & \mathrm{sign}\left[\Delta\left(\boldsymbol{\ensuremath{\sigma}}_{\mathcal{C}_{q}}(\boldsymbol{\ensuremath{\sigma}}_{\boldsymbol{\ensuremath{\sigma}}_{\mathcal{H}}}(\mathbf{A}));\mathcal{I}\right)\right]\nonumber \\
 & = & (-1)^{k_{q}}\cdot\mathrm{sign}\left[\Delta\left(\boldsymbol{\ensuremath{\sigma}}_{\boldsymbol{\ensuremath{\sigma}}_{\mathcal{H}}}(\mathbf{A});\mathcal{I}\right)\right].\label{eq: parity of classes}
\end{eqnarray}
Using the partition $(\mathcal{C}_{1},\dots,\mathcal{C}_{P})$, one
can uniquely express $\mathcal{H}\subseteq[n]$ in terms of its components
$\mathcal{H}_{q}:=\mathcal{H}\cap\mathcal{C}_{q}$ and get, for each
$\mathfrak{p}:=\{q_{1},\dots,q_{T}\}\subseteq[P]$, the following
equalities 
\begin{eqnarray}
\mathcal{H}\Delta\left(\bigcup_{q\in\mathfrak{p}}\mathcal{C}_{q}\right) & = & \left(\bigcup_{s\in[P]\setminus\mathfrak{p}}\mathcal{H}_{q}\right)\cup\left(\bigcup_{q\in\mathfrak{p}}\mathcal{C}_{q}\setminus\mathcal{H}_{q}\right)\nonumber \\
 & = & \Delta_{q_{1}}\circ\Delta_{q_{2}}\circ\cdots\circ\Delta_{q_{T}}(\mathcal{H}).\label{eq: composition of involutions}
\end{eqnarray}
Hence, if $\mathcal{H}_{2}=\mathcal{H}_{1}\Delta\left(\bigcup_{q\in\mathfrak{p}}\mathcal{C}_{q}\right)$
(as in (\ref{eq: clustering equivalence classes same signature, +})),
then from (\ref{eq: parity of classes}) and (\ref{eq: composition of involutions})
one gets 
\begin{equation}
\mathrm{sign}\left[\Delta\left(\boldsymbol{\ensuremath{\sigma}}_{\mathcal{H}_{2}}(\mathbf{A});\mathcal{I}\right)\right]=\left(\prod_{q\in\mathfrak{p}}(-1)^{k_{q}}\right)\cdot\mathrm{sign}\left[\Delta\left(\boldsymbol{\ensuremath{\sigma}}_{\mathcal{H}_{1}}(\mathbf{A});\mathcal{I}\right)\right].\label{eq: parity classes and composition involutions}
\end{equation}
Now choose any ancillary soliton parameter $\kappa_{0}\notin\{\kappa_{1},\dots,\kappa_{n}\}$,
introduce the matrix 
\begin{equation}
\mathbf{\ensuremath{\mathring{A}}}:=\left(1\right)\oplus\mathbf{A}=\left(\begin{array}{cc}
1 & \vec{0}_{n}^{T}\\
\vec{0}_{k} & \mathbf{A}
\end{array}\right).\label{eq: extended coefficient matrix, row to column}
\end{equation}
and denote the new column by the index $0$. There exists a bijection
between the non-vanishing minors of $\mathbf{A}$ and those of $\mathbf{\ensuremath{\mathring{A}}}$,
that is $\mathcal{I}\mapsto\{0\}\cup\mathcal{I}$. Any choice of signs
$\boldsymbol{\ensuremath{\sigma}}_{\mathcal{H}}$ for the columns of $\mathbf{A}$,
$\mathcal{H}\subseteq[n]$, can be extended to a choice for $\mathbf{\ensuremath{\mathring{A}}}$
as $\boldsymbol{\ensuremath{\mathring{\sigma}}}_{\mathcal{H}}:=(1)\oplus\boldsymbol{\ensuremath{\sigma}}_{\mathcal{H}}$.
Furthermore, the choice $R=-1$ in (\ref{eq: row normalization})
is restated as a new signature for columns only, i.e., $(-1)\oplus\idd_{n}=:\boldsymbol{\ensuremath{\mathring{\sigma}}}_{\{0\}}$.
Hence this equivalent model describes the original one using only
column operations, up to a common multiplicative factor for the terms
$\Lambda_{\mathcal{I}}(\boldsymbol{x})$. The pivot set for $\mathbf{\ensuremath{\mathring{A}}}$
is $\mathring{\mathcal{V}}:=\{0\}\cup\mathcal{V}$, and this extends
the relation $\rightarrow_{\mathcal{V}}$ to the equivalence $\rightarrow_{\mathring{\mathcal{V}}}$
whose classes are $\mathcal{C}_{1},\dots,\mathcal{C}_{P}$ and $\{0\}$,
since $\mathring{\mathcal{V}}_{\alpha}^{0}\in\mathcal{P}_{k}[n]\setminus\mathfrak{G}$
for all $\alpha\in[n]$.

So there are $2^{n+1}$ choices of signs for the columns of $\mathbf{\ensuremath{\mathring{A}}}$
and $P+1$ distinct equivalence classes for $\rightarrow_{\mathring{\mathcal{V}}}$.
The cardinality of one of them, i.e., $\#\{0\}=1$, is odd: hence,
from (\ref{eq: parity of classes}), the substitution $\mathfrak{p}\mapsto\{0\}\cup\mathfrak{p}$
induces the mapping 
\begin{eqnarray}
\mathrm{sign}\left[\Delta\left(\boldsymbol{\ensuremath{\sigma}}_{\mathcal{H}_{2}}(\mathbf{\ensuremath{\mathring{A}}});\mathcal{I}\right)\right] & \mapsto & \mathrm{sign}\left[\Delta\left(\boldsymbol{\ensuremath{\sigma}}_{\mathcal{H}_{2}\Delta\{0\}}(\mathbf{\ensuremath{\mathring{A}}});\mathcal{I}\right)\right]\nonumber \\
 & = & -\mathrm{sign}\left[\Delta\left(\boldsymbol{\ensuremath{\sigma}}_{\mathcal{H}_{2}}(\mathbf{\ensuremath{\mathring{A}}});\mathcal{I}\right)\right]\label{eq: row sign as column involution}
\end{eqnarray}
for any subset $\mathfrak{p}\subseteq[P]$. In conclusion, Lemma \ref{lem: cluster classes}
states that, for each $\mathcal{H}_{1}\subseteq\{0\}\cup[n]$, the
possible sets $\mathcal{H}_{2}$ satisfying (\ref{eq: hp equal signs minors, +})
lie in (\ref{eq: family of conjugated signatures translated}); each
element of this family satisfies (\ref{eq: parity classes and composition involutions})
independently on $\mathcal{I}\in\mathfrak{G}$; finally, from (\ref{eq: row sign as column involution})
it follows that each signature $\mathfrak{p}\subseteq[P]$ corresponds
to the opposite one $\{0\}\cup\mathfrak{p}$, so they occur in equal
numbers. This means that exactly half of the terms in (\ref{eq: family of conjugated signatures translated})
have the same signature of $\mathcal{H}_{1}$, while the other half
have opposite signature. Hence there are $\frac{1}{2}2^{P+1}=2^{P}$
subsets associated with a single signature, and the number of allowed
signatures is $\frac{2^{n+1}}{2^{P}}=2^{n+1-P}$. 
\end{proof}
\begin{rem}
\label{rem: classes generating projective sign vector} The previous
discussion also implies that the sets (\ref{eq: family of conjugated signatures translated})
are equipollent and pairwise disjoint, since each of them contains
all the possible combinations of subsets of $[n]$ that induce a given
signature $\Sigma$ or the opposite $-\Sigma$. 
\end{rem}
The freedom in the choice of signs generalizes free statistical amoebas
(including an additional row sign flip), which fall within the case
$P=1$. Indeed, let $\mu_{s}$ denote the number of distinct signatures
induced by the sign flip of exactly $s$ columns of $\mathbf{A}$,
$s\in[n]$, without limitations on $R\in\{\pm1\}$. Each combination
of signs for row and columns can be labelled by an element in $\mathcal{P}_{s}[n]\times\{\pm1\}$,
hence $\mu_{s}\leq2\cdot{{n} \choose {s}}$. At $s=\frac{n}{2}$,
$\mathcal{H}\in\mathcal{P}_{n/2}[n]$ produces the same signature
of its complement $[n]\setminus\mathcal{H}\in\mathcal{P}_{n/2}[n]$
for an appropriate choice of $R$, so $\mu_{n/2}\leq{{n} \choose {n/2}}$.
These bounds, along with the action of (\ref{eq: parity classes and composition involutions})
and the result in Proposition \ref{prop: number row/column configurations },
give 
\[
2^{n}=\sum_{s=0}^{n/2}\mu_{s}=\sum_{s=0}^{n/2}\frac{\mu_{s}+\mu_{n-s}}{2}\leq\sum_{s=0}^{n}{{n} \choose {s}}=2^{n}.
\]
Thus all the bounds are in fact equalities, i.e., $\mu_{s}=2\cdot{{n} \choose {s}}$
at $s\neq\frac{n}{2}$ and $\mu_{n/2}={{n} \choose {n/2}}$ .

\section{\label{sec: Levels of constrained amoebas} Levels of constrained
amoebas }

The previous discussion leads to an extension of the concept of statistical
amoeba to higher dimensional cases. Following the construction for
free statistical amoebas in \cite{AK2016b}, one can focus on the
family of functions (or the associated locus of zeros) obtained from
$\tau$ through all the combinations of $s$ sign flips for columns
at $s\in[n]$ fixed.

We can associate with each $\mathcal{S}\subseteq[n]$ the vector $v(\mathcal{S})\in\mathbb{F}_{2}^{n}$
defined as $v(\alpha)=1$ if $\alpha\in\mathcal{S}$ and $v(\alpha)=0$
otherwise. The intersection of $\mathcal{S}_{1},\mathcal{S}_{2}\in\mathcal{P}[n]$
is given by the componentwise product $v(\mathcal{S}_{1}\cap S_{2})_{\alpha}=v(\mathcal{S}_{1})_{\alpha}\cdot v(\mathcal{S}_{2})_{\alpha}$,
hence the parity of $\#(\mathcal{S}_{1}\cap\mathcal{S}_{2})$ is equal
to the dot product 
\begin{equation}
\#(\mathcal{S}_{1}\cap\mathcal{S}_{2})\cong\sum_{\alpha=1}^{n}v(\mathcal{S}_{1})_{\alpha}\cdot v(\mathcal{S}_{2})_{\alpha}=:v(\mathcal{S}_{1})\cdot v(\mathcal{S}_{2})\quad\mathrm{mod}\,2.\label{eq: parity intersection dot product}
\end{equation}
Motivated by this, we adopt the following notation 
\begin{eqnarray}
\mathcal{S}\parallel\mathcal{T} & \Leftrightarrow & \#(\mathcal{S\cap T})\cong1\quad\mathrm{mod}\,2,\\
\mathcal{S}\perp\mathcal{T} & \Leftrightarrow & \#(\mathcal{S\cap T})\cong0\quad\mathrm{mod}\,2\label{eq: parity intersection, parallel and perpendicular}
\end{eqnarray}
and $\mathcal{H}\perp\mathfrak{L}:=\left\{ \mathcal{L}\in\mathfrak{L}:\,\mathcal{H}\perp\mathcal{L}\right\} $
with $\mathcal{H}\in\mathcal{P}[n]$ and $\mathfrak{L}\subseteq\mathcal{P}[n]$.

Unlike the free statistical amoeba at $k=1$, only some strata are
visible under choices of signs allowed by determinantal/integrability
requirements. For instance, a single change of sign preserves neither
the determinantal structure (see Example \ref{exa: no 1-stratum free})
nor the solution of the KP II equation in generic situations. Hence
the $1$-stratum for the free statistical amoeba is not part of the
constrained amoeba.

Let us focus on the case $\mathfrak{G}=\mathcal{P}_{k}[n]$. For all
the choices of $\boldsymbol{\ensuremath{\sigma}}_{\mathcal{S}}$, $s\in[n]$
and $\mathcal{S}\in\mathcal{P}_{s}[n]$, the number $\Omega(n,k;s)$
of $-$ signs generated by $\mathcal{S}$ is equal to the cardinality
of $\mathcal{S}\parallel\mathcal{P}_{k}[n]$. This holds for all possible choices of $\mathcal{S}\in\mathcal{P}_{s}[n]$
by permutation symmetry. The singular locus corresponding to $\mathcal{S}$
is a subset of all the singular loci in the free $\Omega(n,k;s)$-statistical
amoeba.
\begin{prop}
If $\mathfrak{G}=\mathcal{P}_{k}[n]$, then a choice (\ref{eq: sign switch operator, A})
with $\#\mathcal{S}=s$ induces a signature with 
\begin{equation}
\hspace*{-1cm}
\Omega(n,k;s):=
\left\{\begin{array}{ccc}
\omega(n,k;s),& & m\geq k+s,\\
\frac{1-(-1)^{n-k-s}}{2}{{n}\choose{k}}&+(-1)^{n-k-s}\omega(n,n-k;n-s), & m<k+s
\end{array}\right.\label{eq: accessible strata}
\end{equation}
where 
\begin{equation}
\omega(n,k;s):=\frac{1}{2}{{n} \choose {k}}-\frac{1}{2}{{n-s} \choose {k}}\cdot_{2}F_{1}(-s,-k;n-s-k+1;-1).\label{eq: first branch omega}
\end{equation}
\end{prop}
\begin{proof}
Let $\mathcal{I}\in\mathcal{P}_{s}[n]$. The number of subsets $\mathcal{A}\in\mathcal{P}_{k}[n]$
satisfying $\mathcal{A}\parallel\mathcal{I}$ is 
\begin{eqnarray}
 &  & \sum_{a=0}^{\frac{1}{2}\min\{k,s\}}{{s} \choose {2a+1}}\cdot{{n-s} \choose {k-2a+1}}\nonumber \\
 & = & \sum_{a=0}^{s}\frac{(-1)^{a+1}+1}{2}\cdot{{s} \choose {a}}\cdot{{n-s} \choose {k-a}}\nonumber \\
 & = & \frac{1}{2}\cdot\sum_{a=0}^{k}(-1)^{a+1}\cdot{{s} \choose {a}}\cdot{{n-s} \choose {k-a}}+\frac{1}{2}\cdot\sum_{a=0}^{s}{{s} \choose {a}}\cdot{{n-s} \choose {k-a}}\label{eq: number 2-parallel k-subspaces}
\end{eqnarray}
where ${{t}\choose{w}}=0$ at $t<w$ or $w<0$. We observe that $\mathcal{A}\parallel\mathcal{S}$
is equivalent to 
\begin{eqnarray}
\#\left(([n]\setminus\mathcal{S})\cap([n]\setminus\mathcal{A})\right) & = & \#[n]\setminus(\mathcal{S}\cup\mathcal{A})\nonumber \\
& \cong & n-s-k+1\quad\mathrm{mod}\,2\label{eq: symmetry complements for computation}
\end{eqnarray}
where the principle of inclusion-exclusion has been used in the second
line. So, at even $n-s-k$ (respectively, odd $n-s-k$), the enumeration
of sets $\mathcal{A}\in\mathcal{P}_{k}[n]$ with $\mathcal{A}\parallel\mathcal{S}$
is equivalent to the enumeration of subsets $[n]\setminus\mathcal{A}\in\mathcal{P}_{n-k}[n]$
with $[n]\setminus\mathcal{A}\parallel[n]\setminus\mathcal{S}$ (respectively,
$[n]\setminus\mathcal{A}\perp[n]\setminus\mathcal{S}$). Furthermore,
exactly one holds between $n-s-k\geq0$ or $n-(n-s)-(n-k)=s+k-n>0$. 

So we can look at the situations where $n-s-k\geq0$ and derive the
required quantity in other cases by the previous observation. First
note that 
\begin{equation}
\sum_{a=0}^{s}{{s} \choose {a}}\cdot{{n-s} \choose {k-a}}={{n} \choose {k}}.\label{eq: k-subsets decomposition counting through intersection}
\end{equation}
counts the number of elements of $\mathcal{P}_{k}[n]$. For the first
summation, one has 
\begin{eqnarray}
 &  & \sum_{a=0}^{k}(-1)^{a+1}\cdot{{s} \choose {a}}\cdot{{n-s} \choose {k-a}}\nonumber \\
 & = & -{{n-s} \choose {k}}\cdot\sum_{a=0}^{k}{{k} \choose {a}}\cdot\frac{(-1)^{a}\cdot s!(n-s-k)!}{(s-a)!(n-s-k+a)!}\nonumber \\
 & = & -{{n-s} \choose {k}}\cdot_{2}F_{1}(-s,-k;n-s-k+1;-1)\label{eq: sign-weighted k-subsets}
\end{eqnarray}
where $_{2}F_{1}$ is the Gaussian hypergeometric function, and all
the equivalences are well-posed due to the condition $n-s-k\geq0$.
Adding the two contributions, we get $\Omega(n,k;s)=\omega(n,k;s)$.
From this we also find that, at $n-s-k<0$, the number of $(n-k)$-subsets
$\mathcal{B}\subseteq[n]$ with $\mathcal{B}\parallel([n]\setminus\mathcal{S})$
is $\omega(n,n-k;n-s)$, while those with $\mathcal{B}\perp([n]\setminus\mathcal{S})$
are ${{n}\choose{k}}-\omega(n,n-k;n-s)$. By (\ref{eq: symmetry complements for computation})
and subsequent observations, these two quantities respectively enumerate
the number of $s$-subsets $\mathcal{A}\subseteq[n]$ with $\mathcal{A}\parallel\mathcal{S}$
at even $n-k-s$ and at odd $n-k-s$. These results can be expressed
as in (\ref{eq: accessible strata}), which concludes the proof. 
\end{proof}

It is worth remarking that there is a duality between the dimension $k$ and the level $s$: the number of pairs $(\mathcal{H}_{1};\mathcal{H}_{2})\in\mathcal{P}_{k}[n]\times\mathcal{P}_{s}[n]$ such that $\mathcal{H}_{1}\parallel\mathcal{H}_{2}$ can be enumerated in two ways, i.e., fixing one of the two components $\mathcal{H}_{i}$, $i\in\{1,2\}$ and considering all the subsets $\mathcal{I}$ with $\mathcal{H}_{i}\parallel\mathcal{I}$. This double counting implies the following identity   
\begin{equation}
{{n}\choose{s}}\cdot\Omega(n,k;s)={{n}\choose{k}}\cdot\Omega(n,s;k).\label{eq: duality s<->k}
\end{equation}
From this, one also finds  
\begin{equation}
\Omega(n,k;s)<\frac{1}{2}{{n}\choose{k}}\Leftrightarrow\frac{{{n}\choose{s}}}{{{n}\choose{k}}}\Omega(n,k;s)<\frac{1}{2}{{n}\choose{s}}\Leftrightarrow \Omega(n,s;k)<\frac{1}{2}{{n}\choose{s}}.\label{eq: +/- amoeba duality}
\end{equation}
In this sense, the distinction between (free) amoebas and antiamoebas, as defined in Section \ref{subsec: Statistical amoebas }, is compatible with such a duality. Furthermore, when $g_{\mathcal{I}}>0$ for all $\mathcal{I}\in\mathcal{P}_{k}[n]$, the quantity $\Omega(n,s;k)$ dual to $\Omega(n,k;s)$ is related to the behaviour of the constrained amoeba at large values of $\boldsymbol{x}$, namely, to its \emph{tropical limit} \cite{AK2016b}: outside the locus where $\displaystyle{\max_{\mathcal{H}\in\mathcal{P}_{k}[n]}}\Lambda_{\mathcal{H}}(\boldsymbol{x})=:\Lambda_{\mathcal{D}}(\boldsymbol{x})$ is attained more than once, the sign of $\tau_{\mathcal{S}}(\boldsymbol{x})$ at $||\boldsymbol{x}||\rightarrow \infty$ coincides with the induced sign for this dominant term, that is $(-1)^{\#(\mathcal{S}\cap\mathcal{D})}$. So $\Omega(n,s;k)$ represents the number of subsets $\mathcal{S}\in\mathcal{D}\parallel\mathcal{P}_{s}[n]$ where $\tau_{\mathcal{S}}(\boldsymbol{x})<0$.

\section{\label{sec: Applications} Applications }

\subsection{\label{subsec: Application to complexity} Application to information
transfer via message coding}

The previous results suggest using an initial function (\ref{eq: Cauchy-Binet formula, solitons})
to encode information. Specifically, we can think at each signature
as a message encoded in a $G$-bits string, where $G:=\#\mathfrak{G}$.
Here we assume that the order of the bits in the string corresponds
to a given (e.g., lexicographical) order for the elements of $\mathcal{P}_{k}[n]$.
Analogously, we can represent (\ref{eq: Cauchy-Binet formula, solitons})
as a string with ${{n}\choose{k}}$ entries in $\{0,1,\bot\}$, where
the entries label the minors of $\mathbf{A}$ and the symbol $\bot$
is associated with vanishing minors. 

If one knows that the original function $\tau(\boldsymbol{x})$ solves the KP II equation
and receive a new function $\tau(\boldsymbol{x})-2\cdot\tau_{\Sigma}(\boldsymbol{x})$,
then the fulfilment of the KP II equation implies that a particular
choice of signs has been sent. We stress that all the signs have to
be checked to confirm that $\Sigma$ is induced by row/column operations:
indeed, in the generic case $G={{n}\choose{k}}$, a single switch of
sign converts a constrained signature to a non-constrained one, as
shown in Example \ref{exa: no 1-stratum free}. However, in the present
approach, one can check if $\Sigma$ is induced by a choice (\ref{eq: sign switch operator, A})
indirectly, i.e., without having to \textit{find} such a configuration
$\sigma_{\mathcal{S}}$ and, hence, avoiding the effort to get this
additional knowledge. 

In order to quantify the amount of information that can be acquired
through this single check, we consider the well-known Kullback-Leibler
divergence \cite{Cover2006}. After the reception of the message,
but before any additional check on the signs, one can recover
the data related to dependence relations among the columns of $\mathbf{A}$,
namely $\mathfrak{G}$, $G$ and, for any fixed $\mathcal{V}\in\mathfrak{G}$, the relation $\rightarrow_{\mathcal{V}}$ and the dimensions $k_{1},\dots,k_{P}$. At this point, the \textit{prior}
information is based on strings of bits indexed by $\mathfrak{G}\subseteq\mathcal{P}_{k}[n]$.
With no additional constraint, we can assume a prior distribution $u$
where all the $2^{G}$ $G$-bits strings have the same statistical weight
$2^{-G}$. In the generic case $G={{n}\choose{k}}$ there are $\exp\left({{n}\choose{k}}\cdot\ln2\right)$
such strings, but only $2^{n}$ of them satisfy the KP II equation
by Theorem \ref{thm: compatible choices of signs, solution KP, including vanishing}
and Proposition \ref{prop: number row/column configurations }. This
restricted family of strings associated with solitonic signatures
is the support for the \textit{posterior} distribution $u_{\mathrm{KP}}$. So the Kullback-Leibler
divergence is given by 
\begin{eqnarray}
\mathrm{D^{\star}_{KL}}(u_{\mathrm{KP}}||u) & := & \sum_{\mathcal{S}\subseteq[n]}2^{-n}\cdot\ln\left(\frac{\exp(-n\cdot\ln2)}{\exp\left(-{{n}\choose{k}}\cdot\ln2\right)}\right)\nonumber \\
& = & \ln2\cdot\left[{{n}\choose{k}}-n\right].\label{eq: KL divergence, general and generic}
\end{eqnarray}
When $\#\mathfrak{G}<{{n}\choose{k}}$, the number of distinct general
strings is $2^{G},$ while the number of those satisfying the KP II
equation is $2^{n+1-P}$ by Proposition \ref{prop: number row/column configurations }.
So one gets 
\begin{equation}
\mathrm{D^{\star}_{KL}}(u_{\mathrm{KP}}||u)=\ln2\cdot(G-n+P-1)\label{eq:eq: KL divergence, general including non-generic}
\end{equation}
which is a real non-negative quantity, hence we also find 
\begin{equation}
G+P\geq n+2.\label{eq: additional Euler-like inequality}
\end{equation}

In the previous cases, the assumptions of equiprobability for strings
in $\{\pm1\}^{\mathfrak{G}}$ and for choices of signs induced by
(\ref{eq: sign switch operator, A}) are consistent as a result of
Proposition \ref{prop: number row/column configurations }, since the
equivalence classes associated to the different signatures are equipollent (also see Remark \ref{rem: classes generating projective sign vector}).
One can get the quantities (\ref{eq: KL divergence, general and generic})-(\ref{eq:eq: KL divergence, general including non-generic})
through the choices $(R,\boldsymbol{\ensuremath{\sigma}}_{\mathcal{S}})$,
$R\in\{\pm1\}$ and $\mathcal{S}\subseteq[n]$, which reduces to the
same multiple-counting of both general and solitonic signatures. This
type of equivalence does not necessarily hold when further restrictions
affect the prior. In particular, we consider a situation where the
number $s=\#\mathcal{S}$ for allowed $\sigma_{\mathcal{S}}$ is fixed
and known, as in the discussion on statistical amoebas in \cite{AK2016b}. 

In order to explore some features of this setting, we first estimate
the information content associated with the check of the KP II equation
in the generic situation $G={{n}\choose{k}}$. Let us look at the effects
of the knowledge of a fixed value for $s$ at $s\neq\frac{n}{2}$
to avoid the occurrence of solitonic signatures related by the involutions
(\ref{eq: involutions for classes}) and, hence, to single out the
contribution of the evidence on $s$. This enables us to consider
a preliminary check before looking at the KP II equation: if neither
the received string $\Sigma$ nor its opposite $-\Sigma$ satisfy
$\#\boldsymbol{\ensuremath{\mathrm{NS}}}=\Omega(n,k;s)$, then $\Sigma$ is not obtained through
(\ref{eq: sign switch operator, A}) with $R\in\{\pm1\}$ and, hence, does not satisfy the KP II equation. Furthermore, at $\Omega(n,k;s)\neq\frac{1}{2}{{n}\choose{k}}$,
also the row sign $R$ is fixed by this preliminary check, and this
restricts the support of the uniform distributions for the prior $u$ and
the posterior $u_{\mathrm{KP}}$ to $\mathcal{P}_{\Omega(n,k;s)}\left[\mathcal{P}_{k}[n]\right]$
and $\mathcal{P}_{s}[n]$, respectively. The uncertainty on $R$ still
remains at $\Omega(n,k;s)=\frac{1}{2}{{n}\choose{k}}$ . So the Kullback-Leibler
divergence is equal to
\begin{equation}
\mathrm{D^{s}_{KL}}(u_{\mathrm{KP}}||u)=\ln\left[{{n}\choose{s}}^{-1}\cdot{{{{n}\choose{k}}}\choose{\Omega(n,k;s)}}\right]-\ln2\cdot\delta_{2\cdot\Omega(n,k;s),{{n}\choose{k}}}.\label{eq: KL divergence, constrained number of - signs and generic}
\end{equation}
The last term, which involves the cases $\Omega(n,k;s)=\frac{1}{2}{{n}\choose{k}}$,
does not appear if one also knows that only column operations can
be performed. 

We present a graphical representation of the Kullback-Leibler divergence
at $G={{n}\choose{k}}$: the behaviour of (\ref{eq: KL divergence, constrained number of - signs and generic}) at different values of $s$ is shown in Figure \ref{fig: KL divergence, unconstrained number of - signs}, and can be compared with (\ref{eq: KL divergence, general and generic}) in Figure \ref{fig: unconstrained}. The analytic continuations of both these formulas have been used.
\begin{figure}
\centering
	\begin{minipage}{0.32\textwidth}
	\includegraphics[width=0.95\textwidth]{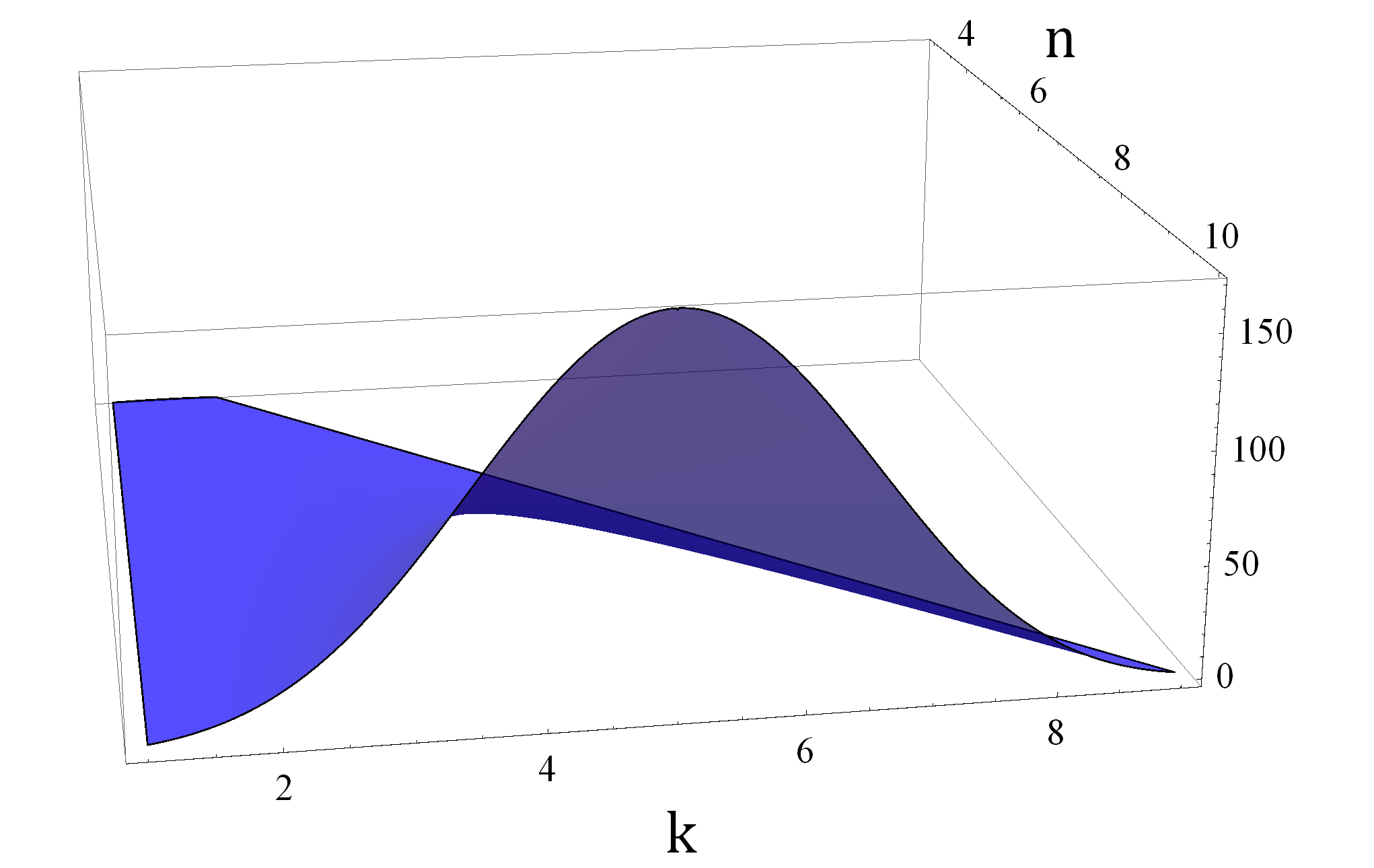}
	\subcaption{$s=1$}
	\label{fig: s=1}
	\end{minipage}
	\begin{minipage}{0.32\textwidth}
	\includegraphics[width=0.95\textwidth]{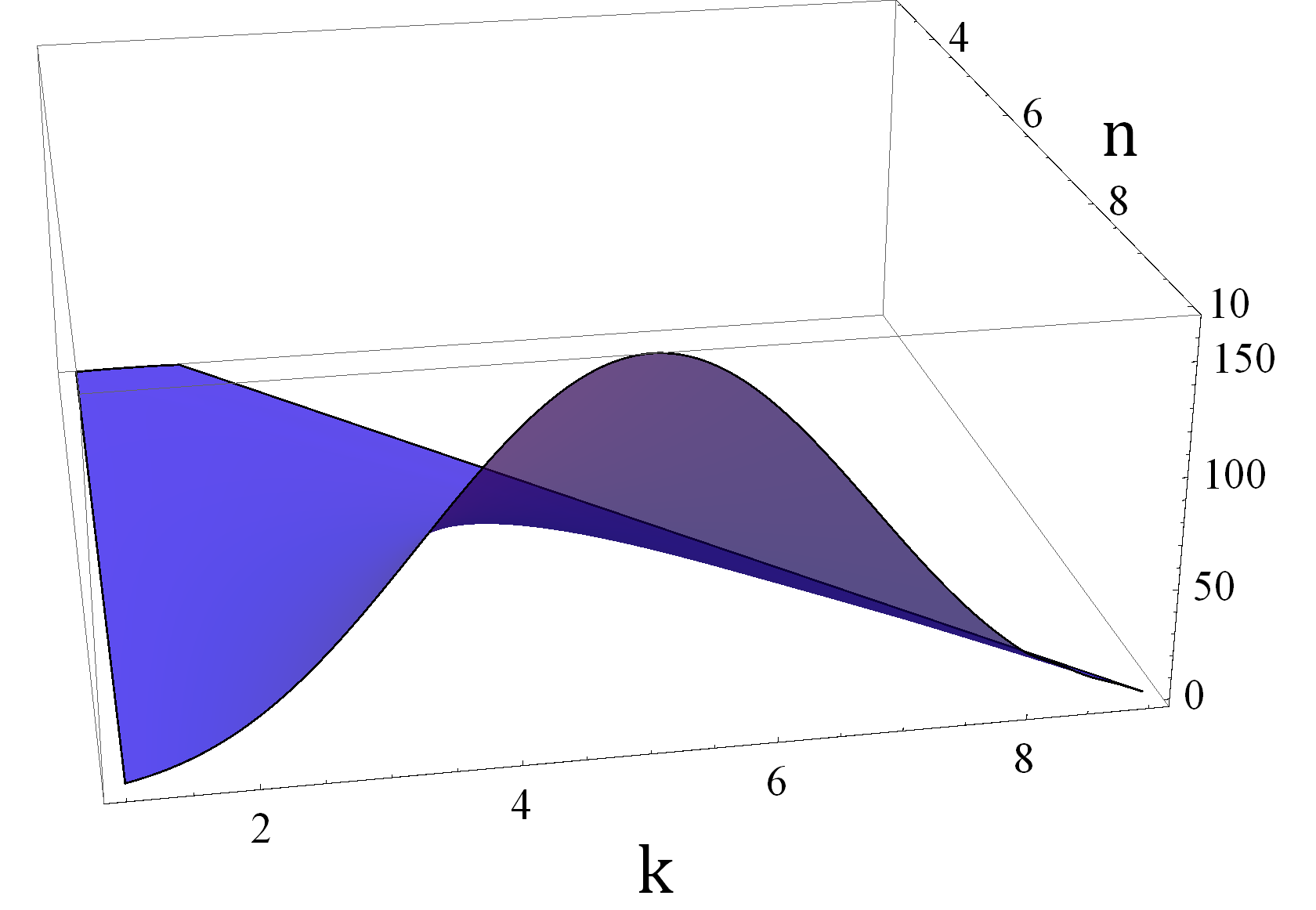}
	\subcaption{$s=2$}
	\label{fig: s=2}
	\end{minipage}
	\begin{minipage}{0.32\textwidth}
	\includegraphics[width=0.95\textwidth]{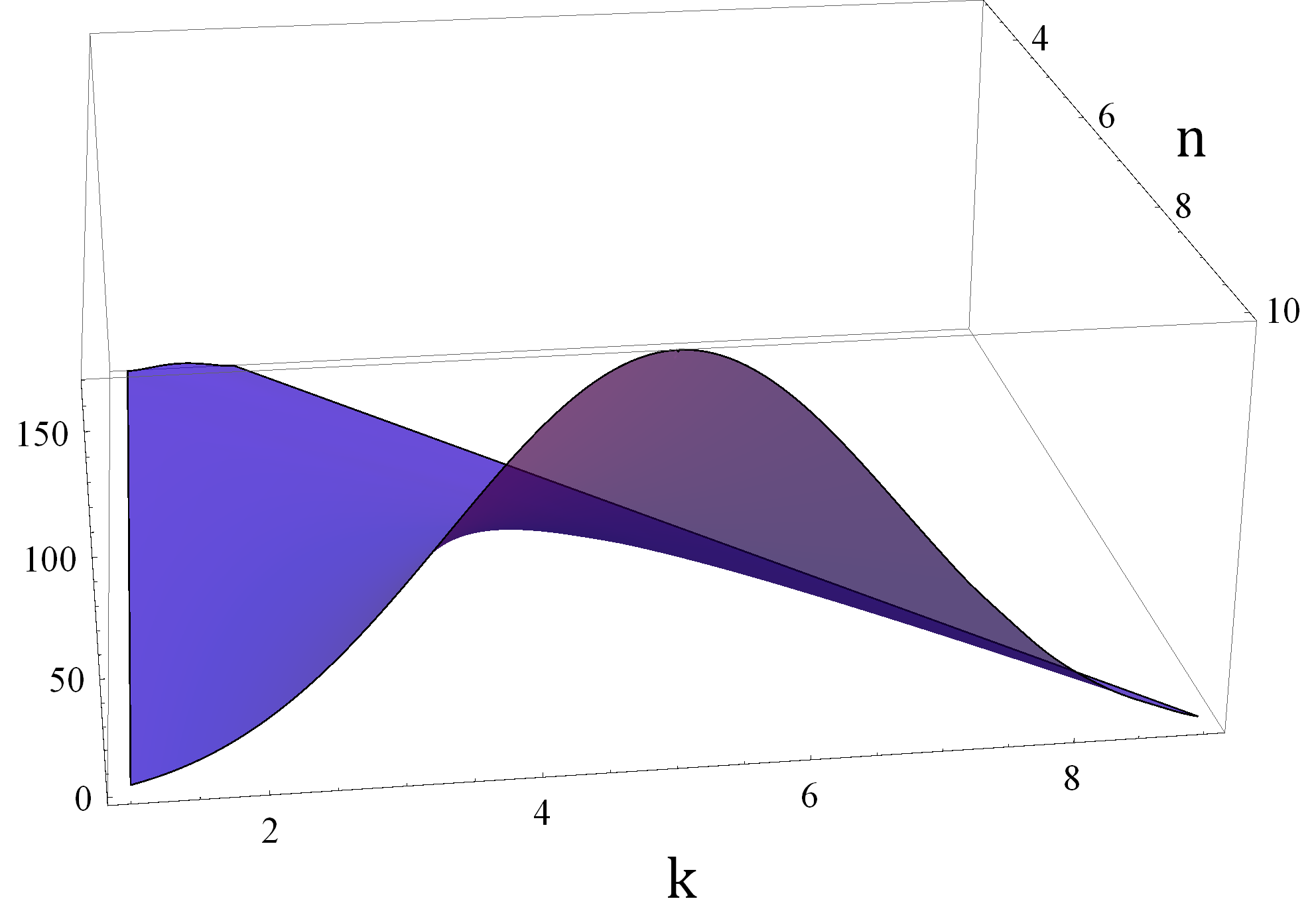}
	\subcaption{$s=3$}
	\label{fig: s=3}
	\end{minipage}
	\begin{minipage}{0.32\textwidth}
	\includegraphics[width=0.95\textwidth]{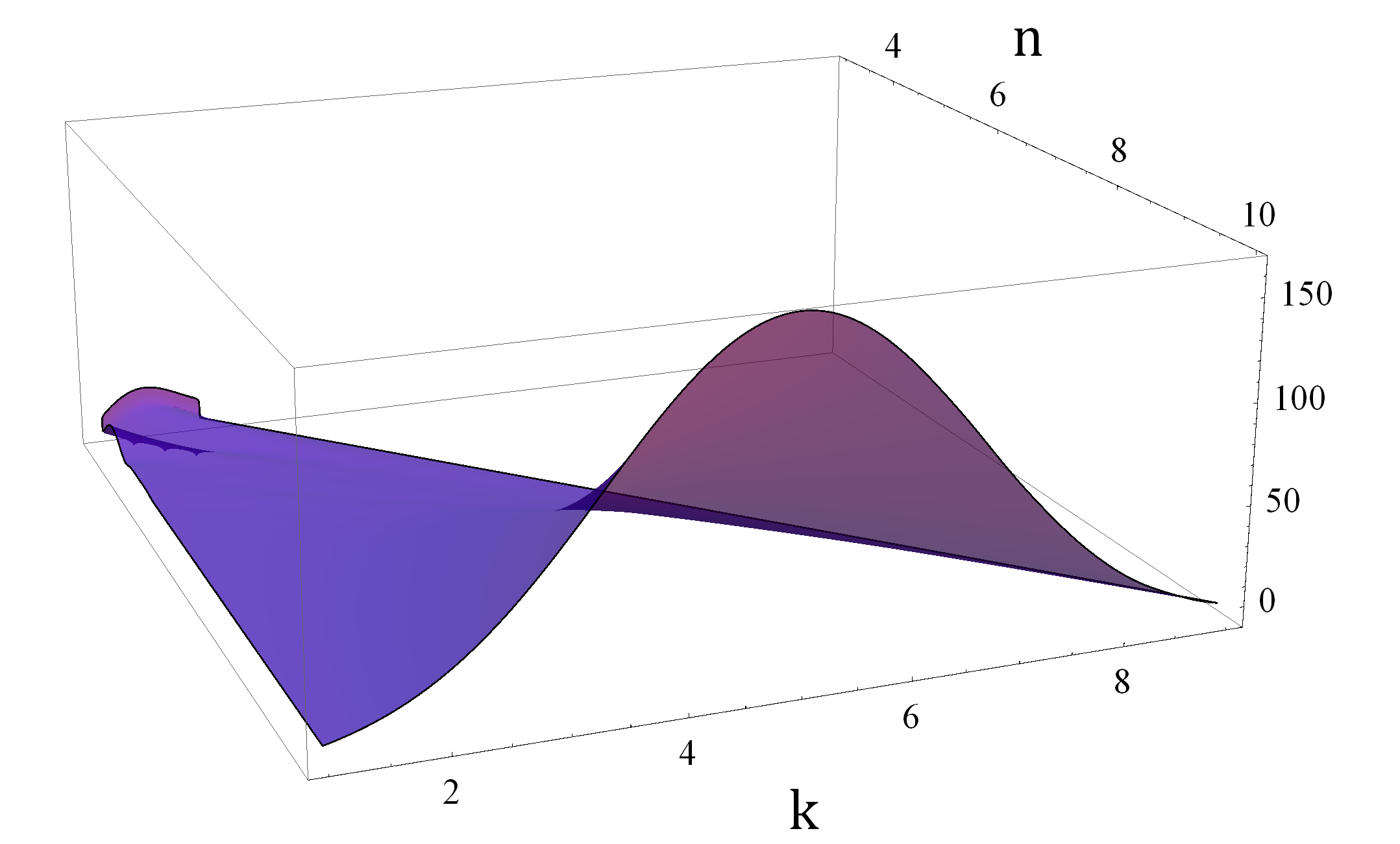}
	\subcaption{$s=4$}
	\label{fig: s=4}
	\end{minipage}
		\begin{minipage}{0.29\textwidth}
	\includegraphics[width=0.95\textwidth]{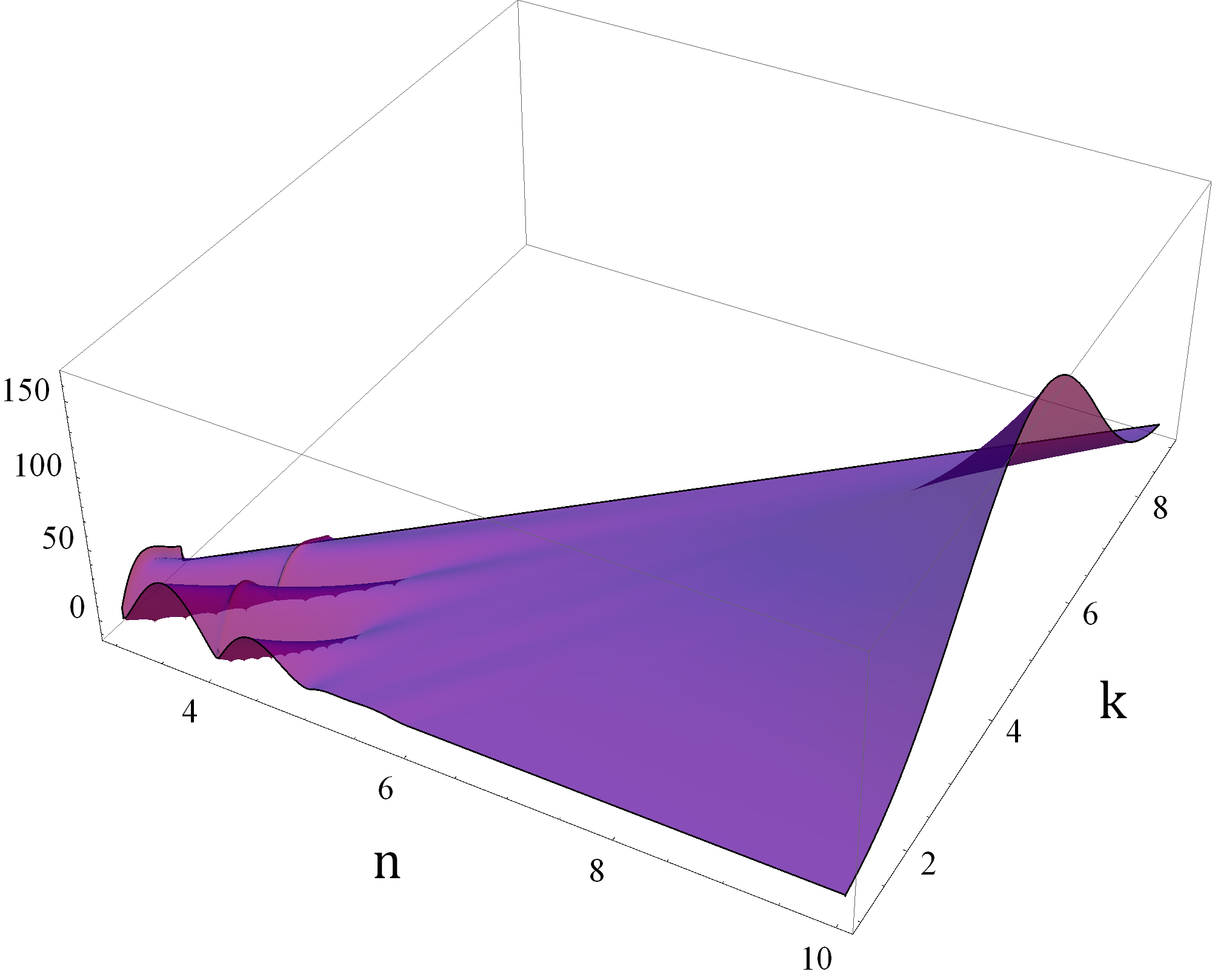}
	\subcaption{$s=5$}
	\label{fig: s=5}
	\end{minipage}
		\begin{minipage}{0.32\textwidth}
	\includegraphics[width=0.95\textwidth]{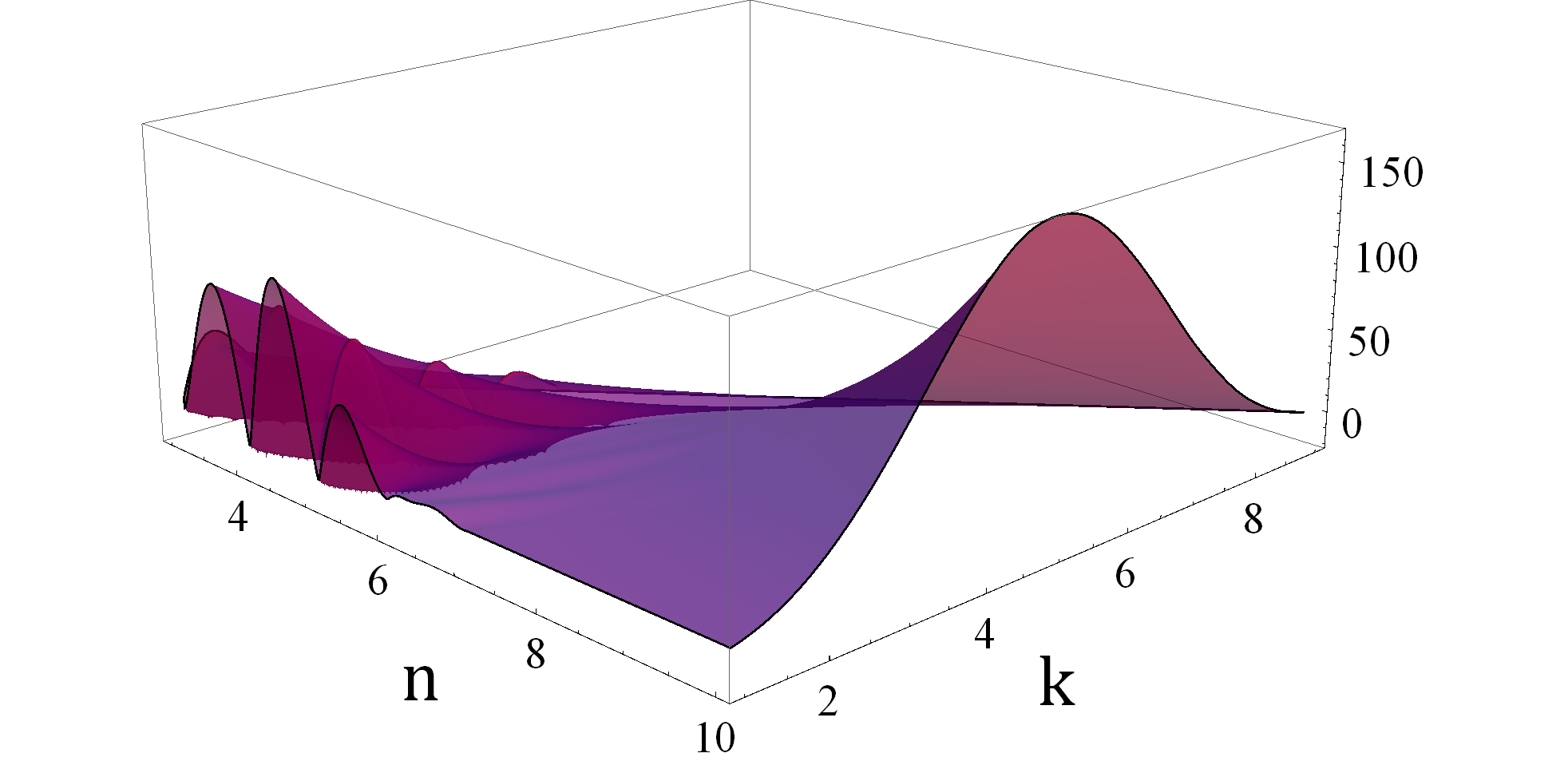}
	\subcaption{$s=6$}
	\label{fig: s=6}
	\end{minipage}
		\begin{minipage}{0.32\textwidth}
	\includegraphics[width=0.95\textwidth]{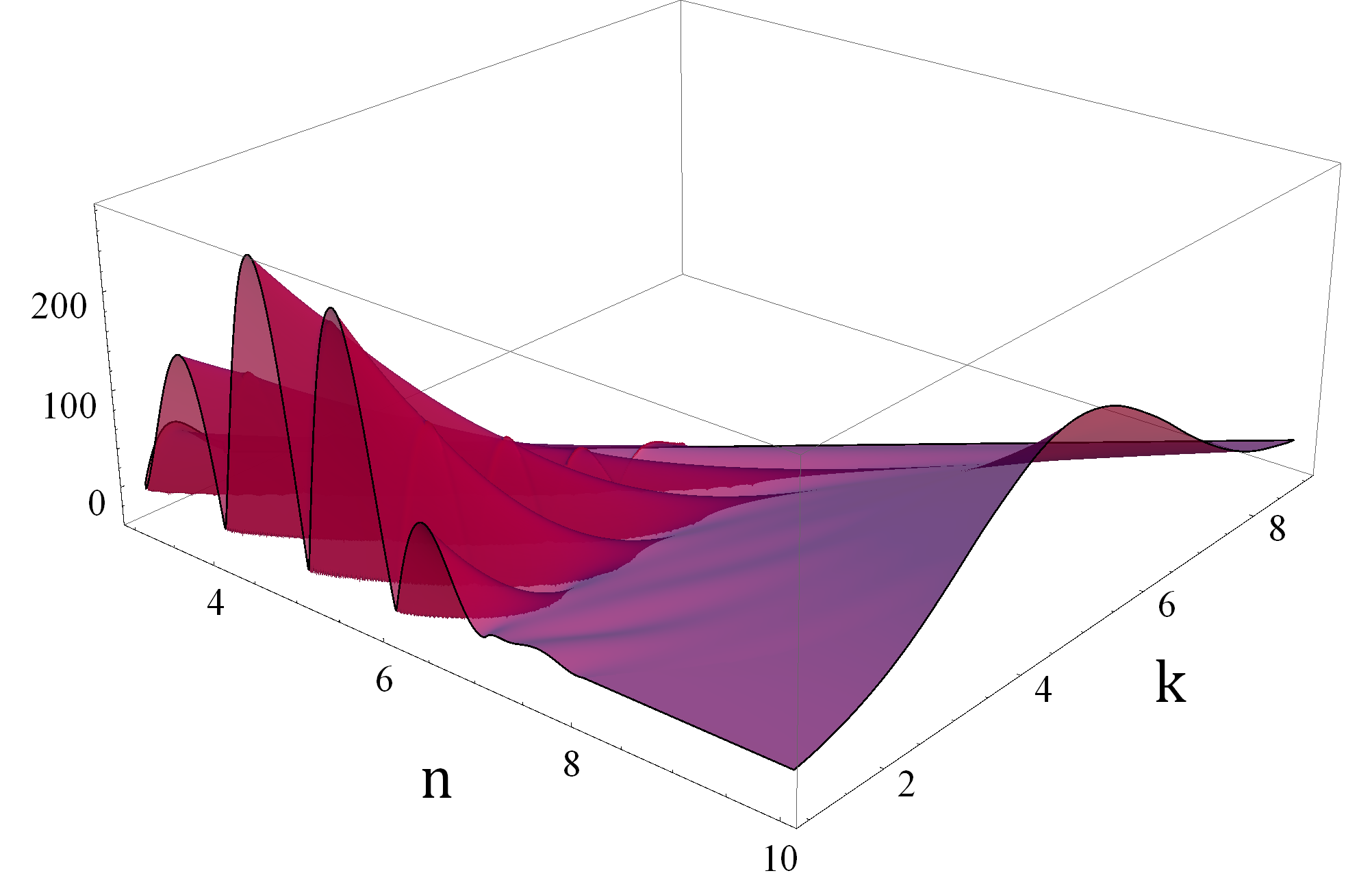}
	\subcaption{$s=7$}
	\label{fig: s=7}
	\end{minipage}
		\begin{minipage}{0.32\textwidth}
	\includegraphics[width=0.95\textwidth]{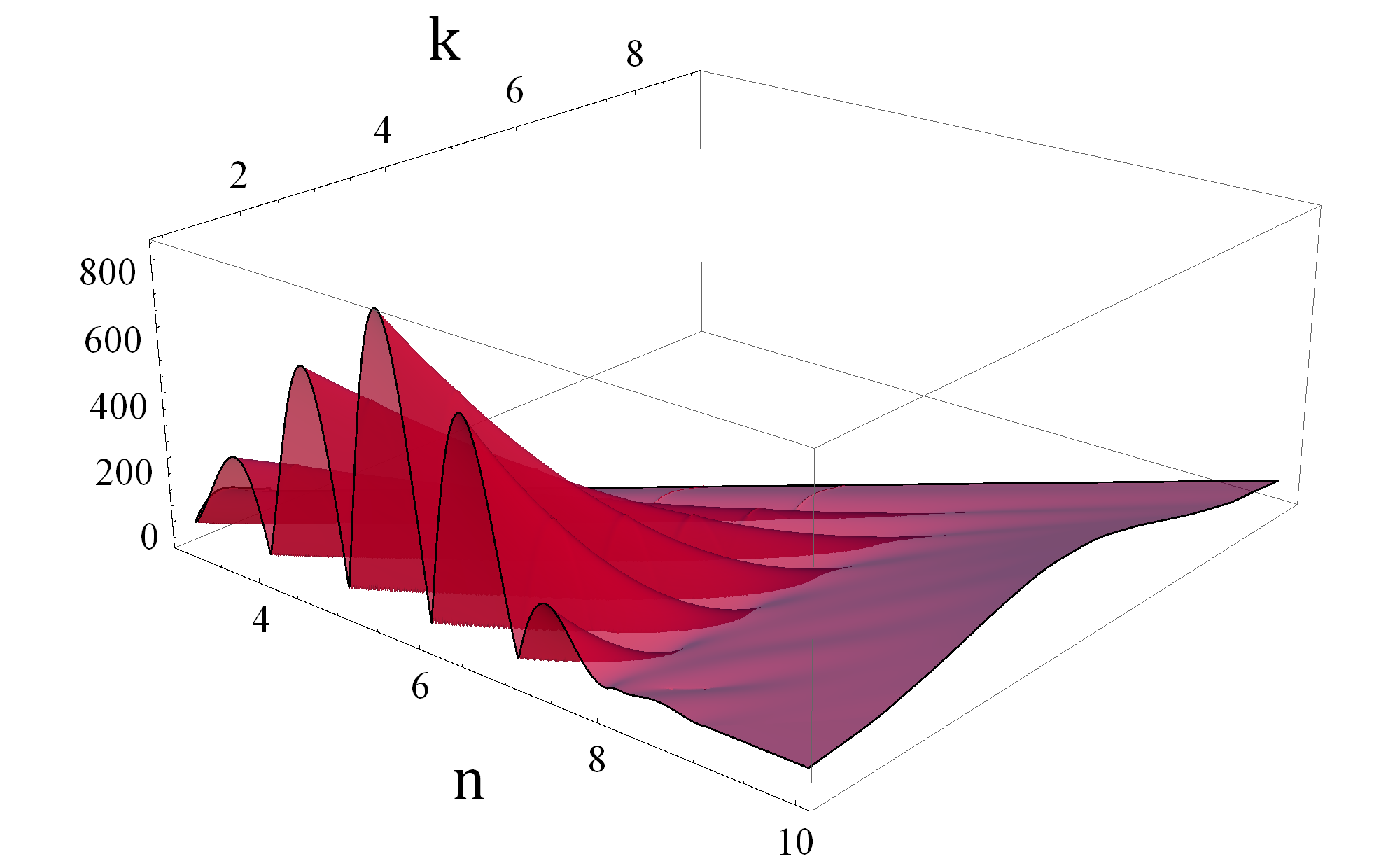}
	\subcaption{$s=8$}
	\label{fig: s=8}
	\end{minipage}
		\begin{minipage}{0.32\textwidth}
	\includegraphics[width=0.95\textwidth]{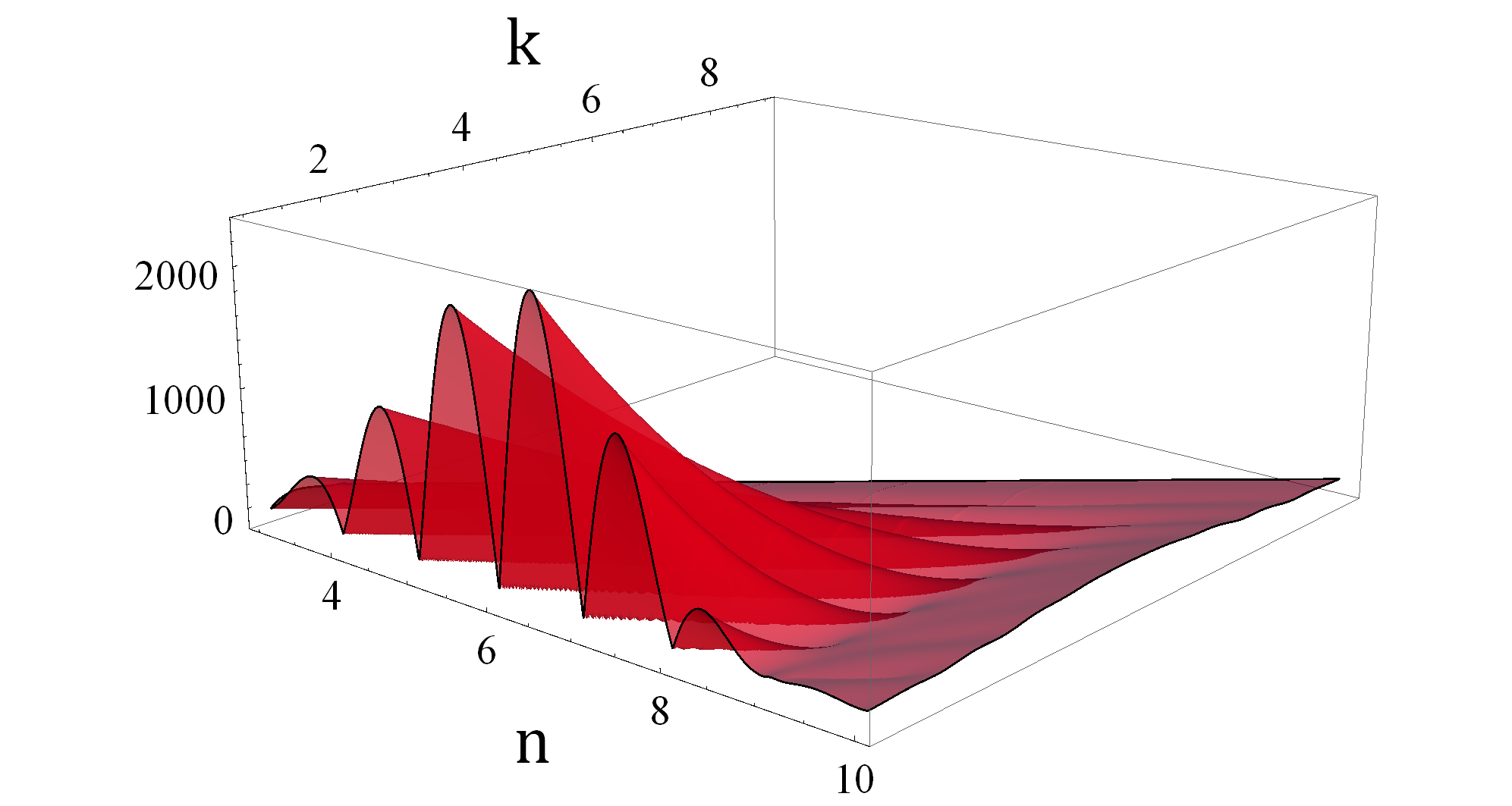}
	\subcaption{$s=9$}
	\label{fig: s=9}
	\end{minipage}
\caption{Values of $\mathrm{D^{s}_{KL}}(u_{\mathrm{KP}}||u)$ at $n\in[3,10]$ and $k\in[1,n]$ when $s$ is known.}
\label{fig: KL divergence, unconstrained number of - signs}
\end{figure}
\begin{figure}
\centering
	\begin{minipage}{0.38\textwidth}
	\includegraphics[width=0.95\textwidth]{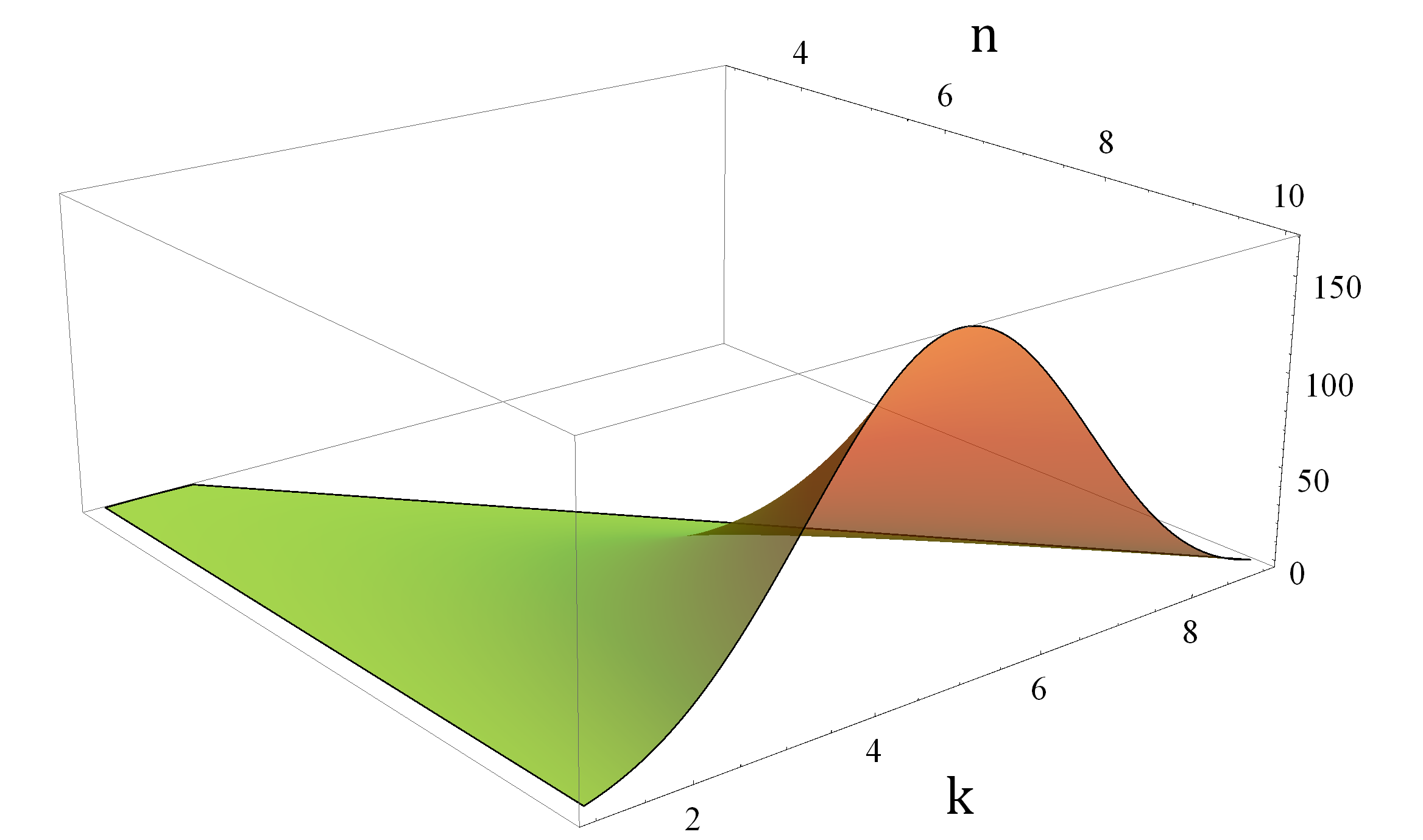}
	\subcaption{$\mathrm{D^{\star}_{KL}}(u_{\mathrm{KP}}||u)$.}
	\label{fig: unconstrained}
	\end{minipage}		
		\begin{minipage}{0.60\textwidth}
	\includegraphics[width=0.80\textwidth]{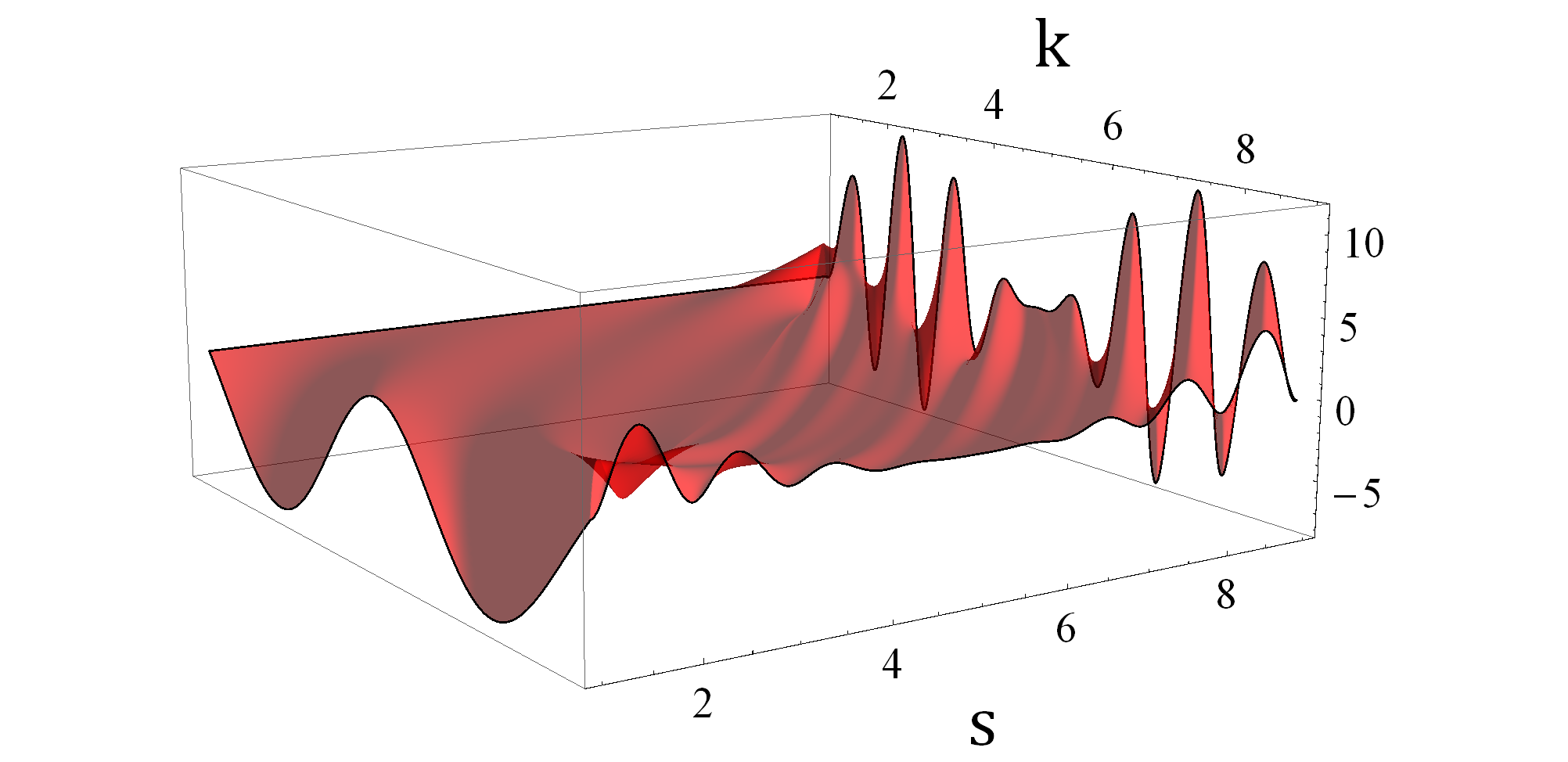}
	\subcaption{$\mathrm{D^{s}_{KL}}(u_{\mathrm{KP}}||u)-\mathrm{D^{\star}_{KL}}(u_{\mathrm{KP}}||u)$ at $n=10$.}
	\label{fig: difference}
	\end{minipage}
\caption{Behaviour of $\mathrm{D^{\star}_{KL}}(u_{\mathrm{KP}}||u)$, and its comparison with $\mathrm{D^{s}_{KL}}(u_{\mathrm{KP}}||u)$ at $n=10$.}
\label{fig: KL divergence, constrained number of - signs}
\end{figure}
 
The gain or loss of information can be quantified by the difference
$\mathrm{D^{s}_{KL}}(u_{\mathrm{KP}}||u)-\mathrm{D^{\star}_{KL}}(u_{\mathrm{KP}}||u)$ (see Figure \ref{fig: difference}). It depends on the values $n,k,s$, as shown in the following example. 
\begin{example}
\label{exa: relative contributions entropy, generic case} Let us
consider the generic case $\mathfrak{G}=\mathcal{P}_{k}[n]$. At $(n,k,s):=(7,3,1)$
we find that $G=35$, $\Omega(7,3,1)=15$, $\mathrm{D_{KL}^{\star}}=(35-7)\cdot\log2\approx19.4081$
and $\mathrm{D_{KL}}(1)=\ln\left(\frac{1}{7}\cdot{{35}\choose{15}}\right)\approx19.9554$,
so $\mathrm{D_{KL}^{\star}-D_{KL}}(1)<0$. 

On the other hand, moving to the case $(n,k,s)=(9,4,4)$, the quantities
$G=126$, $\Omega(9,4,4)=60$, $\mathrm{D_{KL}^{\star}}=(126-9)\cdot\log2\approx81.0982$
and $\mathrm{D_{KL}}(4)=\ln\left({{9}\choose{4}}^{-1}\cdot{{126}\choose{60}}\right)\approx79.7126$,
so $\mathrm{D_{KL}^{\star}-D_{KL}}(4)>0$. 
\end{example}
In contrast to the generic case, situations where $\mathfrak{G}\neq\mathcal{P}_{k}[n]$
may exhibit different combinatorial features, and the equivalence
between solitonic signatures (\ref{eq: choice of signs}) and the
induced representations (\ref{eq: sign switch operator, A}) is weakened.
In particular, the constraint on the number of negative signs $\#\boldsymbol{\ensuremath{\mathrm{NS}}}$
in $\Sigma$ does not correspond to a fixed number of signs $s$ for
$\mathcal{S}$ in (\ref{eq: sign switch operator, A}), and vice versa.
Furthermore, if one assumes a uniform prior distributions on $\mathcal{P}_{s}[n]$
in the construction of solitonic signatures by the sender, then it
is natural to assume that the solitonic signatures are not equiprobable.
In fact, moving from signatures in $\{\pm1\}^{\mathfrak{G}}$ to elements
of $\mathcal{P}_{s}[n]$, the assumption that choices for $\mathcal{S}\in\mathcal{P}_{s}[n]$
occur with the same probability ${{n}\choose{s}}^{-1}$ endows the associated
$G$-bits solitonic strings with multiplicity weights. Each weight
expresses the redundancy of the representation (\ref{eq: sign switch operator, A})
counting the number of subsets $\mathcal{S}\in\mathcal{P}_{s}[n]$
that induce the same signature. At fixed $s$, these weights do not
coincide for all the signatures in general. For instance, if there
is no subset $\mathfrak{p}\subseteq[P]$ such that 
\begin{equation}
\sum_{q\in\mathfrak{p}}k_{q}=2\sum_{q\in\mathfrak{p}}s_{q},\label{eq: balancing of incidences for sub-family}
\end{equation}
then each involution (\ref{eq: involutions for classes}) makes the
index $s$ of the signature change. On the other hand, if (\ref{eq: balancing of incidences for sub-family})
is satisfied, then both $\mathcal{S}$ and $\mathcal{S}\Delta\bigcup_{q\in\mathfrak{p}}\mathcal{H}_{q}$
have $s$ negative signs and induce the same signature. We clarify
these issues through the following examples.
\begin{example}
\label{exa: Non-generic case, a} Consider the matrix 
\begin{equation}
\mathbf{P_{1}}:=
\left(\begin{array}{ccccccc}
1 & 1 & 1 & 0 & 0 & 0 & 0\\
0 & 0 & 0 & 1 & 1 & 0 & 0\\
0 & 0 & 0 & 0 & 0 & 1 & 1
\end{array}\right).\label{eq: matrix for non-generic case example a}
\end{equation}
\end{example}
The non-vanishing minors for $\mathbf{P_{1}}$ can be indexed by elements
in $\{1,2,3\}\times\{4,5\}\times\{6,7\}$, so $G=12$, while it is
easily seen that $P=3$. Not all the choices in $\mathcal{P}_{2}[7]\times\{\pm1\}$
correspond to distinct signatures, due to coincidences from (\ref{eq: involutions for classes})
and $R\in\{\pm1\}$. In fact, there are $22$ distinct solitonic signatures
associated with as many classes in a partition of $\mathcal{P}_{2}[n]$,
where each class includes all the subsets $\mathcal{S}$ that returns
the same signature. Assuming a uniform distribution on $\mathcal{P}_{2}[7]\times\{\pm1\}$,
we get the following induced distribution for solitonic signatures
\begin{eqnarray}
& & \left\{ \frac{1}{42},\frac{1}{42},\frac{1}{42},\frac{1}{42},\frac{1}{21},\frac{1}{21},\frac{1}{21},\frac{1}{21},\frac{1}{42},\frac{1}{42},\right.\nonumber\\
& & \left.\frac{1}{21},\frac{1}{21},\frac{1}{21},\frac{1}{21},\frac{1}{21},\frac{1}{21},\frac{1}{21},\frac{1}{21},\frac{1}{21},\frac{1}{21},\frac{2}{21},\frac{2}{21}\right\} .\label{eq: induced distribution for solitonic signatures, example a}
\end{eqnarray}
The Kullback-Leibler divergence from the uniform prior over $2^{\mathfrak{G}}$
to (\ref{eq: induced distribution for solitonic signatures, example a})
is $\mathrm{D_{KL}(\mathbf{P_{1}};}s=2)\approx5.30625$. Note that
the relative entropy, starting from $\mathbf{P_{1}}$ but without
the information on $s=2$, is $7\cdot\ln2\approx4.85203<\mathrm{D_{KL}(\mathbf{P_{1}};}s=2)$
by (\ref{eq:eq: KL divergence, general including non-generic}). On
the other hand, using (\ref{eq: KL divergence, general and generic})
and (\ref{eq: KL divergence, constrained number of - signs and generic}),
we find that $\mathrm{D_{KL,unconstrained}\approx19.4081}$ and $\mathrm{D_{KL,s=2}\approx18.8568}$
in the generic case. Thus, while the unconstrained case is favourable
in the generic case, the unconstrained one may be preferred when $\mathfrak{G}\neq\mathcal{P}_{k}[n]$. 
\begin{example}
\label{exa: Non-generic case, b} Consider the matrix 
\begin{equation}
\mathbf{P_{2}}:=
\left(\begin{array}{cccccccccc}
1 & 1 & 1 & 1 & 0 & 0 & 0 & 0 & 0 & 0\\
1 & 2 & 3 & 4 & 0 & 0 & 0 & 0 & 0 & 0\\
0 & 0 & 0 & 0 & 1 & 1 & 1 & 0 & 0 & 0\\
0 & 0 & 0 & 0 & 1 & 2 & 3 & 0 & 0 & 0\\
0 & 0 & 0 & 0 & 0 & 0 & 0 & 1 & 1 & 1\\
0 & 0 & 0 & 0 & 0 & 0 & 0 & 1 & 2 & 3
\end{array}\right).\label{eq: matrix for non-generic case example b}
\end{equation}
It is easily seen that $\mathbf{P_{1}}$ is a rectangular block matrix
whose three blocks are full-rank Vandermonde matrices, call them $\mathbf{k_{1}},\mathbf{k_{2}},\mathbf{k_{3}}$.
Thus one can enumerate the non-vanishing minors of $\mathbf{P_{1}}$
as in the previous example and find $G={{4}\choose{2}}\cdot{{3}\choose{2}}^{2}=54$. 

First we note that the constraints on $\#\mathcal{S}$ in (\ref{eq: sign switch operator, A})
and on the negative signs in $\Sigma$ do not match: at $s=2$, the
choice $\mathcal{S}=\{1,2\}$ returns a signature $\Sigma$ with $36$
negative signs, while $\mathcal{S}=\{1,5\}$ gives $27$ negative
signs. The operation (\ref{eq: row normalization}) preserves the
parity $\#\boldsymbol{\ensuremath{\mathrm{NS}}}$, since $G$ is even, so the choices $\{1,2\}$
and $\{1,5\}$ still have different parity when the choice for $R$
is included. Vice versa, a fixed number of negative signs for solitonic
signatures does not correspond to the same constraint for $\#\mathcal{S}$:
both $\mathcal{S}_{1}=\{1,5\}$ and $\mathcal{S}_{2}=\{1,5,9\}$ generate
signatures with $27$ negative signs, despite having different cardinalities. 

Let us focus on the case $s=2$. Before undertaking any check on the
signature, we choose a uniform prior on the $2^{G}=2^{54}$ possible
ones. In contrast to Example \ref{exa: Non-generic case, a}, distinct
choices of $R$ produce distinct signatures: indeed, two subsets $\mathcal{S}_{1},\mathcal{S}_{2}\in\mathcal{P}_{s}[n]$
inducing opposite signatures $\Sigma_{1}=-\Sigma_{2}$ has to be related
by the action of some involutions (\ref{eq: involutions for classes})
respecting (\ref{eq: balancing of incidences for sub-family}). But
any operation of this type returns the same signature, since the ranks
of the matrices $\mathbf{k_{1}},\mathbf{k_{2}},\mathbf{k_{3}}$ are
all even. However, also in this case the uniform distribution on $\mathcal{P}_{2}[n]$
induces a non-uniform distribution on the solitonic signatures: for
instance, the signature associated with $\{8,9\}$ does not coincides
with other ones, while $\{1,2\}$ and $\{3,4\}$ lie in the same equivalence
class, and this results in different weights. Computing the relative
entropy as in the previous example, one gets $\mathrm{D_{KL}}(\mathbf{P_{2}};s=2)\approx33.0226$,
while the unconstrained case gives $46\cdot\ln(2)$. 
\end{example}
The features of non-generic cases and the dependence on the parameters
$(n,k,s)$ can be used to adapt the amount of information content
in the two cases of unconstrained and fixed cardinality for $\mathcal{S}$
in (\ref{eq: sign switch operator, A}). A specific analysis in this
regard will be carried on in a separate work. 

\subsection{\label{subsec: Intersection property and its geometric interpretation} Intersection property and its geometric interpretation}

The introduction of the families  
\begin{equation}
\mathfrak{N}(\boldsymbol{x}):=\left\{ \mathcal{S}\in\mathcal{P}[n]:\,\tau_{\mathcal{S}}(\boldsymbol{x})<0\right\} ,\quad\mathfrak{N}_{s}:=\mathfrak{N}\cap\mathcal{P}_{s}[n]\label{eq: set of negatives}
\end{equation}
plays a significant role in the identification of the stratified structure of statistical amoebas at $k=1$. This comes from a simple combinatorial property (see Proposition 5 in \cite{AK2016b}), which can be extended to the case $k>1$ as follows. 
 
\begin{prop}
\label{prop: bound pairwise disjoint negatives} If $\Delta_{\mathbf{A}}(\mathcal{I})\cdot\Delta_{\mathbf{K}}(\mathcal{I})\geq0$
for all $\mathcal{I}\in\mathcal{P}_{k}[n]$, then there are no $2k$
pairwise disjoint sets in $\mathfrak{N}(\boldsymbol{x})$. 
\end{prop}
\begin{proof}
Let us suppose that such $2k$ sets exist, i.e., $\{\mathcal{I}_{1},\dots\mathcal{I}_{2k}\}\subseteq\mathfrak{N}(\boldsymbol{x})$
such that $\mathcal{I}_{a}\cap\mathcal{I}_{b}=\emptyset$ for all
$a\neq b$. By definition, $\tau_{\mathcal{I}_{a}}(\boldsymbol{x})<0$
is equivalent to the inequalities 
\begin{equation}
\sum_{\mathcal{H}\parallel\mathcal{I}_{a}}^{\mathcal{H}\in\mathcal{P}_{k}[n]}\Lambda_{\mathcal{H}}(\boldsymbol{x})>\sum_{\mathcal{K}\perp\mathcal{I}_{a}}^{\mathcal{K}\in\mathcal{P}_{k}[n]}\Lambda_{\mathcal{K}}(\boldsymbol{x})\quad a\in[2k].\label{eq: inequality odd-even}
\end{equation}
Let us divide these $2k$ subsets in two classes $\mathfrak{F}_{1}:=\{\mathcal{I}_{a}:\,1\leq a\leq k\}$
and $\mathfrak{F}_{2}:=\{\mathcal{I}_{a}:\,k+1\leq a\leq2k\}$. Adding
the inequalities (\ref{eq: inequality odd-even}) for $\mathfrak{F}_{1}$
term by term, one gets 
\begin{equation}
\sum_{w=0}^{\left\lfloor \frac{k}{2}\right\rfloor }(k-2w)\cdot\sum_{\#(\mathcal{H}\parallel\mathfrak{F}_{1})=k-w}^{\mathcal{H}\in\mathcal{P}_{k}[n]}\Lambda_{\mathcal{H}}(\boldsymbol{x})>\sum_{w=0}^{\left\lfloor \frac{k}{2}\right\rfloor }(k-2w)\cdot\sum_{\#(\mathcal{K}\perp\mathfrak{F}_{1})=k-w}^{\mathcal{K}\in\mathcal{P}_{k}[n]}\Lambda_{\mathcal{K}}(\boldsymbol{x}).\label{eq: sum odd-even, F1}
\end{equation}
Similarly, for $\mathfrak{F}_{2}$ one has 
\begin{equation}
\sum_{w=0}^{\left\lfloor \frac{k}{2}\right\rfloor }(k-2w)\cdot\sum_{\#(\mathcal{H}\parallel\mathfrak{F}_{2})=k-w}^{\mathcal{H}\in\mathcal{P}_{k}[n]}\Lambda_{\mathcal{H}}(\boldsymbol{x})>\sum_{w=0}^{\left\lfloor \frac{k}{2}\right\rfloor }(k-2w)\cdot\sum_{\#(\mathcal{K}\perp\mathfrak{F}_{2})=k-w}^{\mathcal{K}\in\mathcal{P}_{k}[n]}\Lambda_{\mathcal{K}}(\boldsymbol{x}).\label{eq: sum odd-even, F2}
\end{equation}
If $\Lambda_{\mathcal{H}}(\boldsymbol{x})\neq0$ appears in the left hand
side of (\ref{eq: sum odd-even, F1}) (respectively (\ref{eq: sum odd-even, F2})),
then it has non-empty intersection with more than $\left\lfloor \frac{k}{2}\right\rfloor $
elements in $\mathfrak{F}_{1}$ (respectively, $\mathfrak{F}_{2})$.
In the same way, if $\mathcal{H}\parallel\mathfrak{F}_{2}=\{\mathcal{I}_{b},\,b\in[k-w]\}$
with $w<\frac{k}{2}$ and distinct $\mathcal{I}_{b}$, then $w_{1}:=\#(\mathcal{H}\parallel\mathfrak{F}_{1})\leq w<\frac{k}{2}$.
So $\#\left(\mathcal{H}\perp\mathfrak{F}_{1}\right)=k-w_{1}$ and
the coefficient of $\Lambda_{\mathcal{H}}(\boldsymbol{x})$ in the right
hand side of (\ref{eq: sum odd-even, F1}) is $k-2w_{1}\geq k-2w$.
Since $\Lambda_{\mathcal{H}}(\boldsymbol{x})\geq0$ by hypothesis, this
means that the right hand side of (\ref{eq: sum odd-even, F1}) is
an upper bound for the left hand side of (\ref{eq: sum odd-even, F2}),
which implies 
\begin{equation}
\sum_{w=0}^{\left\lfloor \frac{k}{2}\right\rfloor }(k-2w)\cdot\sum_{\#(\mathcal{H}\parallel\mathfrak{F}_{1})=k-w}^{\mathcal{H}\in\mathcal{P}_{k}[n]}\Lambda_{\mathcal{H}}(\boldsymbol{x})>\sum_{w=0}^{\left\lfloor \frac{k}{2}\right\rfloor }(k-2w)\cdot\sum_{\#(\mathcal{H}\parallel\mathfrak{F}_{2})=k-w}^{\mathcal{H}\in\mathcal{P}_{k}[n]}\Lambda_{\mathcal{H}}(\boldsymbol{x})\label{eq: lowering RHS, a}
\end{equation}
Arguing in the same way for (\ref{eq: sum odd-even, F2}), one gets
\begin{equation}
\sum_{w=0}^{\left\lfloor \frac{k}{2}\right\rfloor }(k-2w)\cdot\sum_{\#(\mathcal{K}\parallel\mathfrak{F}_{2})=k-w}^{\mathcal{H}\in\mathcal{P}_{k}[n]}\Lambda_{\mathcal{H}}(\boldsymbol{x})>\sum_{w=0}^{\left\lfloor \frac{k}{2}\right\rfloor }(k-2w)\cdot\sum_{\#(\mathcal{H}\parallel\mathfrak{F}_{1})=k-w}^{\mathcal{H}\in\mathcal{P}_{k}[n]}\Lambda_{\mathcal{H}}(\boldsymbol{x})\label{eq: lowering RHS, b}
\end{equation}
which is incompatible with (\ref{eq: lowering RHS, a}), i.e., a contradiction. 
\end{proof}

The previous bound also holds in the restriction from $\mathfrak{N}(\boldsymbol{x})$ to each individual family $\mathfrak{N}_{s}(\boldsymbol{x})$, and this provides an extension of a geometric property that can be stressed in the case $s=1$ to higher levels $s>1$. At this purpose, we introduce the matrices 
\begin{equation}
\boldsymbol{\ensuremath{\zeta}}_{\alpha}:=\vec{e}_{\alpha}\cdot\vec{e}_{\alpha}^{T}=(\delta_{\beta\gamma}\cdot\delta_{\alpha\beta})_{\beta,\gamma\in[n]},\quad\alpha\in[n]\label{eq: rank-1}
\end{equation}
where $\left\{ \vec{e}_{\alpha}:\,\alpha\in[n]\right\} $ is the standard
basis for $\mathbb{R}^{n}$, and 
\begin{equation}
\mathbf{L}:=(\mathbf{A}\cdot\boldsymbol{\ensuremath{\Theta}}(\boldsymbol{x})\cdot\mathbf{K})^{-1}\cdot\mathbf{A}\cdot\boldsymbol{\ensuremath{\Theta}}(\boldsymbol{x}).\label{eq: left-inverse}
\end{equation}
Note that $(\mathbf{A}\cdot\boldsymbol{\ensuremath{\Theta}}(\boldsymbol{x})\cdot\mathbf{K})^{-1}$
exists at $\det(\mathbf{A}\cdot\boldsymbol{\ensuremath{\Theta}}(\boldsymbol{x})\cdot\mathbf{K})\neq0$
(i.e., outside the singular locus), and $\mathbf{L}$ is a left-inverse
of $\mathbf{K}$. In particular, $(\mathbf{K}\cdot\mathbf{L})^{2}=\mathbf{K}\cdot(\mathbf{L}\cdot\mathbf{K})\cdot\mathbf{L}=\mathbf{K}\cdot\idd_{k}\cdot\mathbf{L}=\mathbf{K}\cdot\mathbf{L}$,
so $\mathbf{K}\cdot\mathbf{L}$ is idempotent. Then, the role played
by $\alpha\in[n]$ in the behaviour of (\ref{eq: Cauchy-Binet as partition function})
can be assessed by the sign of 
\begin{eqnarray}
\frac{\tau_{\{\alpha\}}(\boldsymbol{x})}{\tau(\boldsymbol{x})} & = & \det((\mathbf{A}\cdot\boldsymbol{\ensuremath{\Theta}}(\boldsymbol{x})\cdot\mathbf{K})^{-1})\cdot\det(\mathbf{A}\cdot\boldsymbol{\ensuremath{\sigma}}_{\alpha}\cdot\boldsymbol{\ensuremath{\Theta}}(\boldsymbol{x})\cdot\mathbf{K})\nonumber \\
 & = & \det\left((\mathbf{A}\cdot\boldsymbol{\ensuremath{\Theta}}(\boldsymbol{x})\cdot\mathbf{K})^{-1}\cdot\mathbf{A}\cdot\boldsymbol{\ensuremath{\Theta}}(\boldsymbol{x})\cdot(\idd_{n}-2\cdot\boldsymbol{\ensuremath{\zeta}}_{\alpha})\cdot\mathbf{K}\right)\nonumber \\
 & = & \det(\idd_{k}-2\mathbf{L}\cdot\boldsymbol{\ensuremath{\zeta}}_{\alpha}\cdot\mathbf{K})\label{eq: sign tau as rank-1 correction, a}
\end{eqnarray}
Taking into account that $\boldsymbol{\ensuremath{\zeta}}_{\alpha}^{2}=\boldsymbol{\ensuremath{\zeta}}_{\alpha}$,
one can apply Sylvester's determinant identity \cite{Gantmacher1977}
twice and get 
\begin{eqnarray}
\det(\idd_{k}-2\mathbf{L}\cdot\boldsymbol{\ensuremath{\zeta}}_{\alpha}\cdot\mathbf{K}) & = & \det(\idd_{n}-2\cdot\boldsymbol{\ensuremath{\zeta}}_{\alpha}\cdot\mathbf{K}\cdot\mathbf{L})\nonumber \\
 & = & \det(\idd_{n}-2\cdot\boldsymbol{\ensuremath{\zeta}}_{\alpha}^{2}\cdot\mathbf{K}\cdot\mathbf{L})\nonumber \\
 & = & \det(\idd_{n}-2\cdot\boldsymbol{\ensuremath{\zeta}}_{\alpha}\cdot\mathbf{K}\cdot\mathbf{L}\cdot\boldsymbol{\ensuremath{\zeta}}_{\alpha}).\label{eq: sign tau as rank-1 correction, b}
\end{eqnarray}
Note that $\boldsymbol{\ensuremath{\zeta}}_{\alpha}\cdot\mathbf{K}\cdot\mathbf{L}\cdot\boldsymbol{\ensuremath{\zeta}}_{\alpha}=\langle\mathbf{K}_{\alpha}|\mathbf{L}^{\alpha}\rangle\cdot\vec{e}_{\alpha}\cdot\vec{e}_{\alpha}^{T}$,
where $\mathbf{K}_{\alpha}$ (respectively $\mathbf{L}^{\alpha}$)
is the vector corresponding to the $\alpha$th row of $\mathbf{K}$
(respectively column of $\mathbf{L}$) and $\langle\;|\;\rangle$
is the usual Euclidean scalar product. From the matrix determinant
lemma \cite{Gantmacher1977}, one finds 
\begin{equation}
\det(\idd_{n}-2\cdot\boldsymbol{\ensuremath{\zeta}}_{\alpha}\cdot\mathbf{K}\cdot\mathbf{L}\cdot\boldsymbol{\ensuremath{\zeta}}_{\alpha})=1-2\cdot\langle\mathbf{K}_{\alpha}|\mathbf{L}^{\alpha}\rangle\cdot(\vec{e}_{\alpha}\cdot\vec{e}_{\alpha}^{T})=1-2\langle\mathbf{K}_{\alpha}|\mathbf{L}^{\alpha}\rangle.\label{eq: sign tau as rank-1 correction}
\end{equation}
So $\frac{\tau_{\{\alpha\}}(\boldsymbol{x})}{\tau(\boldsymbol{x})}<0$ if
and only if $\vec{e}_{\alpha}^T\cdot(\mathbf{K}\mathbf{L}\vec{e}_{\alpha})=\langle\mathbf{K}_{\alpha}|\mathbf{L}^{\alpha}\rangle>\frac{1}{2}$.
Thus, the result in Proposition \ref{prop: bound pairwise disjoint negatives}
is equivalent to the property that there exist at most $2k-1$ diagonal
entries of $\mathbf{K}\cdot\mathbf{L}$ such that $\langle\mathbf{K}_{\alpha}|\mathbf{L}^{\alpha}\rangle>\frac{1}{2}$,
independently on $n$. 

\section{\label{sec: Conclusion and future perspectives} Discussion and future
perspectives}

In this work we investigated the effects of particular requirements
connected to a type of complexity reduction, namely determinantal
and integrability constraints, on real-valued partition functions.
Such a reduction can be observed through the statistical amoeba associated
with an initial sum of exponentials (\ref{eq: Cauchy-Binet formula, solitons})
fulfilling the given constraints. In such a framework, the family of allowed choices of signs for pre-exponential terms (\ref{eq: determinantal degeneration}) coincides with the signs induced by row/column sign choices for the coefficient matrix $\mathbf{A}$. In particular, the consistency with
the KP II equation returns $\tau$-functions for the whole KP hierarchy.
This led to the exploration of the number of distinct signatures for a general $\mathbf{A}$, levels of constrained statistical amoebas, and their applications in the information-theoretic and geometric settings.

These results give rise to questions on further links between the
combinatorics of complex structures and integrability, some of which
have already been pointed out. In particular, it is worth exploring
in more detail the redundancy in the description of signatures (\ref{eq: choice of signs})
through subsets (\ref{eq: sign switch operator, A}) and the concept
of instability domains. These issues are also related to the investigation of the tropical limit 
of constrained statistical amoebas, as briefly mentioned in Section \ref{sec: Levels of constrained amoebas}. In fact, these matrix models provide a natural framework for the realization of different tropical concepts, in particular for the nested tropical expansion and the tropical symmetry introduced in \cite{Angelelli2017}. The nested expansion relies on the extension of the ordering of individual phases $\varphi_{\alpha}(\boldsymbol{x})$ at points where $||\boldsymbol{x}||\rightarrow\infty$ to an ordering of  collective phases $\sum_{\alpha\in\mathcal{I}}\varphi_{\alpha}$, $\mathcal{I}\in\mathcal{P}_{k}[n]$, while the tropical copies of elements of $[n]$ can be represented as copies of columns of $\mathbf{A}$. Some additional remarks in this regard are given in Appendix \ref{sec: Tropical limit}. 
A careful analysis of these subjects could be useful in the development of concrete models for the thermodynamic and statistical systems mentioned in \cite{Angelelli2017}, hence it fits within the increasing number of applications of tropical techniques in the description of physical systems (see, e.g., \cite{Maeda2007,Inoue2012,Passare2012,DM-H2014}).

Besides the theoretical interest, the previous points prompt a search for new applications of soliton-like structures to the propagation of information. Furthermore, the concept of dimensionality reduction can be studied in more depth in the context of subspaces classification. Indeed,
the bounds discussed in Section \ref{subsec: Intersection property and its geometric interpretation} 
may be implemented for the purpose of statistical regression \cite{Hastie2009},
in particular when generalized/weighted least square methods are employed
(see, e.g., \cite{Ma2014}). These structures involve bounds for diagonal
elements of a projection matrix (also called leverages for orthogonal
projections) and might be applied in signal processing and machine
learning \cite{Behrens1994,Johansson2006}. More generally, the presentation
of statistical amoebas as families of partitions of the type (\ref{eq: partition from signature})
could be combined with cross-validation techniques. The links between these physical and information-theoretic concepts deserve further investigations for a better understanding, and they will be explored in more detail in a separate paper.

\appendix

\section{\label{sec: Determinantal constraints} Preservation of determinantal constraints}

\begin{rem}
\label{rem: model reduction} There exist sets $\mathfrak{v}\subseteq[n]$,
$R(\mathfrak{v})\subseteq[k]$, and $\boldsymbol{a}\in\mathbb{R}_{\star}^{(k-\#R(\mathfrak{v}))\times(n-\#\mathfrak{v})}$ that preserve the Cauchy-Binet expansion (\ref{eq: Cauchy-Binet as partition function}) up to a factor independent on $\mathcal{I}\in\mathcal{P}_{k}[n]$,
namely 
\begin{equation}
\tau(\boldsymbol{x})=C\cdot\exp\left(\sum_{\nu\in\mathfrak{v}}\varphi_{\nu}(\boldsymbol{x})\right)\cdot\det\left(\boldsymbol{a}\cdot\boldsymbol{\Theta}(\boldsymbol{x})_{\llbracket[n]\setminus\mathfrak{v};[n]\setminus\mathfrak{v}\rrbracket}\cdot\boldsymbol{K}_{\llbracket[n]\setminus\mathfrak{v};[k]\setminus R(\mathfrak{v})\rrbracket}\right)\label{eq: reduction matrix for factorization}
\end{equation}
with $C$ constant. 
\end{rem}
\begin{proof}
Given $\boldsymbol{A}\in\mathbb{R}_{\star}^{k\times n}$, let $\boldsymbol{A_{0}}$
be the reduced row-echelon form of $\boldsymbol{A}$ with $\boldsymbol{D}\cdot\boldsymbol{A_{0}}=\boldsymbol{A}$,
$\boldsymbol{D}\in GL_{k}(\mathbb{R})$. Introduce the set 
\begin{equation}
\mathfrak{v}:=\bigcap_{\mathcal{I}\in\mathfrak{G}}\mathcal{I}=\left\{ \alpha\in[n]:\,(\alpha\notin\mathcal{I}\Rightarrow\mathcal{I}\in\mathcal{P}_{k}[n]\setminus\mathfrak{G})\right\} .\label{eq: core}
\end{equation}
In particular, from $\mathcal{V}\in\mathfrak{G}$
one finds $\mathfrak{v}\subseteq\mathcal{V}$ , so we can consider
$R(\mathfrak{v}):=\{i\in[k]:\,\nu_{i}\in\mathfrak{v}\}$. The dependence
of $\mathfrak{v}$ and $R(\mathfrak{v})$ on $\boldsymbol{A}$ will
be implicit when no ambiguity arises. 

Let $\pi$ be the permutation of $[n]$ such that both the restrictions
$\left.\pi^{-1}\right|_{\mathfrak{v}}$ and $\left.\pi^{-1}\right|_{[n]\setminus\mathfrak{v}}$ are
increasingly monotone, and $\pi^{-1}(\alpha)<\pi^{-1}(\beta)$ for each $\alpha\in\mathfrak{v}$, $\beta\in[n]\setminus\mathfrak{v}$. Similarly, take the permutation $\varpi$
of $[k]$ such that $\left.\varpi^{-1}\right|_{R(\mathfrak{v})}$ and $\left.\varpi^{-1}\right|_{[k]\setminus R(\mathfrak{v})}$
are both order-preserving and, for any $i\in R(\mathfrak{v})$, $j\in[k]\setminus R(\mathfrak{v})$,
$\varpi^{-1}(i)<\varpi^{-1}(j)$. The action of $\pi$ on $\boldsymbol{A}$,
$\boldsymbol{\Theta}(\boldsymbol{x})$ and $\boldsymbol{K}$ via the
matrix representation $\boldsymbol{\pi}:=\left(\delta_{\pi(\alpha),\beta}\right)_{\alpha,\beta\in[n]}$
is 
\begin{equation}
\boldsymbol{A}\mapsto\boldsymbol{A}\cdot\boldsymbol{\pi}^{-1},\quad\boldsymbol{\Theta}(\boldsymbol{x})\mapsto\boldsymbol{\pi}\cdot\boldsymbol{\Theta}(\boldsymbol{x})\cdot\boldsymbol{\pi}^{-1},\quad\boldsymbol{K}\mapsto\boldsymbol{\pi}\cdot\boldsymbol{K},\label{eq: matrix action of permutation}
\end{equation}
which preserves the product $\boldsymbol{A}\cdot\boldsymbol{\Theta}(\boldsymbol{x})\cdot\boldsymbol{K}$
. The matrix $\boldsymbol{\pi}\cdot\boldsymbol{\Theta}(\boldsymbol{x})\cdot\boldsymbol{\pi}^{-1}$
is still diagonal, since $\boldsymbol{\pi}$ lies in the normalizer
of diagonal matrices, and its entries are the same of $\boldsymbol{\Theta}(\boldsymbol{x})$.
The parity of the number of inversions induced by $\pi$ on 
$\mathcal{I}\in\mathcal{P}_{k}[n]$ is the same for $\Delta\left(\boldsymbol{A}\cdot\boldsymbol{\pi}^{-1};\mathcal{I}\right)$
and $\Delta\left(\boldsymbol{\pi}\cdot\boldsymbol{K};\mathcal{I}\right)$.
So the action of $\pi$ preserves the terms in the Cauchy-Binet expansion
(\ref{eq: Cauchy-Binet as partition function}) up to their permutation
denoted by $\Lambda_{\mathcal{I}}(\boldsymbol{x})\mapsto\Lambda_{\pi(\mathcal{I})}(\boldsymbol{x})$.
Hence, we can express $g_{\mathcal{I}}$ using the order given by
$\pi$ for columns, and an additional relabelling of the rows via
the left action of the matrix $\boldsymbol{\varpi}$ representing
$\varpi$: 
\begin{eqnarray}
g_{\mathcal{I}} & = & \det(\boldsymbol{D})\cdot\det(\boldsymbol{\varpi})\cdot\det\left(\boldsymbol{A_{0}}_{\llbracket[k]\setminus R(\mathfrak{v});\mathcal{I}\setminus\mathfrak{v}\rrbracket}\right)\cdot\mathrm{VdM}(\boldsymbol{\kappa};\mathfrak{v}) \nonumber \\ 
& & \cdot\mathrm{VdM}(\boldsymbol{\kappa};\mathcal{I}\setminus\mathfrak{v})\cdot\prod_{\alpha\in\mathfrak{v}}\prod_{\beta\in\mathcal{I}\setminus\mathfrak{v}}(\kappa_{\beta}-\kappa_{\alpha}).\label{eq: factorization constant and hidden-reduced}
\end{eqnarray}
Introduce 
\begin{equation}
P_{\beta}:=\prod_{\alpha\in\mathfrak{v}}(\kappa_{\beta}-\kappa_{\alpha}),\quad\beta\in[n]\setminus\mathfrak{v}\label{eq: polynomial for hidden roots}
\end{equation}
which are non-vanishing under the hypothesis of distinct parameters
$\boldsymbol{\kappa}$. So the matrix 
\begin{equation}
\boldsymbol{a}:=\mathrm{diag}\left(P_{\nu}^{-1}:\,\nu\in\mathcal{V}\setminus\mathfrak{v}\right)\cdot\boldsymbol{A_{0}}_{\llbracket[k]\setminus R(\mathfrak{v});[n]\setminus\mathfrak{v}\rrbracket}\cdot\mathrm{diag}\left(P_{\alpha}:\,\alpha\in[n]\setminus\mathfrak{v}\right)\label{eq: matrix corrected with polynomials}
\end{equation}
is well-defined. One can reformulate (\ref{eq: factorization constant and hidden-reduced})
as 
\begin{equation}
g_{\mathcal{I}}=C\cdot\Delta_{\boldsymbol{a}}(\mathcal{I}\setminus\mathfrak{v})\cdot\mathrm{VdM}(\boldsymbol{\kappa};\mathcal{I}\setminus\mathfrak{v})\label{eq: reduction matrix without core}
\end{equation}
where 
\begin{equation}
C:=\det(\boldsymbol{D})\cdot\det(\boldsymbol{\varpi})\cdot\left(\prod_{\nu\in\mathcal{V}\setminus\mathfrak{v}}P_{\nu}\right)\cdot\mathrm{VdM}(\boldsymbol{\kappa};\mathfrak{v}).\label{eq: constant factor core}
\end{equation}
does not depend on $\mathcal{I}$. Since the soliton parameters
$\boldsymbol{\kappa}$ are pairwise distinct, all the terms in (\ref{eq: constant factor core})
are non-vanishing by definition, thus $C\neq0$. All the non-vanishing
terms in (\ref{eq: Cauchy-Binet as partition function}) have a common
factor $C\cdot\exp\left(\sum_{\nu\in\mathfrak{v}}\varphi_{\nu}(\boldsymbol{x})\right)$,
and the Cauchy-Binet expansion, along with (\ref{eq: reduction matrix without core})
and the correspondence 
\begin{equation}
\mathcal{I}\mapsto\mathcal{I}\setminus\mathfrak{v},\quad\mathcal{I}\in\mathfrak{G},\label{eq: reductions, subsets without core}
\end{equation}
allows to formulate the $\tau(\boldsymbol{x})$ as in (\ref{eq: reduction matrix for factorization}). 

The erasure of rows and columns associated with indices $\nu\in\mathfrak{v}$
in (\ref{eq: reduction matrix without core}) preserves the reduced
row-echelon form: this is inherited by $\boldsymbol{a}$ when the
normalization $\mathrm{diag}\left(P_{\nu}^{-1}:\right.$ $\left.\nu\in\mathcal{V}\setminus\mathfrak{v}\right)$
in (\ref{eq: matrix corrected with polynomials}) is taken into account.
The equality (\ref{eq: reduction matrix for factorization}) also
implies that $\boldsymbol{a}$ has maximal rank and no null columns.
Note that $\#(\mathcal{I}\setminus\mathfrak{v})=k-\#\mathfrak{v}$
for all the terms $\mathcal{I}\in\mathfrak{G}$.
In particular, $\mathcal{V}\setminus\mathfrak{v}$ is still the minimum
element of $\mathcal{P}_{k}\left([n]\setminus\mathfrak{v}\right)$
associated with a non-vanishing minor of $\boldsymbol{a}$ with respect
to the lexicographic order on $[n]\setminus\mathfrak{v}$ induced
by $[n]$. Accordingly, a signature $\Sigma$ for $\tau$ induces
a signature for the reduced model (\ref{eq: reduction matrix for factorization}),
which will still be denoted by $\Sigma$ 
\begin{equation}
\Sigma(\mathcal{I}\setminus\mathfrak{v}):=\mathrm{sign}(C)\cdot\Sigma(\mathcal{I}).\label{eq: induced signature on reduced model}
\end{equation} 
\end{proof}

As remarked in the Introduction, we focus on Cauchy-Binet
expansions that generate statistical amoebas relative to the exponential
functions in (\ref{eq: Cauchy-Binet as partition function}), rather
than determinants. Indeed, each function $f(\boldsymbol{x})$ can
be trivially expressed as the determinant of a diagonal matrix $\mathrm{diag}\left(1,\dots,1,f(\boldsymbol{x})\right)$,
and no constraints arise. Also the number of degrees of freedom, which
is related to the dimensions of the matrices involved in the expansion,
has to be bounded in order to get non-trivial constraints. In fact,
any sum of functions $\sum_{t=1}^{W}f_{t}(\boldsymbol{x})$, $W\in\mathbb{N}$,
can be expressed via (\ref{eq: Cauchy-Binet formula, solitons}) as
\begin{equation}
\det\left[\left(\left(f_{1}(\boldsymbol{x})\, \dots \, f_{W}(\boldsymbol{x})\right) \oplus\idd_{\ell}\right)\cdot\left(\left(1 \, \dots \, 1\right)^{T}\oplus\idd_{\ell}\right)\right].\label{eq: example dimension to get constraints}
\end{equation}
This leads us to define determinantal choices of signs, i.e. signatures
preserving the determinantal expansion (\ref{eq: Cauchy-Binet formula, solitons, a}),
as follows. 
\begin{defn}
\label{def: determinantal choices of signs} We call
a signature $\Sigma$ a \textit{determinantal choice
of signs}, acting on a determinant (\ref{eq: Cauchy-Binet formula, solitons, a}),
if there exist a set $\{\tilde{\kappa}_{\alpha}:\,\alpha\in[n]\setminus\mathfrak{v}\}$
and a matrix $\boldsymbol{\tilde{a}}\in\mathbb{R}^{\#(k-\mathfrak{v})\times\#(n-\mathfrak{v})}$
such that $\Sigma$ is induced by these data through (\ref{eq: Cauchy-Binet formula, solitons}),
up to a common scale factor $\lambda(\boldsymbol{x})$. Specifically,
this means that 
\begin{equation}
\Delta_{\boldsymbol{\tilde{a}}}(\mathcal{I}\setminus\mathfrak{v})\cdot\Delta_{\boldsymbol{\tilde{k}}}(\mathcal{I}\setminus\mathfrak{v})\cdot\exp\left(\sum_{\alpha\in\mathcal{I}\setminus\mathfrak{v}}\tilde{\varphi}_{\alpha}(\boldsymbol{x})\right)=\lambda(\boldsymbol{x})\cdot\Sigma(\mathcal{I})\cdot\Delta_{\boldsymbol{A}}(\mathcal{I})\cdot\Delta_{\boldsymbol{K}}(\mathcal{I})\cdot\exp\left(\sum_{\alpha\in\mathcal{I}}\varphi_{\alpha}(\boldsymbol{x})\right),\quad\mathcal{I}\in\mathfrak{G}\label{eq: determinantal choice of signs}
\end{equation}
where $\tilde{\varphi}_{\alpha}(\boldsymbol{x}):=\sum_{u=1}^{d}\tilde{\kappa}_{\alpha}^{u}x_{u}$
and $\lambda(\boldsymbol{x})\neq0$. 
\end{defn}
This definition only relies on the data provided by
the terms in the expansion (\ref{eq: Cauchy-Binet formula, solitons, a})
indexed by $\mathfrak{G}$, as it is shown in the next lemma. 
\begin{lem}
\label{lem: determinantal choice of signs preserves soliton parameters}
Assuming that the components of $\boldsymbol{\kappa}$ are pairwise
distinct, the data $n-\#\mathfrak{v}$, $k-\#\mathfrak{v}$
and $\{\kappa_{\alpha}:\,\alpha\in[n]\setminus\mathfrak{v}\}$ are
uniquely determined. Furthermore, a determinantal choice of signs
preserves the absolute values of the maximal minors $\Delta_{\boldsymbol{a}}(\mathcal{I})$
of the matrix $\boldsymbol{a}$ defined in (\ref{eq: matrix corrected with polynomials}),
up to a multiplicative factor independent on $\mathcal{I}$. 
\end{lem}
\begin{proof}
Let us start from $\boldsymbol{A}$ and the associated sets $\mathfrak{G}$
and $\mathfrak{v}$, and take any $\alpha\in[n]\setminus\mathfrak{v}$.
From the lack of null columns, there exist $\mathcal{I},\mathcal{J}\in\mathfrak{G}$
with $\alpha\in\mathcal{I}\setminus\mathcal{J}$. The exchange property
(\ref{eq: exchange relation}) implies that there exists $\beta\in\mathcal{J}\setminus\mathcal{I}$
with $\mathcal{I},\mathcal{I}_{\beta}^{\alpha}\in\mathfrak{G}$. We
are looking at transformations that preserve each exponential ${\displaystyle \exp\left(\sum_{\alpha\in\mathcal{I}}\sum_{u=1}^{d}\kappa_{\alpha}^{u}x_{u}\right)}$
in (\ref{eq: Cauchy-Binet as partition function}) with $\mathcal{I}\in\mathfrak{G}$,
then the ratio of the exponentials corresponding to $\mathcal{I}$
and $\mathcal{I}_{\beta}^{\alpha}$ 
\begin{equation}
\exp\left(\sum_{u=1}^{d}(\kappa_{\alpha}^{u}-\kappa_{\beta}^{u})\cdot x_{u}\right)=\exp\left(\sum_{\gamma\in\mathcal{I}}\sum_{u=1}^{d}\kappa_{\gamma}^{u}x_{u}\right)^{-1}\cdot\exp\left(\sum_{\delta\in\mathcal{I}_{\beta}^{\alpha}}\sum_{u=1}^{d}\kappa_{\delta}^{u}x_{u}\right)\label{eq: fixed ration exponentials}
\end{equation}
is preserved too. In particular, the coefficients $\kappa_{\alpha}-\kappa_{\beta}$
of $x_{1}$ and $\kappa_{\alpha}^{2}-\kappa_{\beta}^{2}$ of $x_{2}$
are left unchanged. The assumption $\kappa_{\alpha}\neq\kappa_{\beta}$
for all $\alpha\neq\beta$ implies that both the quantities $\kappa_{\alpha}-\kappa_{\beta}$
and $\frac{\kappa_{\alpha}^{2}-\kappa_{\beta}^{2}}{\kappa_{\alpha}-\kappa_{\beta}}=\kappa_{\alpha}+\kappa_{\beta}$
are well-defined and fixed. This means that we can recover the values
of data $\kappa_{\alpha}$ (and $\kappa_{\beta}$) for all $\alpha\in[n]\setminus\mathfrak{v}$
from the form of the exponential terms. 

Now look at another full-rank matrix $\boldsymbol{A_{1}}\in\mathbb{R}^{k_{1}\times n_{1}}$
without null columns, which generates the data $\mathfrak{G}_{1}$
as in (\ref{eq: set of vanishing minors}) and $\mathfrak{v}_{1}$
as in (\ref{eq: core}), and a vector $\boldsymbol{\kappa_{1}}\in\mathbb{R}^{n_{1}}$
such that (\ref{eq: determinantal choice of signs}) also holds after
the substitution $\boldsymbol{A}\mapsto\boldsymbol{A_{1}}$ and $\boldsymbol{\kappa}\mapsto\boldsymbol{\kappa_{1}}$.
Same as above, for each $\alpha\in[n_{1}]\setminus\mathfrak{v}_{1}$
we can find $\mathcal{I}\in\mathfrak{G}_{1}$ and $\beta\in[n_{1}]\setminus\mathcal{I}$
such that $\alpha\in\mathcal{I}$ and $\mathcal{I}_{\beta}^{\alpha}\in\mathfrak{G}_{1}$
and, from the previous observations, we recover the same set $\{\tilde{\kappa}_{\alpha}:\,\alpha\in[n]\setminus\mathfrak{v}\}=\{\tilde{\kappa}_{\alpha}:\,\alpha\in[n_{1}]\setminus\mathfrak{v}_{1}\}$
of distinguishable soliton parameters. This establishes a correspondence
$w:\,\mathfrak{G}\longrightarrow\mathfrak{G}_{1}$ defined by $\mathcal{I}\setminus\mathfrak{v}=w(\mathcal{I})\setminus\mathfrak{v}_{1}$
for all $\mathcal{I}\in\mathfrak{G}$. 

For each $\alpha\in[n]\setminus\mathfrak{v}$, we can choose $\mathcal{I}(\alpha)\in\mathfrak{G}$
with $\alpha\in\mathcal{I}(\alpha)$, which exists by the lack of
null columns in $\boldsymbol{A}$, then $\alpha\in w(\mathcal{I}(\alpha))$
too. From (\ref{eq: reduction matrix without core})
and (\ref{eq: determinantal choice of signs}), this means that $\Delta_{\boldsymbol{A_{1}}}(w(\mathcal{I}(\alpha)))\cdot\Delta_{\boldsymbol{K_{1}}}(w(\mathcal{I}(\alpha)))\neq0$, which implies that the parameters $(\boldsymbol{\kappa_{1}})_{\nu}$,
$\nu\in\mathfrak{v}_{1}$, and $(\boldsymbol{\kappa_{1}})_{\alpha}$
are pairwise distinct. Since this holds for all $\alpha\in[n]\setminus\mathfrak{v}$,
all the components of $\boldsymbol{\tilde{\kappa}}$ are pairwise
distinct too. So one has $\Delta_{\boldsymbol{A}}(\mathcal{I})=0$
if and only if $\Delta_{\boldsymbol{A_{1}}}(w(\mathcal{I}))=0$. These
data induce the same form (\ref{eq: reduction matrix for factorization})
with matrices $\boldsymbol{a}$ and $\boldsymbol{a_{1}}$ and multiplicative
constants $C,C_{1}\neq0$. From the previous discussion $\mathrm{VdM}(\boldsymbol{\kappa};\mathcal{I}\setminus\mathfrak{v})=\mathrm{VdM}(\boldsymbol{\kappa_{1}};w(\mathcal{I})\setminus\mathfrak{v}_{1})$
for all $\mathcal{I}\in\mathfrak{G}$, so the equalities (\ref{eq: determinantal choice of signs})
imply $|\Delta_{\boldsymbol{a}}(\mathcal{I}\setminus\mathfrak{v})|=\frac{C_{1}}{C}|\Delta_{\boldsymbol{a_{1}}}(\mathcal{I}\setminus\mathfrak{v})|$.
$\square$ 
\end{proof}
\begin{thm}
\label{thm: compatible choices of signs, determinant} Assuming that
the parameters $\kappa_{1},\dots,\kappa_{n}$ are pairwise distinct,
a choice of signs is determinantal if and only if it is induced by
a choice of sign for the rows and the columns of $\boldsymbol{A}$
(up to the action of $GL_{k}(\mathbb{R})$). 
\end{thm}
\begin{proof}
We fix the gauge given by the $GL_{k}(\mathbb{R})$-action setting 
$\frac{C_{1}\cdot\lambda_{1}(\boldsymbol{x})}{C\cdot\lambda(\boldsymbol{x})}=1$.
By Remark \ref{rem: model reduction} and Lemma \ref{lem: determinantal choice of signs preserves soliton parameters},
a signature is determinantal only if the absolute values of maximal
minors of $\boldsymbol{a}$ are preserved. If $\mathfrak{v}\neq\emptyset$,
then the overall $\mathrm{sign}(C)$ in (\ref{eq: induced signature on reduced model})
can be expressed as a choice of sign for a row in $R(\mathfrak{v})$,
and the study is reduced to $[k]\setminus R(\mathfrak{v})$
and $[n]\setminus\mathfrak{v}$ through the map (\ref{eq: reductions, subsets without core}).
Therefore, to simplify the notation, we can focus on the case $\boldsymbol{a}\in\mathbb{R}_{\star}^{k\times n}$
with pivot set $\mathcal{V}$ without loss of generality.

The equality 
\begin{equation}
|a_{i\beta}|=\left|\frac{\Delta_{\boldsymbol{a}}\left(\mathcal{V}_{\beta}^{\nu_{i}}\right)}{\Delta_{\boldsymbol{a}}(\mathcal{V})}\right|=\left|\frac{\Delta_{\tilde{\boldsymbol{a}}}\left(\mathcal{V}_{\beta}^{\nu_{i}}\right)}{\Delta_{\tilde{\boldsymbol{a}}}(\mathcal{V})}\right|=|\tilde{a}_{i\beta}|,\quad i\in[k],\beta\in[n]\backslash\mathcal{V}\label{eq: fixd absolute values entries}
\end{equation}
means that the absolute values of the entries of $\boldsymbol{a}$
are fixed too. So, the transformation (\ref{eq: choice of signs})
is induced by a choice of signs 
\begin{equation}
\sigma:\,\left\{ (i,\alpha)\in[k]\times[n]:\,a_{i\alpha}\neq0\right\} \longrightarrow\{\pm1\}\label{eq: associated choice of signs entries RREF}
\end{equation}
for the non-vanishing entries of $\boldsymbol{a}$, which will be
denoted by $\sigma(\boldsymbol{a})$ as well.

We now construct a sequence of operations to label the rows and columns
of $\sigma(\boldsymbol{a})$ with signs $\varrho:\,[k]\longrightarrow\{\pm1\}$
and $\chi:\,[n]\longrightarrow\{\pm1\}$ respectively.
Let $\mathcal{G}:=([k],E)$
be a graph whose vertices label the rows of $\boldsymbol{a}$. The
pair $(i,j)$ is an edge if and only if there
exists $\gamma\in[n]$ such that $a_{i\gamma}\neq0\neq a_{j\gamma}$, i.e. 
\begin{equation} 
(i,j)\in E\Leftrightarrow\exists\gamma\in[n]:\,\Delta_{\boldsymbol{a}}\left(\mathcal{V}^{\nu_{i}}_{\gamma}\right)\neq0\neq\Delta_{\boldsymbol{a}}\left(\mathcal{V}^{\nu_{j}}_{\gamma}\right). 
\label{eq: edges rows}
\end{equation}
So fix an arbitrary row $h_{1}\in[k]$ of $\boldsymbol{a}$, e.g.
$h_{1}=1$. Without loss of generality, we can set $\varrho(h_{1}):=+1$.
Define $\mathfrak{c}_{1}:=\left\{ \alpha\in[n]:\,a_{h_{1}\alpha}\neq0\right\} $
and $\chi(\alpha):=\sigma(h_{1},\alpha)$ for all $\alpha\in\mathfrak{c}_{1}$.
Note that, for each $i\in[k]$, all the products $\sigma(h_{1},\alpha)\cdot\sigma(i,\alpha)$,
$\alpha\in\mathfrak{c}_{1}$ and $a_{i\alpha}\neq0$, coincide: indeed,
this is trivially true at $i=h_{1}$. If $i\neq h_{1}$ and $(i,\alpha)$,
$(i,\beta)$ are such that $\alpha,\beta\in\mathfrak{c}_{1}$, then
$0\notin\{a_{h_{1}\alpha},a_{h_{1}\beta}\}$ by definition. Hence,
at $a_{i\alpha}\neq0\neq a_{i\beta}$, the constraint $|\Delta_{\sigma(\boldsymbol{a})}(\mathcal{V}\setminus\{\nu_{h_{1}},\nu_{i}\}\cup\{\alpha,\beta\}\})|=|\Delta_{\boldsymbol{a}}(\mathcal{V}\setminus\{\nu_{h_{1}},\nu_{i}\}\cup\{\alpha,\beta\}\})|$
is equivalent to 
\begin{eqnarray}
& & \left|\sigma(h_{1},\alpha)\cdot\sigma(i,\beta)\cdot a_{h_{1}\alpha}\cdot a_{i\beta}-\sigma(i,\alpha)\cdot\sigma(h_{1},\beta)\cdot a_{i\alpha}\cdot a_{h_{1}\beta}\right| \nonumber \\ 
& = & \left|a_{h_{1}\alpha}\cdot a_{i\beta}-a_{i\alpha}\cdot a_{h_{1}\beta}\right|.\label{eq: equality absolute values, step 1}
\end{eqnarray}
From $a_{h_{1}\alpha}\cdot a_{i\beta}\cdot a_{i\alpha}\cdot a_{h_{1}\beta}\neq0$,
one gets $|a_{h_{1}\alpha}\cdot a_{i\beta}-a_{i\alpha}\cdot a_{h_{1}\beta}|\neq|a_{h_{1}\alpha}\cdot a_{i\beta}+a_{i\alpha}\cdot a_{h_{1}\beta}|$,
thus 
\begin{eqnarray}
\hspace*{-1.5cm}
\sigma(h_{1},\alpha)\cdot\sigma(i,\beta)=\sigma(i,\alpha)\cdot\sigma(h_{1},\beta) & \Leftrightarrow & \sigma(h_{1},\alpha)\cdot\sigma(h_{1},\beta)\cdot\sigma(h_{1},\alpha)\cdot\sigma(i,\beta)\nonumber \\
& = & \sigma(h_{1},\alpha)\cdot\sigma(h_{1},\beta)\cdot\sigma(i,\alpha)\cdot\sigma(h_{1},\beta)\nonumber \\
& \Leftrightarrow & \sigma(h_{1},\beta)\cdot\sigma(i,\beta)\nonumber \\ 
& = & \sigma(h_{1},\alpha)\cdot\sigma(i,\alpha).\label{eq: consistent row sign, step 1}
\end{eqnarray}
So, for any $(h_{1},i)\in E$ , the $i$th row can be labelled with
a sign $\varrho(i):=\sigma(i,\alpha)\cdot\chi(\alpha)$ for some $\alpha\in\mathfrak{c}_{1}$
assuming that $\mathfrak{c}_{1}\neq\emptyset$. Then let $\mathfrak{r}_{1}:=\{i\in[k]\setminus\{h_{1}\}:\,(i,h_{1})\in E\}$
and $h_{2}:=\min\mathfrak{r}_{1}$. By the previous observation, one
has a definite sign $\varrho(h_{2})$. Let $\mathfrak{c}_{2}:=\left\{ \alpha\in[n]\setminus\mathfrak{c}_{1}:\,a_{h_{2}\alpha}\neq0\right\} $
and assign $\chi(\alpha):=\sigma(h_{2},\alpha)\cdot\varrho(h_{2})$
for any $\alpha\in\mathfrak{c}_{2}$. As before, for each fixed $i\in[k],$
all the signs $\sigma(h_{2},\alpha)\cdot\sigma(i,\alpha)$, $\alpha\in\mathfrak{c}_{2}$
and $a_{i\alpha}\neq0$, coincide. So the sign $\varrho_{2}(i):=\sigma(i,\alpha)\cdot\chi(\alpha)$,
for any $\alpha\in\mathfrak{c}_{2}$, is well-defined. It may exist
$g\in[k]\setminus\{h_{1},h_{2}\}$ such that $a_{g\alpha}\neq0$ and
$a_{g\beta}\neq0$ for some $\alpha\in\mathfrak{c}_{1}$, $\beta\in\mathfrak{c}_{2}$.
We can check that $\varrho_{2}(g)=\varrho(g)$: in fact, the signs
restricted to rows $h_{1},h_{2},g$ and columns $\alpha,\beta$ can
be depicted as 
\begin{equation}
\begin{array}{ccc}
 & \mbox{column }\alpha & \mbox{column }\beta\\
\mbox{row }h_{1} & \varrho(h_{1})\cdot\chi(\alpha) & 0\\
\mbox{row }h_{2} & W & \varrho(h_{2})\cdot\chi(\beta)\\
\mbox{row }g & \varrho(g)\cdot\chi(\alpha) & \varrho_{2}(g)\cdot\chi(\beta)
\end{array}\label{eq: configuration consistence row signs, iterations}
\end{equation}
If $W\neq0$, then it equals $\varrho(h_{2})\cdot\chi(\alpha)$ by
definition. Thus, the constraint $|\Delta_{\boldsymbol{a}}(\mathcal{I})|=|\Delta_{\sigma(\boldsymbol{a})}(\mathcal{I})|$
at $\mathcal{I}=\mathcal{V}\setminus\{\nu_{h_{2}},\nu_{g}\}\cup\{\alpha,\beta\}$
implies that 
\begin{equation}
\varrho(h_{2})\cdot\chi(\alpha)\cdot\varrho_{2}(g)\cdot\chi(\beta)=\varrho(h_{2})\cdot\chi(\beta)\cdot\varrho(g)\cdot\chi(\alpha)\label{eq: equality absolute values, a}
\end{equation}
that is $\varrho_{2}(g)=\varrho(g)$. Now assume that $W=0$. Since
$(h_{2},h_{1})\in E$, there exists $\gamma\in[n]$ such that $a_{h_{1}\gamma}\neq0\neq a_{h_{2}\gamma}$,
and $\gamma$ is clearly different from $\alpha$ and $\beta$ since
$a_{h_{1}\beta}=0=a_{h_{2}\alpha}$, then consider the extended scheme
\begin{equation}
\begin{array}{cccc}
 & \mbox{column }\alpha & \mbox{column }\beta & \mbox{column }\gamma\\
\mbox{row }h_{1} & \varrho(h_{1})\cdot\chi(\alpha) & 0 & \varrho(h_{1})\cdot\chi(\gamma)\\
\mbox{row }h_{2} & 0 & \varrho(h_{2})\cdot\chi(\beta) & \varrho(h_{2})\cdot\chi(\gamma)\\
\mbox{row }g & \varrho(g)\cdot\chi(\alpha) & \varrho_{2}(g)\cdot\chi(\beta) & X
\end{array}\label{eq: extended configuration consistence row signs, iterations}
\end{equation}
If $X\neq0$ (hence $|a_{g\gamma}|\neq0$), then $X=\varrho(g)\cdot\chi(\gamma)$ and, applying $|\Delta_{\boldsymbol{a}}(\mathcal{I})|=|\Delta_{\sigma(\boldsymbol{a})}(\mathcal{I})|$ at $\mathcal{I}=\mathcal{V}\setminus\{\nu_{h_{2}},\nu_{g}\}\cup\{\beta,\gamma\}$, one gets $\varrho_{2}(g)=\varrho(g)$. Also note that this can be found without evaluating $X$ from the requirement $|\Delta_{\boldsymbol{a}}(\mathcal{I})|=|\Delta_{\sigma(\boldsymbol{a})}(\mathcal{I})|$ at both $\mathcal{I}=\mathcal{V}\setminus\{\nu_{h_{1}},\nu_{g}\}\cup\{\alpha,\gamma\}$ 
and $\mathcal{I}=\mathcal{V}\setminus\{\nu_{h_{2}},\nu_{g}\}\cup\{\beta,\gamma\}$ simultaneously, i.e.  
\begin{eqnarray}
+1 & = & \varrho(h_{1})\cdot\chi(\alpha)\cdot X\cdot\varrho(h_{1})\cdot\chi(\gamma)\cdot\varrho(g)\cdot\chi(\alpha)\nonumber \\
& = & \varrho(h_{2})\cdot\chi(\beta)\cdot X\cdot\varrho(h_{2})\cdot\chi(\gamma)\cdot\varrho_{2}(g)\cdot\chi(\beta)\label{eq: equality absolute values, iteration a}
\end{eqnarray}
that gives 
\begin{equation}
\varrho(h_{1})\cdot\chi(\alpha)\cdot\varrho(h_{1})\cdot\chi(\gamma)\cdot\varrho(g)\cdot\chi(\alpha)=\varrho(h_{2})\cdot\chi(\beta)\cdot\varrho(h_{2})\cdot\chi(\gamma)\cdot\varrho_{2}(g)\cdot\chi(\beta).\label{eq: compatibility three, a}
\end{equation}
If instead $X=0$, considering the constraint $|\Delta_{\boldsymbol{a}}(\mathcal{I})|=|\Delta_{\sigma(\boldsymbol{a})}(\mathcal{I})|$
in the case $\mathcal{I}=\mathcal{V}\setminus\{\nu_{h_{1}},\nu_{h_{2}},\nu_{g}\}\cup\{\alpha,\beta,\gamma\}$,
we get 
\begin{eqnarray}
& & |-a_{h_{1}\gamma}\cdot a_{h_{2}\beta}\cdot a_{g\alpha}-a_{h_{2}\gamma}\cdot a_{g\beta}\cdot a_{h_{1}\alpha}|\nonumber \\
& = & |\varrho(h_{1})\cdot\chi(\gamma)\cdot\varrho(h_{2})\cdot\chi(\beta)\cdot\varrho(g)\cdot\chi(\alpha)\cdot a_{h_{1}\gamma}\cdot a_{h_{2}\beta}\cdot a_{g\alpha}\nonumber \\
& + & \varrho(h_{2})\cdot\chi(\gamma)\cdot\varrho_{2}(g)\cdot\chi(\beta)\cdot\varrho(h_{1})\cdot\chi(\alpha)\cdot a_{h_{2}\gamma}\cdot a_{g\beta}\cdot a_{h_{1}\alpha}|\label{eq: equality absolute values, iteration b}
\end{eqnarray}
and, from $0\neq a_{h_{1}\gamma}\cdot a_{h_{2}\beta}\cdot a_{g\alpha}\cdot a_{h_{2}\gamma}\cdot a_{g\beta}\cdot a_{h_{1}\alpha}$,
we obtain 
\begin{equation}
\varrho(h_{1})\cdot\chi(\gamma)\cdot\varrho(h_{2})\cdot\chi(\beta)\cdot\varrho(g)\cdot\chi(\alpha)=\varrho(h_{2})\cdot\chi(\gamma)\cdot\varrho_{2}(g)\cdot\chi(\beta)\cdot\varrho(h_{1})\cdot\chi(\alpha).\label{eq: compatibility three, b}
\end{equation}
Since all the factors in (\ref{eq: compatibility three, a}) and (\ref{eq: compatibility three, b})
belong to $\{\pm1\}$ and, in particular, they are non-vanishing and
idempotent, these expressions simplify as 
\begin{equation}
\varrho_{2}(g)=\varrho(g).\label{eq: compatibility row signs}
\end{equation}
and the sign $\varrho_{2}(g)\cdot\chi(\gamma)=\varrho(g)\cdot\chi(\gamma)$
can be consistently fixed. This construction can be extended through the iterations 
\begin{eqnarray}
\mathfrak{r}_{r-1} & := & \left\{ i\in[k]\setminus\{h_{1},\dots,h_{r-1}\}:\,(i,h_{r-1})\in E\right\} ,\nonumber \\
h_{r} & := & \min\mathfrak{r}_{r-1},\nonumber \\ 
\mathfrak{c}_{r} & := & \left\{ \alpha\in[n]\setminus\bigcup_{t=1}^{r-1}\mathfrak{c}_{t}:\,a_{h_{r}\alpha}\neq0\right\} \label{eq: iteration steps}
\end{eqnarray}
and the signs $\varrho_{r-1}(h_{r})$ and $\chi(\alpha):=\sigma(h_{r},\alpha)\cdot\varrho_{r-1}(h_{r})$
for any $\alpha\in\mathfrak{c}_{r}$. If the signs $\varrho_{s}$,
as long as $\chi$, are uniquely defined at steps $s<r$, then one
can check the compatibility of $\varrho_{r}$ with $\varrho$ too,
formally $\varrho_{r}(g)=\varrho(g)$ for all rows $g$ connected
to both $h_{r}$ and $h_{s}$, $s<r$. At this purpose, it is worth
introducing \textit{paths} on $\boldsymbol{a}$, that are defined
as finite sequences of non-vanishing elements of $\boldsymbol{a}$
connected by alternate moves along rows and columns, i.e. chains of
the type (\ref{eq: path}). The process that links the signs $\varrho_{r}(g)$,
derived from $h_{r}$, and $\varrho(g)$, obtained in a previous step
$h_{s}$, can be represented by a path $\Phi$ starting with $(g,\alpha_{1})\rightarrow(h_{i_{2}},\alpha_{1})\rightarrow\dots$,
passing through some signed rows $\Phi_{\varrho}:=\{g,h_{i_{2}},\dots,h_{i_{L}}\}$
and columns $\Phi_{\chi}:=\{\alpha_{1},\dots,\alpha_{L}\}$,
and ending with $\dots\rightarrow(h_{i_{L}},\alpha_{L})\rightarrow(g,\alpha_{L})$,
where $\alpha_{1}=\alpha_{s}\in\mathfrak{c}_{s}$, $\alpha_{L}=\alpha_{r}\in\mathfrak{c}_{r}$
and $a_{g\alpha_{s}}\neq0\neq a_{g\alpha_{r}}$. 

Note that the column of $\boldsymbol{a}$ associated with any $\alpha_{T}\in\Phi_{\chi}$ has at least two non-vanishing entries $(h_{i_{T}},\alpha_{T})$ and $(h_{i_{T+1}},\alpha_{T})$, while there is only one non-vanishing entry in pivot columns by definition, hence $\mathcal{V}\cap\Phi_{\chi}=\emptyset$. Suppose that there
is a row $\bar{h}\in\Phi_{\varrho}$ with $0\notin\{a_{\bar{h}\alpha_{T-1}},a_{\bar{h}\alpha_{T}},a_{\bar{h}\alpha_{X}}\}$,
$X\notin\{T-1,T\}$, and first assume that $T<X$: since the column
$\alpha_{X}$ is reached at $(h_{i_{X}},\alpha_{X})$, we can substitute
the chain $(\bar{h},\alpha_{T-1})\rightarrow(\bar{h},\alpha_{T})\rightarrow\dots\rightarrow(h_{i_{X}},\alpha_{X})\rightarrow(h_{i_{X+1}},\alpha_{X})$
with $(\bar{h},\alpha_{T-1})\rightarrow(\bar{h},\alpha_{X})\rightarrow(h_{i_{X+1}},\alpha_{X})$
in the path. At $X<T-1$, likewise, we can change $(h_{i_{X}},\alpha_{X})\rightarrow(h_{i_{X+1}},\alpha_{X})\rightarrow\dots\rightarrow(\bar{h},\alpha_{T-1})\rightarrow(\bar{h},\alpha_{T})$
with $(h_{i_{X}},\alpha_{X})\rightarrow(\bar{h},\alpha_{X})\rightarrow(\bar{h},\alpha_{T})$.
In both cases, the result is a path connecting $\alpha_{s}$ to $\alpha_{r}$
with shorter length, which can still be used to compare $\varrho_{r}(g)$
and $\varrho(g)$. Similar substitutions can be carried out for columns $\bar{c}\in\Phi_{\chi}$ that appear as components of more than two elements of the path. So we focus on paths of minimal length: the previous
construction shows that the submatrix extracted from $\boldsymbol{a}$
selecting the rows in $\Phi_{\varrho}$ and columns in $\Phi_{\chi}$
associated with a minimal path $\Phi$ from $(g,\alpha_{s})$ to $(g,\alpha_{r})$ has exactly two non-vanishing
elements per row and column. Hence, the condition $|\Delta_{\boldsymbol{a}}(\mathcal{V}\setminus\{\nu_{i}:\,i\in\Phi_{\varrho}\}\cup\Phi_{\chi})|=|\Delta_{\sigma(\boldsymbol{a})}(\mathcal{V}\setminus\{\nu_{i}:\,i\in\Phi_{\varrho}\}\cup\Phi_{\chi})|$
can be expressed as in (\ref{eq: compatibility three, b}), that is
\begin{equation}
\varrho(g)\cdot\chi(\alpha_{1})\cdot\left(\prod_{T=2}^{L}\varrho(h_{i_{T}})\cdot\chi(\alpha_{T})\right)=\left(\prod_{T=1}^{L-1}\varrho(h_{i_{T+1}})\cdot\chi(\alpha_{T})\right)\cdot\varrho_{2}(g)\cdot\chi(\alpha_{L})
\label{eq: product +1 on a closed path}
\end{equation}
which gives $\varrho_{r}(g)=\varrho(g)$.  

These steps can be repeated while $\mathfrak{r}_{r}\neq\emptyset\neq\mathfrak{c}_{r}$:
in this process we start from a node $h_{1}$ of the graph $\mathcal{G}$
and follow a path of adjacent vertices. If $\mathfrak{r}_{r}=\emptyset$,
then we can follow this path backwards until we reach $\mathfrak{r}_{u}$,
$u\in[r-1]$, such that there exists $\tilde{h}_{u}\in\mathfrak{r}_{u}$
with $\tilde{h}_{u}\neq h_{w}$, $w\in[r-1]$. Then, these operations
can be repeated along another path of adjacent nodes starting from
$\tilde{h}_{u}$. At each stage an additional sign is selected compatibly with the previous
ones. Such a process explores each node in the connected
component of $\mathcal{G}$ containing $h_{1}$ exactly once. Repeating these steps for all the connected components of $\mathcal{G}$
we give a sign to all the rows of $\boldsymbol{a}$. Since each column
$\alpha$ contains at least one non-vanishing element $a_{h_{r}\alpha}\neq0$
and all the rows are visited, the index $\alpha$ belongs to $\mathfrak{c}_{r}$
at certain step $r$. Hence, all the columns are labelled by a sign
as well. By construction, one has $\sigma(i,\alpha)=\varrho(i)\cdot\chi(\alpha)$
for all $(i,\alpha)$ such that $a_{i\alpha}\neq0$. Clearly, this assignment produces a determinantal choice of signs. 
\end{proof}
\begin{rem}
\label{rem: signs for arrays from deformation} The proof of Theorem
\ref{thm: compatible choices of signs, determinant} is constructive:
it generates one of the possible sign configurations that induce a
given choice $\Sigma$. This configuration is not unique, e.g. switching
the sign of two distinct rows of $\boldsymbol{A}$ or $\boldsymbol{a}$
returns the same $\Sigma$. The uniqueness of the previous construction
follows from the choice of a (arbitrary) $\sigma$ in (\ref{eq: associated choice of signs entries RREF})
for the entries of the reduced row echelon form $\boldsymbol{a}$
and signs of distinguished nodes (such as $h_{1}$) in the connected
components of $\mathcal{G}$. Furthermore, the algorithm attributes a sign
to vanishing entries of $\boldsymbol{a}$ too. The labelling $+0$ or
$-0$ can be thought as the sign of the associated entry in a perturbed
matrix $\boldsymbol{b}$ such that $\mathrm{sign}(\Delta_{\boldsymbol{a}}(\mathcal{I}))=\mathrm{sign}(\Delta_{\boldsymbol{b}}(\mathcal{I}))$
for all $\mathcal{I}\in\mathcal{P}_{k}[n]$ with $\Delta_{\boldsymbol{a}}(\mathcal{I})\neq0$.
These points are relevant when $\boldsymbol{a}$ has many vanishing
entries. 
\end{rem}

\section{\label{sec: Tropical limit} Remarks on the tropical limit} 

With regard to the issues dealt with in the present work, tropical
methods have proved useful in the study of algebraic amoebas \cite{Maeda2007,Passare2012}
and the analysis of KdV and KP II soliton solutions and their singularities
\cite{DM-H2011,DM-H2014,KodamaWilliams2013}. Here we briefly discuss
how the models we have described give a concrete realization to the
concepts discussed in \cite{Angelelli2017}, especially the role
of order and enumeration in the tropical limit in statistical physics. Relevant
quantities, like the free energy $-k_{B}T\cdot\ln\mathcal{Z}$ associated
with (\ref{eq: standard partition function}), fit naturally in this
process, and their tropical limit is suitable for the description
of phenomena like exponential degenerations of energy levels and limiting
temperatures, see \cite{AK2015}. 

Determinantal partition functions (\ref{eq: Cauchy-Binet formula, solitons, a})
give a concrete realization of the \textit{nested} tropical limit
discussed in \cite{Angelelli2017}. For the sake of concreteness, let us consider the case when $g_{\mathcal{I}}\geq0$ for all $\mathcal{I}\in\mathcal{P}_{k}[n]$. Outside the locus of points $\boldsymbol{x}\in\mathbb{R}^{d}$
where $\varphi_{\alpha}(\boldsymbol{x})=\varphi_{\beta}(\boldsymbol{x})$
for some $\alpha\neq\beta$, one has $\varphi_{\pi(1)}(\boldsymbol{x})>\dots>\varphi_{\pi(n)}(\boldsymbol{x})$
for a permutation $\pi\in\mathcal{S}_{n}$. Here, for large values
of $\boldsymbol{x}$, there exists exactly one dominant term $\Lambda_{\mathcal{D}}(\boldsymbol{x})>\Lambda_{\mathcal{H}}(\boldsymbol{x})$,
$\mathcal{H}\in\mathcal{P}_{k}[n]\setminus\{\mathcal{D}\}$, as follows
from the polynomial form of $\varphi_{\alpha}(\boldsymbol{x})$. In
particular, $\mathcal{D}$ is the least set in $\mathcal{P}_{k}[n]$
associated with a non-vanishing $g_{\mathcal{D}}$ with respect to
the lexicographical order induced by $(\pi(1),\dots,\pi(n))$. Hence
one has ${\displaystyle \Lambda_{\mathcal{H}}(\boldsymbol{x})}\ll\Lambda_{\mathcal{D}}(\boldsymbol{x})$
for all $\mathcal{H}\neq\mathcal{D}$ and $\tau_{\mathcal{I}}(\boldsymbol{x})<0$
if and only $\mathcal{D}\parallel\mathcal{I}$. Moreover, if $g_{\mathcal{I}}>0 for all \mathcal{I}\in\mathcal{P}_{k}[n]$, then the number of such subsets $\mathcal{I}$ with odd intersection with $\mathcal{D}$ is $\Omega(n,s;k)$, that is the dual quantity of $\Omega(n,k;s)$ as
in (\ref{eq: duality s<->k}) and (\ref{eq: +/- amoeba duality}).
The nested form (see \cite{Angelelli2017}, \textsection 6) comes from the identification
of dominant term $\mathcal{D}$: if one sets $\mathfrak{D}(r):=\left\{ \mathcal{H}\in\mathcal{P}_{k}[n]:\,\max\{i\in[k]\cup\{0\}:\,\{\pi(1),\dots,\pi(i)\}\subseteq\mathcal{H}\cap\mathcal{D}\}=r\right\} $,
then $\tau$ can be expressed as 
\begin{eqnarray}
\hspace*{-2cm}
& & \tau(\boldsymbol{x})\nonumber \\ 
& = & e^{\varphi_{\pi(1)}(\boldsymbol{x})}\left(e^{\varphi_{\pi(2)}(\boldsymbol{x})}\left(\dots e^{\varphi_{\pi(k-1)}(\boldsymbol{x})}\left(e^{\varphi_{\pi(k)}(\boldsymbol{x})}\left(g_{\mathcal{D}}+\sum_{\mathcal{I}\in\mathfrak{D}(k-1)}g_{\mathcal{I}}(\boldsymbol{x}) e^{-\varphi_{\pi(k)}(\boldsymbol{x})}\prod_{\alpha\in\mathcal{I}\setminus\pi([k-1])}e^{\varphi_{\alpha}(\boldsymbol{x})}\right)\right.\right.\right.\nonumber \\
\hspace*{-2.5cm} & + & \left.\left.\sum_{\mathcal{I}\in\mathfrak{D}(k-2)}g_{\mathcal{I}}(\boldsymbol{x}) e^{-\varphi_{\pi(k-1)}(\boldsymbol{x})}\prod_{\alpha\in\mathcal{I}\setminus\pi([k-2])}e^{\varphi_{\alpha}(\boldsymbol{x})}\right)+\dots\sum_{\mathcal{I}\in\mathfrak{D}(1)}g_{\mathcal{I}}(\boldsymbol{x})e^{-\varphi_{\pi(2)}(\boldsymbol{x})}\prod_{\alpha\in\mathcal{I}\setminus\{\pi(1)\}}e^{\varphi_{\alpha}(\boldsymbol{x})}\right)\nonumber \\
\hspace*{-2.5cm} & + & \left.\sum_{\mathcal{I}\in\mathfrak{D}(0)}g_{\mathcal{I}}(\boldsymbol{x}) e^{-\varphi_{\pi(1)}(\boldsymbol{x})}\prod_{\alpha\in\mathcal{I}}e^{\varphi_{\alpha}(\boldsymbol{x})}\right).\label{eq: nested form tau function}
\end{eqnarray}
In such a way, one can recursively find the first $k$ dominant terms,
thus realising a refinement of the tropical limit of the exponents
$\varphi_{\alpha}$ through the limit for the collections $\sum_{\alpha\in\mathcal{I}}\varphi_{\alpha}$,
$\mathcal{I}\in\mathcal{P}_{k}[n]$.

The tropical limit drastically affects the enumeration process too,
mainly due to idempotence. This property can be seen as a symmetry
of the statistical system with respect to creation of copies of its
constituents: if the copying process involves all the constituents
at once, e.g. all the distinct terms in the partition function, the
symmetry is global, otherwise it is local (for more details, see \cite{Angelelli2017},
\textsection 8). Matrix models also provide a realization of such a symmetry:
for each $\alpha\in[n]$, one can introduce a ``copy'' the $\alpha$th
soliton as a copy of the associated row/column, formally 
\begin{eqnarray}
\boldsymbol{A} & \mapsto & \mathrm{T}_{\alpha}(\boldsymbol{A}):=\boldsymbol{A}\cdot\boldsymbol{T}_{\alpha},\nonumber \\
\boldsymbol{K} & \mapsto & \mathrm{T}_{\alpha}(\boldsymbol{K}):=\boldsymbol{T}_{\alpha}^{T}\cdot\boldsymbol{K},\nonumber \\
\boldsymbol{\Theta}(\boldsymbol{x}) & \mapsto & \mathrm{T}_{\alpha}(\boldsymbol{\Theta}(\boldsymbol{x})):=\boldsymbol{\pi}_{\mathrm{cyc},\alpha}^{-1}\cdot\left(\boldsymbol{\Theta}(\boldsymbol{x})\oplus\left(e^{\varphi_{\alpha}(\boldsymbol{x})}\right)\right)\cdot\boldsymbol{\pi}_{\mathrm{cyc},\alpha} \label{eq: copy creation tau function}
\end{eqnarray}
where $\boldsymbol{\pi}_{\mathrm{cyc},\alpha}$ is the matrix associated with the cyclic permutation $(\alpha+1\,\alpha+2\,\dots\,n+1)$ and
\begin{equation}
(\boldsymbol{T}_{\alpha})_{\gamma\beta}:=\sum_{\omega=1}^{\alpha}\delta_{\gamma,\beta}\cdot\delta_{\gamma,\omega}+\sum_{\omega=\alpha}^{n}\delta_{\gamma,\beta-1}\cdot\delta_{\gamma,\omega},\quad\alpha,\gamma\in[n],\,\beta\in[n+1].\label{eq: copying matrix}
\end{equation}
If the same number of copies, say $\ell$, are created for all $\alpha\in[n]$,
then the resulting $\tau$-function gets a multiplicative factor $(\ell+1)^{k}$
which disappears in the derivation (\ref{eq: solution from tau}),
so the soliton solution is preserved. Thus, there is an intrinsic
tropical global symmetry in this particular class of solutions, which
is also compatible with the $GL_{k}(\mathbb{R})$-action defining
the Grassmannian. Further, the effect of the tropical copy of a single
soliton $\kappa_{\alpha_{0}}$ on the singular locus defined by $\det\left(\boldsymbol{A}\cdot\boldsymbol{\Theta}(\boldsymbol{x})\cdot\boldsymbol{K}\right)=0$
is the same as its erasure, since 
\begin{eqnarray}
& & \det\left(\mathrm{T}_{\alpha_{0}}^{(\ell)}(\boldsymbol{A})\cdot\mathrm{T}_{\alpha_{0}}^{(\ell)}(\boldsymbol{\Theta}(\boldsymbol{x}))\cdot\mathrm{T}_{\alpha_{0}}^{(\ell)}(\boldsymbol{K})\right)\nonumber \\ 
& = & \ell\cdot\sum_{\alpha_{0}\in\mathcal{I}}^{\mathcal{I}\in\mathcal{P}_{k}([N])}\Delta_{\mathcal{I}}(\boldsymbol{A})\Delta_{\mathcal{I}}(\boldsymbol{K})\exp\left(\sum_{\beta\in\mathcal{I}}\varphi_{\beta}(\boldsymbol{x})\right)\nonumber \\ 
& = & -\ell\cdot\sum_{\alpha_{0}\notin\mathcal{I}}^{\mathcal{I}\in\mathcal{P}_{k}([N])}\Delta_{\mathcal{I}}(\boldsymbol{A})\Delta_{\mathcal{I}}(\boldsymbol{K})\exp\left(\sum_{\beta\in\mathcal{I}}\varphi_{\beta}(\boldsymbol{x})\right). \nonumber \\ \label{eq: local symmetry tau function singular locus}
\end{eqnarray}
This process is independent on the choice of $\ell>1$, since the
factor $\ell$ in (\ref{eq: local symmetry tau function singular locus})
disappears in the derivation (\ref{eq: solution from tau}). In this
sense, the invariance under tropical copies is local on the singular
locus. 

\end{document}